\documentclass[11pt, a4paper, journal,onecolumn]{IEEEtran}

\usepackage{amsfonts}                                                        
\usepackage{epsfig}
\usepackage{amsthm}
\usepackage{graphicx}        
\usepackage{latexsym}
\usepackage{amssymb}
\usepackage{amsmath}
\usepackage{parskip}
\usepackage{multirow}
\usepackage{cite}
\usepackage{mathtools}
\usepackage{epstopdf}
\usepackage{etoolbox}
\usepackage{tikz}
\usetikzlibrary{shapes,arrows}
\usepackage{array}

\usepackage{setspace}
\usepackage{mathptmx}
\usepackage[bottom]{footmisc}
\usepackage{tikz}
\usepackage{algorithm,algorithmic}
\usepackage{caption}
\usepackage{subcaption}

\usetikzlibrary{decorations.pathreplacing}
\usepackage[top=0.9in, bottom=0.9in, left=0.9in, right=0.9in]{geometry}
\usepackage[colorlinks=true]{hyperref}

\usepackage{verbatim}
\usepackage{calrsfs}

\DeclareMathAlphabet{\pazocal}{OMS}{zplm}{m}{n}

\newcommand{\mb}{\mathbb}

\let\bbordermatrix\bordermatrix
\patchcmd{\bbordermatrix}{8.75}{4.75}{}{}
\patchcmd{\bbordermatrix}{\left(}{\left[}{}{}
\patchcmd{\bbordermatrix}{\right)}{\right]}{}{}

\newcommand{\sr}{\stackrel}

\newcommand{\rar}{\rightarrow}

\newcommand{\tri}{\sr{\triangle}{=}}

\newcommand{\be}{\begin{equation}}
\newcommand{\ee}{\end{equation}}
\newcommand{\bea}{\begin{eqnarray}}
\newcommand{\eea}{\end{eqnarray}}
\newcommand{\bes}{\begin{eqnarray*}}
\newcommand{\ees}{\end{eqnarray*}}
\newcommand{\bce}{\begin{center}}
\newcommand{\ece}{\end{center}}
\newcommand{\beae}{\begin{IEEEeqnarray}{rCl}}
\newcommand{\eeae}{\end{IEEEeqnarray}}
\newcommand{\nms}{\IEEEeqnarraynumspace}
\def\VR{\kern-\arraycolsep\strut\vrule &\kern-\arraycolsep}
\def\vr{\kern-\arraycolsep & \kern-\arraycolsep}

\newcommand{\ben}{\begin{enumerate}}
\newcommand{\een}{\end{enumerate}}
\newcommand{\argmax}{\arg\!\max}

\newcommand{\hso}{\hspace{.1in}}

\newcommand{\hst}{\hspace{.2in}}

\newcommand{\noi}{\noindent}

\newtheorem{theorem}{Theorem}[section]

\newtheorem{remark}{Remark}[section]
\newtheorem{corollary}{Corollary}[section]

\newtheorem{assumptions}{Assumptions}[section]
\newtheorem{definition}{Definition}[section]
\newtheorem{lemma}{Lemma}[section]

\begin{document}
%
\title{Single Letter Expression of Capacity  for a Class of Channels with Memory}

\author{\IEEEauthorblockN{Christos K. Kourtellaris, Charalambos~D.~Charalambous and Ioannis Tzortzis }
\thanks{C~K.~Kourtellaris, C~D.~Charalambous and I.~ Tzortzis are with the Department of Electrical and Computer Engineering, University of Cyprus, Nicosia, Cyprus, {\it Email:kourtellaris.christos@ucy.ac.cy, chadcha@ucy.ac.cy,  tzortzis.ioannis@ucy.ac.cy}
This work was financially supported by a medium size University of Cyprus grant entitled ``DIMITRIS"  and by QNRF, a member of Qatar Foundation, under the project NPRP 6-784-2-329}}
%



%

\maketitle

%
%
\begin{abstract}
\bf {We study  finite alphabet channels with Unit Memory on the previous Channel Outputs called UMCO channels. We identify necessary and sufficient conditions, to test whether the capacity achieving channel input distributions with feedback are time-invariant, and whether feedback capacity is  characterized by single letter, expressions, similar to that of memoryless channels. The method is based on showing that a certain dynamic programming equation, which in general, is a nested optimization problem over the sequence of channel input distributions, reduces to a non-nested optimization problem.  Moreover, for UMCO channels, we give a simple expression for the ML error exponent, and we  identify sufficient conditions to test whether feedback does not increase capacity. We derive similar results,  when transmission cost constraints are imposed.
We apply the results to a special class of the UMCO channels, the Binary State Symmetric Channel (BSSC) with and without transmission cost constraints, to show that the optimization problem of feedback capacity is non-nested, the capacity achieving channel input distribution and the corresponding channel output transition probability distribution are time-invariant, and   feedback capacity is characterized by  a single letter formulae, precisely as Shannon's single letter characterization of capacity of memoryless channels. Then we derive closed form  expressions for the capacity achieving channel input distribution and feedback capacity.  We use the closed form expressions to evaluate an error exponent for ML decoding.  
}
\end{abstract}

\section{Introduction}
\par   Shannon in his landmark paper \cite{shannon48}, showed that the  capacity of Discrete Memoryless Channels (DMCs) $\big\{{\mathbb A}, {\mathbb B},\{{\bf P}_{B|A}(b|a): (a,b)\in {\mathbb A}\times {\mathbb B}\} \big\}$ is characterized by the celebrated single letter formulae
\bea
C\tri \max_{{\bf P}_A} I(A;B). \label{cap_fb_c}
\eea
This is often shown by using the converse to the channel coding theorem, to obtain the  upper bounds \cite{cover-thomas2006} 
\bea
C_{A^n ; B^n}^{noFB}\tri \max_{{\bf P}_{A^n}} I(A^n;B^n) \leq \; \max_{{\bf P}_{A_i}, i=0, \ldots, n} \sum_{i=0}^n I(A_i;B_i)\leq (n+1) C \label{cap_nf_c1}
\eea
which are achievable, if and only if  the channel input distribution  satisfies conditional independence $
{\bf P}_{A_i|A^{i-1}}={\bf P}_{A_i},         i=0, 1, \ldots, n$, and   $\{A_i:i =0, 1, \ldots, \}$ is identically distributed, which  implies that the  joint process $\{(A_i, B_i): i=0,1, \ldots, \}$ is independent and identically distributed, and hence stationary ergodic.
 For DMCs, it is shown by Shannon \cite{shannon1956}  and Dobrushin \cite{dubrushin1958} that feedback codes do not incur a higher capacity compared to that of codes without feedback, that is, $C^{FB}=C$. This  is often shown by first applying  the converse to the coding theorem,  to deduce that feedback does not increase capacity \cite{cover-pombra1989}, that is, $C^{FB} \leq C^{noFB}=C$, which then implies that any candidate of optimal channel input distribution with feedback $\Big\{{\bf P}_{A_i|A^{i-1}, B^{i-1}}: i=0, \ldots, n\Big\}$ satisfies conditional independence
\begin{align}
 {\bf P}_{A_i|A^{i-1}, B^{i-1}}(da_i|a^{i-1}, b^{i-1})={\bf P}_{A_i}(da_i),        \hso  i=0, 1, \ldots, n \label{CI_DMC}
\end{align} 
   and hence identity $C^{FB}=C^{noFB}=C$ holds if  $\{A_i:i =0, 1, \ldots, \}$ is identically distributed. 
\par For general channels with memory defined by $\big\{{\bf P}_{B_i|B^{i-1}, A^i}: i=0, 1, \ldots, n\}$, ${\bf P}_{B_0|B^{-1}, A^0}={\bf P}_{B_0|B^{-1}, A_0}$, where  $B^{-1}$ is the initial state,  in general, feedback codes incur a higher capacity  than  codes  without feedback \cite{cover-thomas2006,ihara1993}. The information measure often employed to characterize feedback capacity of such channels is Marko's directed information \cite{marko1973}, put forward by Massey\cite{massey1990}, and  defined by 
\bea
I(A^n\rar B^n)=\sum_{i=0}^n I(A^i;B_i|B^{i-1})\tri \sum_{i=0}^n \int_{  }^{}   \log \Big( \frac{ {\bf P}_{B_i|B^{i-1}, A^i}(\cdot|b^{i-1}, a^i) }{{\bf P}_{B_i|B^{i-1}}(\cdot|b^{i-1})}(b_i)\Big){ \bf P}_{A^i, B^i}( da^i, db^i).
\eea
Indeed, Massey \cite{massey1990} showed that the per unit time limit of the supremum of directed information over channel input distributions ${\cal P}_{[0,n]}^{FB}\tri \big\{{\bf P}_{A_i|A^{i-1}, B^{i-1}}: i=0, \ldots, n\big\}$, defined by 
\bea
C^{FB}_{A^{\infty}\rar B^{\infty}}=\lim_{n\rar\infty}\frac{1}{n+1}C_{A^{n}\rar B^{n}} \hspace{1cm} C_{A^{n}\rar B^{n}}\tri\sup_{{\cal P}_{[0,n]}^{FB}}I(A^n\rar B^n) \label{cap_inf_sta}
\eea
gives a tight bound on any achievable rate of feedback codes, and hence  $C^{FB}_{A^{\infty}\rar B^\infty}$ is a candidate for the capacity of feedback codes. However, for channels with memory, it is generally not known whether the multi-letter  expression of  capacity, \eqref{cap_inf_sta}, can be reduced to a single letter expression, analogous to \eqref{cap_fb_c}.
\par Our main objective is to provide a framework for a single letter characterization of  feedback capacity for a general class of channels with memory. Towards this direction, we provide conditions on channels with memory such that 
\bea
C_{A^n \rar B^n}^{FB}=(n+1)C^{FB} \label{cap_1}
\eea
where $C^{FB}$ is a single letter expression similar to that of DMCs. Specifically, for channels of the form $\big\{{\bf P}_{B_i|B_{i-1}, A_i}: i=0, 1, \ldots, n\}$, where $B_{-1}=b_{-1}\in {\mathbb B}_{-1}$ is the initial state,  we give necessary and sufficient conditions such that the following equality holds.
\bea
C_{A^n \rar B^n}^{FB}=(n+1)\sup_{{\bf P}_{A_0|B_{-1}}(\cdot|b_{-1})}I(A_0;B_0|b_{-1}), \hso \forall b_{-1}\in {\mathbb B}_{-1}. \label{cap_1sa}
\eea
That is, the single letter expression is $C^{FB}\tri \sup_{{\bf P}_{A_0|B_{-1}}(\cdot|b_{-1})}I(A_0;B_0|b_{-1})$, and  is independent of the initial state $b_{-1} \in {\mathbb B}_{-1}$.
\subsection{Main Results and Methodology}
\par First, we consider channels with Unit Memory on the previous Channel Output (UMCO), defined by  
 \bea
{\bf P}_{B_i|B^{i-1}, A^i}={\bf P}_{B_i|B_{i-1}, A_i}, \hst i=0, 1, \ldots, n
\eea
 with and without a transmission cost constraint defined by 
\bea
\frac{1}{n+1}{\bf E}\left\{\sum_{i=0}^{n}{\gamma}^{UM}_i(A_i,B_{i-1})\right\} \label{qvcostc1_av_1} 
\eea
where ${\gamma}^{UM}_i:{\mb A}_{i}\times{\mb B}_{i-1}\longmapsto [0,\infty)$. We identify  necessary and sufficient conditions on the channel so that the optimization problem $C_{A^n \rar B^n}^{FB}$, which is generally a nested optimization problem, often dealt with via dynamic programming, reduces to a non-nested optimization problem. These conditions give rise to a single letter characterization of feedback capacity. Among other results, we  derive sufficient conditions for feedback not to increase capacity, and identify sufficient conditions for asymptotic stationarity of optimal channel input distribution and ergodicity of the joint process $\{(A_i,B_i):i=0,1,\ldots\}$. Moreover, we give an upper bound on the error probability of maximum likelihood decoding. We also treat problems with transmission cost constraints. 
\par Second, we apply the framework of the UMCO channel on the Binary State Symmetric Channel (BSSC), defined by
\begin{IEEEeqnarray}{l}
 {\bf P}_{B_i|A_i, B_{i-1}}(b_i|a_i,b_{i{-}1}) {=} \bbordermatrix{~ & 0,0 & 0,1 & 1,0 & 1,1   \cr
                  0 & \alpha & \beta & 1{-}\beta & 1{-}\alpha  \vspace*{0.5cm} \cr                   
                  1 & 1{-}\alpha & 1{-}\beta  & \beta &  \alpha \cr}, \hso  i=0, 1, \ldots, n, \hso (\alpha, \beta)\in [0,1] \times [0,1] \label{BSSC_1_intro} \IEEEeqnarraynumspace
\end{IEEEeqnarray}
  with and without a transmission cost  constraint defined by 
\bea
\frac{1}{n+1}{\bf E}\left\{\sum_{i=0}^{n}{\gamma}(A_i,B_{i-1})\right\}\leq \kappa,\hso {\gamma}(a_i,b_{i-1})=\overline{a_i\oplus b_{i-1}}, \hso \kappa\in[0,\kappa_{max}] \label{qvcostc1_av_1} 
\eea
where $\overline{x\oplus y}$ denotes the compliment of the modulo2 addition of $x$ and $y$. We calculate the capacity achieving channel input distribution with feedback without cost constraint and show that it is time-invariant. This  illustrates that feedback capacity  satisfies \eqref{cap_1}, it is independent of the initial state $B_{-1}=b_{-1}$, and it is  characterized by 
\bea
C_{A^\infty \rar B^\infty}^{FB}&=&\sup_{ {\bf P}_{A_0|B_{-1}}} I(A_0;B_0|b_{-1}), \hst \forall b_{-1} \in {\mb B}_{-1}\label{cap_2}\\&=&H(\lambda){-}\nu H({\alpha}){-}(1{-}\nu)H({\beta}) \label{cap_2a}
\eea
\par where $\lambda, \nu$ are functions of channel parameters $\alpha, \beta$ (see Theorem~\ref{op_in_out_dis_the}). The characterization (\ref{cap_2}) is precisely analogous to the single letter characterization of   (\ref{cap_fb_c}) and (\ref{cap_nf_c1}) of capacity of DMCs. Additionally, we provide the error exponent evaluated on the capacity achieving channel input distribution with feedback, and we derive an upper bound on the error probability of maximum likelihood decoding which is easy to compute (see Section \ref{ee_bssc}). Finally, we show that a time-invariant first order Markov channel input distribution without feedback achieves feedback capacity \eqref{cap_2a}, and we give the closed form expressions both for the capacity achieving channel input distribution and the corresponding channel output distribution. We also treat the case with cost constraint.
\par The main mathematical concept we invoke to obtain the above results 
are the structural properties of the optimal channel input distributions, \cite{kourtellaris2015information, kourtellarisISIT2016}. Specifically the following.
\begin{itemize}
\item[(a)] For  channels with infinite memory on the previous channel outputs defined by ${\bf P}_{B_i|B^{i-1}, A^i}={\bf P}_{B_i|B^{i-1}, A_i}$,  the maximization of directed information   $I(A^n \rar B^n)$ occurs in the  subset satisfying conditional independence $\big\{{\bf P}_{A_i|A^{i-1}, B^{i-1}}={\bf P}_{A_i|B^{i-1}}: i=0, \ldots, n\big\}$.
\item[(b)] For  channels with limited memory of order $M$ defined by ${\bf P}_{B_i|B^{i-1}, A^i}={\bf P}_{B_i|B_{i-M}^{i-1}, A_i}$, the maximization of directed information   $I(A^n \rar B^n)$ occurs in the  subset satisfying conditional independence $\big\{{\bf P}_{A_i|A^{i-1}, B^{i-1}}={\bf P}_{A_i|B_{i-M}^{i-1}}: i=0, \ldots, n\big\}$.
\item[(c)] For the UMCO channel the maximization of directed information   $I(A^n \rar B^n)$ occurs in the  subset satisfying conditional independence $\big\{{\bf P}_{A_i|A^{i-1}, B^{i-1}}={\bf P}_{A_i|B_{i-1}}: i=0, \ldots, n\big\}$.
\end{itemize}
\par The structural properties, (a), (b) and (c),  along with the fact that $C_{A^n \rar B^n}^{FB}\geq C_{A^n ; B^n}^{noFB}$, are employed in Section~\ref{section_form}  to provide sufficient conditions for feedback not to increase the capacity.  Moreover, the structural property of the UMCO channel, (c), is applied in Section~\ref{UMCO} to construct the finite horizon dynamic programming, the necessary and sufficient conditions on the capacity achieving input distribution, and the necessary and sufficient conditions for the non-nested optimization of feedback capacity. The  methodology and the corresponding theorems of Section~\ref{UMCO}  can be easily extended to channels with finite memory on  previous channel outputs by invoking the structural properties of the capacity achieving  distributions for these channels.
\subsection{Relation to the Literature}  
\par Although for several years significant effort has been devoted to the study of channels with memory, with or without feedback, explicit or closed form expressions for  capacity of such channels  are limited to few but ripe cases. For   non-stationary  non-ergodic  Additive Gaussian Noise (AGN) channels  with memory, Cover and Pombra \cite{cover-pombra1989} showed that feedback codes can increase capacity  by at most half a bit. On the other hand, for a finite alphabet version of the Cover and Pombra channel with certain symmetry, Alajaji \cite{alajaji} showed that feedback does not increase capacity.  Moreover, Permuter, Cuff, Van Roy and Weissman \cite{permuter08} derived the feedback capacity of the trapdoor channel, while Elishco and Permuter  \cite{elishco} employed dynamic programming to evaluate  feedback capacity of the Ising channel.  
\par  The capacity of channels $\big\{{\bf P}_{B_i|B_{i-1}, A_i}: i=0, \ldots, n\big\}$ for  feedback codes is analyzed  by Berger \cite{berger_shannon_lecture} and  Chen and Berger \cite{chen-berger2005}, under the assumption  that the capacity achieving distribution satisfies conditional independence  property $
{\bf P}_{A_i|A^{i-1}, B^{i-1}}={\bf P}_{A_i|B_{i-1}}(a_i|b_{i-1}), i=0,1,\ldots,n$. A derivation of this structural property of capacity achieving distribution is given  in \cite{kourtellaris2015information, kourtellarisISIT2016}. 
\par Recently, Permuter,  Asnani and Weissman \cite{asnani13,asnani13j} derived the feedback capacity for a Binary-Input Binary-Output (BIBO) channel, called the Previous Output STate (POST) channel, where the current state of the channel is the previously received symbol.    The authors in \cite{asnani13j}, showed, among other results, that feedback does not increase capacity.  It can be shown that the POST channel  is within a transformation equivalent to  the Binary State Symmetric channel (BSSC) \cite{kourtellaris_itw2015}, in which the state of the channel is defined as the modulo2 addition of the current input symbol and the previous output symbol.  When there are no transmission cost constraints, our results for the BSSC compliment existing results obtained in \cite{asnani13,asnani13j} regarding the POST channel, in the sense that, we show the time-invariant properties of the capacity achieving distributions,  which implies the single letter characterization of feedback  capacity, we derive closed form expressions for these distributions,  provide an upper bound on the error probability of maximum likelihood decoding, and we show that a first-order Markov channel input distribution without feedback achieves feedback capacity. Moreover, we derive similar closed form expressions when averaged transmission cost constraints are imposed. 
\par A portion of the results established in this paper were utilized to construct a Joint Source Channel Coding (JSCC) scheme for the $BSSC$ with a cost constraint and the Binary Symmetric Markov Source (BSMS) with single letter Hamming distortion measure \cite{kcbisit2015}. The  scheme is a natural generalization of   the JSCC design (uncoded transmission) of an Independent and Identically Distributed (IID) Bernoulli source over a Binary Symmetric Channel (BSC) \cite{jelinek,gastpar}.  

\par The remainder of the paper is organized as follows. In Section \ref{section_form}, we introduce the mathematical formulation and identify sufficient conditions for feedback not to increase  capacity. In Section \ref{UMCO}, we 
 identify sufficient conditions 
to test whether  the capacity achieving input distribution is time invariant. The results are then extended to the infinite horizon case. In Section \ref{cabistsych}, we apply the main theorems of section \ref{UMCO}  to the BSSC, with and without feedback and with and without cost constraint, to prove, among other results, that capacity is given by a single letter characterization.  Finally,  Section \ref{sec_con}  delivers our concluding remarks.

\section{Formulation \& Preliminary Results}
\label{section_form}
In this section we introduce the definitions of feedback capacity, capacity without feedback , and we identify necessary and sufficient conditions for feedback not to increase the capacity.
\subsection{Notation and Definitions}
The probability distribution of a Random Variable (RV)  defined on a probability space $(\Omega, {\cal F}, {\mathbb P})$ by the mapping $X: (\Omega, {\cal F}) \longmapsto ({\mb X}, {\cal  B}({\mb X}))$  is denoted by $ {\bf P}(\cdot) \equiv {\bf P}_X(\cdot)$. 
The space of probability distributions on $\mathbb X$ is denoted by ${\cal M(\mathbb X)}$.  A RV is called discrete if there exists a countable set ${\cal S}$ such that $\sum_{x_i \in {\cal S}} {\mathbb  P} \{ \omega \in \Omega : X(\omega)=x_i\}=1$. The probability distribution ${\bf P}_X(\cdot)$  is then concentrated on  points in ${\cal S}$, and it is defined by 
\bea
 {\bf P}_X(A)  \tri \sum_{x_i \in {\cal S} \bigcap A} {\mathbb P} \{ \omega \in \Omega : X(\omega)=x_i\}, \hso \forall A \in {\cal  B}({\mb X}). 
\eea 
Given another RV $Y:(\Omega, {\cal F})\mapsto ({\mathbb Y},{\cal  B}({\mb Y}))$, ${\bf P}_{Y|X}(dy|x)(\omega)$ is the conditional distribution of RV $Y$ given $X$. For a fixed $X=x$ we denote the conditional distribution by ${\bf P}_{Y|X}(dy|X=x)={\bf P}_{Y|X}(dy|x)$. \\
Let $\mathbb{Z}$ denote the set of integers and  ${\mathbb{N}}\tri\{0, 1,2,\dots,\}$, ${\mathbb{N}^n}\tri\{0, 1,2,\dots,n\}$. The channel input and channel output spaces are sequences of measurable spaces $\{({\mathbb A}_i,{\cal B}({\mathbb A}_i)): i \in {\mathbb Z}\}$ and $\{({\mathbb B}_i,{\cal B}({\mathbb B}_i)): i \in {\mathbb Z}\}$, respectively, while their product spaces are ${\mathbb A}^{\mathbb Z}\tri \times_{i \in \mathbb Z} {\mathbb A}_i$, ${\mathbb B}^{\mathbb Z}\tri \times_{i \in \mathbb Z} {\mathbb B}_i$,  ${\cal B}({\mathbb A}^{\mathbb Z}) \tri \otimes_{i \in \mathbb Z} {\cal B}({\mathbb A}_i)$, ${\cal B}({\mathbb B}^{\mathbb Z})\tri \otimes_{i \in \mathbb Z} {\cal B}({\mathbb B}_i)$. Points in the product spaces are denoted by $a^n \tri \{\ldots, a_{-1}, a_0, a_1, \ldots, a_n\}\in {\mb A}^n$ and $b^n \tri \{\ldots, b_{-1}, b_0, b_1, \ldots, b_n\}\in {\mb B}^n, n \in {\mathbb Z}$. 

\subsection{Capacity with Feedback \& Properties}\label{ssec:b_feed}
Next, we provide the precise formulation of information capacity and some preliminary results. We begin by introducing the definitions of channel distribution, channel input distribution, transmission cost constraint, and feedback code.\\

\begin{definition}(Channel distribution with memory)\label{def:umco_channel}\\
A sequence of conditional distributions  defined by
\begin{align}
{\cal C}_{[0,n]} \tri \Big\{{\bf P}_{B_i|B^{i-1},A^{i}}(d{b}_i|b^{i-1},a^{i})={\bf P}_{B_i|B^{i-1},A_{i}}(d{b}_i|b^{i-1},a_{i}) : \hso i=0, \ldots, n \Big\}. \label{ch_1}
\end{align}
At time $i=0$ the conditional distribution is ${\bf P}_{B_0|B^{-1},A_{0}}(d{b}_0|b^{-1},a_{0})$, where $B^{-1}=b^{-1} \in {\mb B}^{-1}$ is the initial data.
\end{definition}
The initial data, $b^{-1} \in {\mb B}^{-1}$, denotes the initial state of the channel and this should not be misinterpret as feedback information. In this work we assume that the initial data are known both to the encoder and the decoder, unless we state otherwise.\\

\begin{definition}(Channel input distribution with feedback)\\
A sequence of conditional 
distributions  defined by
\bea
{\cal P}_{[0,n]}^{FB} \tri \Big\{{\bf P}_{A_i|A^{i-1},B^{i-1}}({da}_i|a^{i-1},b^{i-1}): \hso  \hso i=0, \ldots, n\Big\}.
\eea
At time $i=0$ the conditional distribution is ${\bf P}_{A_0|A^{-1},B^{-1}}({da}_0|a^{-1},b^{-1})={\bf P}_{A_0|B^{-1}}({da}_0|b^{-1})$. That is, the information structure of the channel input distribution is  ${\cal I}_i^{FB} \tri \{ b^{-1},a_0, b_0, a_1, b_1, \ldots, a_{i-1}, b_{i-1}\}$,  for $i=0, \ldots, n$. For $i=0$ the convention is  ${\cal I}_0^{FB}\tri \{a^{-1}, b^{-1}\}=\{b^{-1}\}$, which states that the channel input distribution depends only on the initial data.\\
\end{definition}

\begin{definition}(Transmission cost constraints)\label{def:umco_cc}\\
The cost of transmitting symbols over the channel (\ref{ch_1})  is 
 a measurable function $c_{0,n}:{\mb A}^{n}\times{\mb B}^{n-1}\longmapsto [0,\infty)$  defined by 
\bea
c_{0,n}(a^n,b^{n-1})\tri\sum_{i=0}^{n}{\gamma}_{i}(a_i,b^{i-1}).
\eea
The transmission cost  constraint is  defined by 
\beae
{\cal P}_{[0,n]}^{FB}(\kappa) {\tri} \Big\{{\bf P}_{A_i|A^{i-1}, B^{i-1}}, i=0, \ldots, n:  {\frac{1}{n{+}1}} {\bf E}_\mu \big\{c_{0,n}(A^n,B^{n{-}1})\big\} \leq \kappa\Big\}, \hso \kappa\in[0,\infty]\nms
\eeae
where $\kappa \in [0, \infty)$, and the subscript notation ${\bf E}_\mu$ indicates the joint distribution over which the expectation is taken is parametrized by the  initial distribution   ${\bf P}_{B^{-1}}(d{b}^{-1})=\mu(db^{-1})$ (and of course the channel input distribution).\\
\end{definition}
\begin{definition}(Feedback code) \label{def_feedback_code}\\
A feedback code for the channel defined by (\ref{ch_1}) with  transmission cost constraint ${\cal P}_{[0,n]}^{FB}(\kappa)$  is a sequence   $\{(n, { M}_n, \epsilon_n):n=0, 1, \dots\}$, which consist of the following elements.  
\begin{itemize}
\item[(a)] A set of uniformly distributed messages ${\cal M}_n \tri \{ 1,  \ldots, M_n\}$ and a set of encoding strategies,  mapping messages  into channel inputs of block length $(n+1)$, defined by\footnote{The superscript on expectation, i.e., ${\bf E}^g$ indicates the dependence of the distribution on the encoding strategies.} 
\begin{multline}\label{block-code-nf-non}
{\cal E}_{[0,n]}^{FB}(\kappa) \triangleq  \Big\{g_i: {\cal M}_n \times {\mathbb A}^{i-1} \times {\mb B}^{i-1}  \longmapsto {\mb A}_i, \hso  a_0=g_0(w, b^{-1}), a_1=g_1(w,b^{-1},a_0,b_0),\ldots, \\ a_n=g_n(w,b^{-1},a_0,b_0,\ldots, a_{n-1}, b_{n-1}), 
  w\in {\cal M}_n: \hso  \frac{1}{n+1} {\bf E}^g\Big(c_{0,n}(A^n,B^{n-1})\Big)\leq \kappa  \Big\}, \hso n=1,2, \ldots. 
\end{multline}
The codeword for any $w \in {\cal M}_n$  is $u_w\in{\mb A}^n$, $u_w=(g_0(w,b^{-1}), g_1(w, b^{-1}, a_0, b_0),\dots,g_n(w, b^{-1},a_0, $ $b_0,\ldots, a_{n-1}, b_{n-1}))$, and ${\cal C}_n=( u_1,u_2,\dots,u_{{M}_n})$ is  the code for the message set ${\cal M}_n$, and $\{A^{-1}, B^{-1}\}$ $=\{b^{-1}\}$.  In general, the code  depends on the initial data, depending on the convention, i.e.,  $B^{-1}=b^{-1}$, which are known to the encoder and decoder (unless specified otherwise). Alternatively, we can take $\{A^{-1}, B^{-1}\}=\{\emptyset\}$.
\item[(b)] Decoder measurable mappings $d_{0,n}:{\mb B}^n\longmapsto {\cal M}_n$,  such that the average
probability of decoding error satisfies
\begin{align}
{\bf P}_e^{(n)} \triangleq \frac{1}{M_n} \sum_{w \in {\cal M}_n} {\bf  P}^g \Big\{d_{0,n}(B^{n}) \neq w |  W=w\Big\}\equiv {\bf  P}^g\Big\{d_{0,n}(B^n) \neq W \Big\} \leq \epsilon_n\nonumber
\end{align}
and the decoder may also assume knowledge of the initial data.\\
The coding rate or transmission rate over the channel is defined by  $r_n\triangleq \frac{1}{n+1} \log M_n$.
A rate $R$ is said to be an achievable rate, if there exists  a  code sequence satisfying
$\lim_{n\longrightarrow\infty} {\epsilon}_n=0$ and $\liminf_{n \longrightarrow\infty}\frac{1}{n+1}\log{{M}_n}\geq R$. 
\end{itemize}
The operational definition of feedback capacity of the channel is the supremum of all achievable rates, i.e., $C\triangleq \sup \{R: R \: \: \mbox{is achievable}\}$.\\
\end{definition}
Given any channel input distribution $\{{\bf  P}_{A_i|A^{i-1}, B^{i-1}}: i=0,1, \ldots, n\} \in {\cal P}_{[0,n]}^{FB}$,  a channel distribution $\{{\bf P}_{B_i|B^{i-1}, A_i}: i=0,1, \ldots, n\}$, and a fixed initial distribution $\mu(b^{-1})$,  then the  induced joint distribution\footnote{If $B^{-1}=b^{-1}$ is fixed, then $\mu(\cdot)=\delta_{B^{-1}}(\cdot)$ is a dirac or delta measure concentrated at $B^{-1}=b^{-1}$.}  ${\bf  P}_{A^n, B^n}$ parametrized by $\mu(\cdot)$ 
is uniquely defined, and a probability space $\Big(\Omega, {\cal F}, {\mathbb P}\Big)$ carrying the sequence of RVs $(A^n, B^n)\tri \{B^{-1}, A_0, B_0, A_1, B_1, \ldots, A_n, B_n\}$  is constructed, as follows.
\begin{align}
 {\mathbb P}\big\{A^n \in d{a}^n, B^n \in d{b}^n\big\}  \tri &
{\bf  P}_{A^n, B^n}(da^n, db^n) \nonumber \\
=&\otimes_{j=0}^n \Big({\bf P}_{B_j|B^{j-1}, A_j}(db_j|b^{j-1}, a_j)\otimes {\bf P}_{A_j|A^{j-1}, B^{j-1}}(da_j|a^{j-1}, b^{j-1})\Big)\otimes \mu(db^{-1}). \label{CIS_2gg_new} \\
{\mathbb  P}\big\{B^n \in db^n\big\} \tri& {\bf  P}_{B^n}(db^n) =  \int_{{\mb A}^n}  {\bf  P}_{A^n, B^n}(da^n, db^n). \label{CIS_3g} 
\end{align}
\begin{align}
{\bf  P}_{B_i|B^{i-1}}(db_i|b^{i-1})=&  \int_{{\mb A}^i} {\bf P}_{B_i|B^{i-1}, A_i}(db_i|b^{i-1}, a_i)\otimes {\bf P}_{A_i|A^{i-1}, B^{i-1}}(da_i|a^{i-1}, b^{i-1}) \nonumber \\
&\otimes {\bf  P}_{A^{i-1}|B^{i-1}}(da^{i-1}|b^{i-1}), \hso  i=0, \ldots, n . \label{CIS_3a}\\ \nonumber\\
{\bf  P}_{B_0|B^{-1}}(db_0|b^{-1})=&  \int_{{\mb A}_0} {\bf P}_{B_0|B^{-1}, A_0}(db_0|b^{-1}, a_0)\otimes {\bf P}_{A_0| B^{-1}}(da_0| b^{-1}). \label{CIS_3ain}
\end{align}
The Directed Information from $A^n\tri\{A_0,A_1,\ldots,A_n\}$ to $B_0^n\tri\{B_0,B_1,\ldots,B_n\}$ conditioned on $B^{-1}$ is defined by \cite{marko1973, massey1990}
\begin{align}
I(A^n\rightarrow B^n)\tri&\sum_{i=0}^{n}I(A^i;B_i|B^{i-1}) = \sum_{i=0}^n I(A_i; B_i|B^{i-1})\nonumber\\
=& \sum_{i=0}^n \int_{  }^{}   \log \Big( \frac{ {\bf P}_{B_i|B^{i-1}, A_i}(\cdot|b^{i-1}, a_i) }{{\bf P}_{B_i|B^{i-1}}(\cdot|b^{i-1})}(b_i)\Big){ \bf P}_{A^i, B^i}( da^i, db^i) \label{eqdi5_a} \\
\equiv & {\mathbb I}_{A^n \rar B^n}^{FB}({\bf P}_{A_i|A^{i-1},B^{i-1}},  {\bf P}_{B_i|B^{i-1},A_i}: i=0,1, \ldots, n)        \label{eqdi5}
\end{align}
where  (\ref{eqdi5_a}) follows from the channel definition, and  the notation ${\mathbb I}_{A^n \rar B^n}^{FB}(\cdot, \cdot)$ indicates that $I(A^n\rar B^n)$ is a functional of the sequences of channel input and channel distributions; its dependence on the  initial distribution $\mu(\cdot)$ is suppressed.\\
Define the information quantities 
\begin{align}
C_{A^n \rar B^n}^{FB} \tri  \sup_{ {\cal P}_{[0,n]}^{FB}  }I(A^n\rar B^n), \hst C_{A^n \rar B^n}^{FB}(\kappa) \tri  \sup_{ {\cal P}_{[0,n]}^{FB}(\kappa)  }I(A^n\rar B^n). \label{FBLF_1}
\end{align}
Under the assumption that  $\{B^{-1}, A_0, B_0, A_1, B_1, \ldots, \}$ is  jointly ergodic or  $\frac{1}{n+1} \sum_{i=0}^n \frac{{\bf P}_{B_i|B^{i-1}, A_i}(\cdot|B^{i-1}, A_i)}{{\bf  P}_{B_i|B^{i-1}}(\cdot|B^{i-1})}(B_i)$ is information stable \cite{dubrushin1958,tatikonda2000} and $c_{0,n}(a^n, b^{n-1})=\frac{1}{n+1}\sum_{i=0}^{n}\gamma_i(A_i,B^{i-1})$ is stable, then the capacity of the channel with  feedback with and without transmission cost   is given by
\begin{align}
 C_{A^\infty \rar B^\infty}^{FB} \tri \lim_{n \longrightarrow  \infty}{\frac{1}{n+1}} C_{A^n \rar B^n}^{FB}, \hst C_{A^\infty \rar B^\infty}^{FB}(\kappa) \tri \lim_{n \longrightarrow  \infty}{\frac{1}{n+1}} C_{A^n \rar B^n}^{FB}(\kappa).\label{feed_capa}
\end{align}

\subsubsection{Convexity Properties.}

Next, we recall the convexity properties of directed information with respect to a specific definition of channel input distributions, which is equivalent to the above definition. \\
Any sequence of channel input distribution $\{{ \bf P}_{A_i|A^{i-1}, B^{i-1}}: i=0,1, \ldots, n\} \in {\cal P}_{[0,n]}^{FB}$ and channel distribution $\{{\bf P}_{B_i|B^{i-1}, A_i}: i=0,1, \ldots, n\}$
 uniquely define the causal conditioned distributions 
\begin{align}
{\overleftarrow {\bf P}}(da^n|b^{n-1}) &\tri \otimes_{i=0}^n {\bf P}_{A_i|A^{i-1}, B^{i-1}}(da_i|a^{i-1}, b^{i-1}), \label{cc_1} \\
 \overrightarrow{ {\bf P}}(db_0^n|a^n,b^{-1}) &\tri \otimes_{i=0}^n {\bf P}_{B_i|B^{i-1}, A^{i}}(db_i|b^{i-1}, a_{i})  \label{cc_2}
\end{align} 
and vice-versa, and these are parametrized by the initial data $b^{-1}$. Moreover, for a fixed $B^{-1}=b^{-1}$ we can formally define the joint distribution of $\{A_0, B_0, A_1, B_1, \ldots, A_n, B_n\}$   and the joint distribution of $\{B_0, B_1, \ldots, B_n\}$ conditioned on $B^{-1}=b^{-1}$ by   
\begin{align}
{\bf P}^{\overleftarrow{P}}(da^n,db_0^n|b^{-1})\tri & (\overleftarrow{ \bf P}\otimes \overrightarrow{\bf P})( da^n, db_0^n|b^{-1}), \\
{\bf P}^{\overleftarrow{P}}(db_0^n|b^{-1})\tri& \int_{\mb A^n} (\overleftarrow{ \bf P}\otimes \overrightarrow{\bf P})( da^n, db_0^n|b^{-1}).
\end{align}
Both distributions are parametrized by the initial data $b^{-1}$. Then, from \cite{charalambous-stavrou2012}, we have the following convexity property of  directed information.
\begin{itemize}
\item[(a)] The set of conditional distributions defined by (\ref{cc_1}), ${\overleftarrow {\bf P}}_{A^n|B^{n-1}}(\cdot|b^{n-1})\in {\cal M}({\mathbb A}^n)$ is convex.
\item[(b)] Directed information is equivalently expressed as follows.
\begin{align}
I(A^n\rightarrow B^n) =& \int   \log \Big( \frac{ \overrightarrow{\bf P}(\cdot|a^n,b^{-1})}{{\bf P}^{\overleftarrow{P}}(\cdot|b^{-1})}(b_0^n)\Big){\bf P}^{\overleftarrow{P}}( da^n, db_0^n|b^{-1})\otimes \mu(db^{-1})
\nonumber \\
\equiv & {\mathbb I}_{A^n \rar B^n}^{FB}(\overleftarrow{\bf P}, \overrightarrow{ \bf P}) .       \label{eqdi5}
\end{align}
\item[(c)] Directed information, ${\mathbb I}_{A^n \rar B^n}^{FB}(\overleftarrow{\bf P}, \overrightarrow{ \bf P})$, is concave with respect to ${\overleftarrow {\bf P}}(\cdot|b^{n-1})\in {\cal M}({\mathbb A}^n)$ for a fixed $\overrightarrow{\bf  P}(\cdot|a^n,b^{-1}) \in {\cal M}({\mathbb B}_{0}^n)$. 
\end{itemize}

Since the set of conditional distributions with or without  transmission cost constraints is convex, and directed information is a concave functional, the optimization problems (\ref{FBLF_1}) are convex, and  we have the following theorem.

\ \

\begin{theorem}(Convexity properties)\\
\label{thm-pr_fb}
Assume the set $ {\cal P}_{[0,n]}^{FB}(\kappa)$ is non-empty and  the supremum of $I(A^n \rar B^n)$  over the  set of distributions $ {\cal P}_{[0,n]}^{FB}(\kappa)$ is achieved (i.e., it exists). Then, the following hold.
\begin{itemize}
\item[(a)] $C_{A^n \rar B^n}^{FB}(\kappa)$ is non-decreasing  concave function of $\kappa \in [0, \infty]$.
\item[(b)] An alternative characterization of  $C_{A^n \rar B^n}^{FB}(\kappa)$ is given by
 \begin{align}
 C_{A^n \rar B^n}^{FB}(\kappa)=\sup_{\overleftarrow{\bf P}:  {\frac{1}{n{+}1}} {\bf  E}_\mu \big\{ c_{0,n}(a^n,b^{n{-}1}) \big\}  =  \kappa } {\mathbb I}_{A^n \rar B^n}^{FB}(\overleftarrow{\bf P}, \overrightarrow{\bf  P}),  \hst \mbox{for} \hso  \kappa \leq \kappa_{max} \label{extr_con_pro}
 \end{align}
where $\kappa_{max}$  is the smallest number belonging to $[0,\infty]$ such that $C_{A^n \rar B^n}^{FB}(\kappa)$ is constant in $[\kappa_{max}, \infty]$, and ${\bf E}_\mu\big\{\cdot \big\}$ denotes expectation with respect to the joint distribution $(\overleftarrow{ \bf P}\otimes \overrightarrow{\bf P})\otimes \mu$.
\end{itemize}  
\end{theorem}

\begin{proof} Since the set $ {\cal P}_{[0,n]}^{FB}(\kappa)$ is convex with respect to $\overleftarrow{ \bf P}(\cdot|b^{n-1})\in {\cal M}({\mathbb A^n})$, the statements follow from 
the convexity and non-decreasing properties \cite{charalambous-stavrou2012}.
\end{proof}

The above theorem states that the extremum problem of feedback capacity is a convex optimization problem, over appropriate sets of distributions.  \\

\subsubsection{Information Structures of Optimal Channel Input Distributions.} Consider the extremum problem $C_{A^n \rar B^n}^{FB}(\kappa)$, given by \eqref{extr_con_pro}. In \cite{kourtellaris2015information, kourtellarisISIT2016}, it is shown that the optimal channel input distribution satisfies the following conditional independence.
\bea
{\bf P}_{A_i|A^{i-1},B^{i-1}}({da}_i|a^{i-1},b^{i-1})={\bf P}_{A_i|B^{i-1}}({da}_i|b^{i-1})\equiv \pi_i(da_i|b^{i-1}), \hso  i=0, \ldots, n. \label{CAD}
\eea
Moreover, in view of the information structure of the optimal channel input distribution, $C_{A^n \rar B^n}^{FB}(\kappa)$ reduces to the following optimization problem.
\begin{align}
C_{A^n \rar B^n}^{FB}(\kappa)=& \sup_{ \overline{\cal P}_{[0,n]}^{FB}(\kappa) } \sum_{i=0}^n \int_{  }^{}   \log \Big( \frac{ {\bf P}_{B_i|B^{i-1}, A_i}(\cdot|b^{i-1}, a_i) }{{\bf P}_{B_i|B^{i-1}}^\pi(\cdot|b^{i-1})}(b_i)\Big){ \bf P}_{A_i, B^i}^\pi( da_i, db^i) \label{IS_1} \\
\equiv & \sup_{ \overline{\cal P}_{[0,n]}^{FB}(\kappa)  }  {\mathbb I}_{A^n \rar B^n}^{FB}(\pi_i,  {\bf P}_{B_i|B^{i-1},A_i}: i=0,1, \ldots, n)      \label{func_1}  
\end{align}
where the transmission cost constraint is defined by 
\beae
\overline{\cal P}_{0,n}^{FB}(\kappa) {\tri} \Big\{ \pi_i(da_i|b^{i-1}), i=0, \ldots, n:  {\frac{1}{n{+}1}} {\bf E}_\mu^\pi \big\{c_{0,n}(A^n,B^{n{-}1})\big\} \leq \kappa\Big\}
\eeae
and the induced joint and transition probability distributions are given by  
\begin{align}
{\bf  P}_{A_i, B^i}^\pi(da_i, db^i)
=&{\bf P}_{B_i|B^{i-1}, A_i}(db_i|b^{i-1}, a_i)\otimes \pi_i(da_i|b^{i-1})\otimes {P}^{\pi}_{B^{i-1}}(d{b}^{i-1}) \label{CIS_2gg_new_n} \\
{\bf  P}_{B_i|B^{i-1}}^\pi(db_i|b^{i-1})=&  \int_{{\mb A}_i} {\bf P}_{B_i|B^{i-1}, A_i}(db_i|b^{i-1}, a_i)\otimes {\pi}_i(da_i|b^{i-1}), \hso  i=0, \ldots, n . \label{CIS_3a_n}
\end{align}
The superscript indicates the dependence of these distributions on $\{\pi_i(da_i|b^{i-1}): i=0, \ldots, n\}$.\\
The information feedback capacity rate  is then given by 
\begin{align}
C_{A^\infty \rar B^\infty}^{FB}(\kappa) \tri \lim_{n \longrightarrow \infty} \frac{1}{n+1} \sup_{ \overline{\cal P}_{[0,n]}^{FB}(\kappa)  }  {\mathbb I}_{A^n \rar B^n}^{FB}(\pi_i,  {\bf P}_{B_i|B^{i-1},A_i}: i=0,1, \ldots, n).  \label{ifc_1}      
\end{align}

\subsection{Feedback Versus No Feedback}
\label{ch4bsscnf}
Here, we address the question whether feedback increases capacity via optimization problem (\ref{IS_1}). First, we recall the definition of channel input distributions without feedback.\\
\begin{definition}(Channels input distribution without feedback)\\
A sequence of conditional 
distributions  defined by
\bea
{\cal P}_{[0,n]}^{noFB} \tri \Big\{{\bf P}_{A_i|A^{i-1},B^{-1}}({da}_i|a^{i-1},b^{-1}) \equiv \pi_i^{noFB}(da_i|a^{i-1},b^{-1}): \hso i=0, \ldots, n\Big\}.\label{nofb_cid}
\eea
The information structure of the channel input distribution without feedback is ${\cal I}_i^{noFB}\tri  \{a^{i-1},b^{-1}\}$. For time $i=0$, the distribution is ${\bf P}_{A_0|A^{-1},B^{-1}}({da}_0|a^{-1},b^{-1}) \equiv \pi_i^{noFB}(da_0|b^{-1})$, hence the information structure is  ${\cal I}_0^{FB}\tri \{a^{-1}, b^{-1}\}=\{b^{-1}\}$, which states that the channel input distribution depends only on the initial data.
\end{definition}
\par Similar to the feedback case (Section~\ref{ssec:b_feed}), the initial state of the channel, $b^{-1}$, is assumed to be known at the encoder.  The transmission cost constraint without feedback, is defined by 
\beae
{\cal P}_{[0,n]}^{noFB}(\kappa) {\tri} \Big\{\pi_i^{noFB}(da_i|a^{i-1},b^{-1}), i=0, \ldots, n:  {\frac{1}{n{+}1}} {\bf E}_\mu \big\{c_{0,n}(A^n,B^{n{-}1})\big\} \leq \kappa\Big\}, \hso \kappa\in[0,\infty]. \label{nofb_cc} \nms
\eeae
Moreover, the set of encoding strategies without feedback,  mapping messages  into channel inputs of block length $(n+1)$, are defined by
\begin{multline}\label{block-code-nfeed}
{\cal E}_{[0,n]}^{noFB}(\kappa) \triangleq  \Big\{g^{noFB}_i: {\cal M}_n \times {\mathbb A}^{i-1} \times {\mb B}^{-1}  \longmapsto {\mb A}_i, \hso  a_0=g^{noFB}_0(w, b^{-1}), a_1=g^{noFB}_1(w,b^{-1},a_0),\ldots, \\ a_n=g^{noFB}_n(w,b^{-1}, a^{n-1}), 
  w\in {\cal M}_n: \hso  \frac{1}{n+1} {\bf E}^{g^{noFB}}\Big(c_{0,n}(A^n,B^{n-1})\Big)\leq \kappa  \Big\}, \hso n=0, 1, \ldots. 
\end{multline}
By employing \eqref{nofb_cc} and \eqref{block-code-nfeed}, a code without feedback is defined similarly to Definition~\ref{def_feedback_code}. 
\par Given any channel input distribution without feedback $\{{\pi}^{noFB}_{i}(d{a}_i|a^{i-1},b^{-1}): i=0,1, \ldots, n\} \in {\cal P}_{[0,n]}^{noFB}(\kappa)$,  a channel distribution $\{{\bf P}_{B_i|B^{i-1}, A_i}: i=0,1, \ldots, n\}$, and a fixed initial distribution ${\bf P}_{B^{-1}}(db^{-1})=\mu(b^{-1})$,  then the  induced joint distribution ${\bf  P}_{A^n, B^n}$ parametrized by $\mu(\cdot)$ is uniquely defined. The mutual information between from $A^n\tri\{A_0,A_1,\ldots,A_n\}$ to $B_0^n\tri\{B_0,B_1,\ldots,B_n\}$ conditioned on $B^{-1}$ is defined by 
\begin{align}
I(A^n ; B^n)\tri 
& {\bf E}_\mu^{\pi^{noFB}}\Big\{ \log \Big( \frac{ {\bf P}_{B_0^n|A^n,B^{-1}}(\cdot|A^n,B^{-1}) }{{\bf P}_{B_0^n|B^{-1}}^{\pi^{noFB}}(\cdot|B^{-1})}(B_0^n)\Big)\Big\} \nonumber \\
=&\sum_{i=0}^n \int_{  }^{} \log \Big( \frac{ {\bf P}_{B_i|B^{i-1}, A_i}(\cdot|b^{i-1}, a_i) }{{\bf P}_{B_i|B^{i-1}}^{\pi^{noFB}}(\cdot|b^{i-1})}(b_i)\Big){ \bf P}_{A^i, B^i}^{\pi^{noFB}}( da^i, db^i)
\nonumber \\
=&\sum_{i=0}^n \int_{  }^{} \log \Big( \frac{ {\bf P}_{B_i|B^{i-1}, A_i}(\cdot|b^{i-1}, a_i) }{{\bf P}_{B_i|B^{i-1}}^{\pi^{noFB}}(\cdot|b^{i-1})}(b_i)\Big){ \bf P}_{A_i, B^i}^{\pi^{noFB}}( da_i, db^i)
\nonumber \\
\equiv & {\mathbb I}_{A^n \rar B^n}^{noFB}(\pi_i^{noFB},  {\bf P}_{B_i|B^{i-1},A_i}: i=0,1, \ldots, n)        \label{eqdi5_nFB}
\end{align}
where the joint distribution and transition probability distribution are induced by $\{\pi_i^{noFB}(da_i|a^{i-1},b^{-1}): i=0, \ldots, n\} \in {\cal P}_{[0,n]}^{noFB}$ as follows. 
\begin{align}
{\bf  P}_{A_i, B^i}^{\pi^{noFB}}(da_i, db^i)
=&{\bf P}_{B_i|B^{i-1}, A_i}(db_i|b^{i-1}, a_i)\otimes {\bf P}_{A_i|B^{i-1}}^{\pi^{noFB}}(da_i|b^{i-1})\otimes {\bf P}_{B^{i-1}}^{\pi^{noFB}}(db^{i-1}) \label{CIS_2gg_new_n_nf} \\
{\bf  P}_{B_i|B^{i-1}}^{\pi^{noFB}}(db_i|b^{i-1})=&  \int_{{\mb A}_i} {\bf P}_{B_i|B^{i-1}, A_i}(db_i|b^{i-1}, a_i)\otimes {\bf P}_{A_i|B^{i-1}}^{{\pi^{noFB}}}(da_i|b^{i-1}), \hso  i=0, \ldots, n  \label{CIS_3a}\\
{\bf P}_{A_i|B^{i-1}}^{\pi^{noFB}}(da_i|b^{i-1})=&\int_{{\mb A}^{i-1}}{\pi}^{noFB}_{i}(d{a}_i|a^{i-1},b^{-1})\otimes {\bf P}_{A^{i-1}|B^{i-1}}^{\pi^{noFB}}(da^{i-1}|b^{i-1}) \label{CIS_2gg_new_n_nf21}\\
{\bf P}_{A^{i-1}|B^{i-1}}^{\pi^{noFB}}(da^{i-1}|b^{i-1})=&\otimes_{j=0}^{i-1}\frac{{\bf P}_{B_j|B^{j-1}, A_j}(db_j|b^{j-1}, a_j)\otimes {\pi}_{j}^{noFB}(da_j|a^{j-1}, b^{-1})}{\int_{{\mb A}_{j}}{\bf P}_{B_j|B^{j-1}, A_j}(db_j|b^{j-1}, a_j)\otimes {\pi}_{j}^{noFB}(da_j|a^{j-1}, b^{-1})}. \label{CIS_2ggnf221}
\end{align}
The superscript in the above distributions are important to distinguish that these are generated by the channel and channel input distributions without feedback, while the functional in (\ref{eqdi5_nFB}) is fundamentally different from the one in (\ref{func_1}). Clearly, compared to the channel with feedback in which the corresponding distributions are  (\ref{CIS_2gg_new_n}) and (\ref{CIS_3a_n}), and they are induced by $\{\pi_i(da_i|b^{i-1}): i=0, \ldots, n\} \in \overline{\cal P}_{0,n}^{FB}(\kappa)$, when  the channel is used without feedback, the distributions (\ref{CIS_2gg_new_n_nf}) and (\ref{CIS_3a}) are induced by $\{\pi_i^{noFB}(da_i|a^{i-1},b^{-1}): i=0, \ldots, n\} \in {\cal P}_{[0,n]}^{noFB}(\kappa)$. 
\par Define the information quantity 
\begin{align}
C_{A^n ; B^n}^{noFB}(\kappa)=& \sup_{ {\cal P}_{[0,n]}^{noFB}(\kappa) } \sum_{i=0}^n \int_{  }^{}   \log \Big( \frac{ {\bf P}_{B_i|B^{i-1}, A_i}(\cdot|b^{i-1}, a_i) }{{\bf P}_{B_i|B^{i-1}}^{\pi^{noFB}}(\cdot|b^{i-1})}(b_i)\Big){ \bf P}_{A_i, B^i}^{\pi^{noFB}}( da_i, db^i) \label{IS_1_nfb} \\
\equiv & \sup_{{\cal P}_{[0,n]}^{noFB}(\kappa)  }  {\mathbb I}_{A^n \rar B^n}^{noFB}(\pi_i^{noFB},  {\bf P}_{B_i|B^{i-1},A_i}: i=0,1, \ldots, n).       
\end{align}  
Then the information capacity without feedback subject to a transmission cost constraint is defined by 
\begin{align}
C_{A^\infty ; B^\infty}^{noFB}(\kappa) \tri \lim_{n \longrightarrow \infty} \frac{1}{n+1} \sup_{ \overline{\cal P}_{[0,n]}^{noFB}(\kappa)  }  {\mathbb I}_{A^n \rar B^n}^{noFB}(\pi_i^{noFB},  {\bf P}_{B_i|B^{i-1},A_i}: i=0,1, \ldots, n).  \label{ifc_1_noFB}      
\end{align}
Next, we note the following. Let $\{\pi_i^*(da_i|b^{i-1}): i=0, \ldots, n\} \in \overline{\cal P}_{[0,n]}^{FB}(\kappa)$ denote the maximizing distribution  in $C_{A^n \rar B^n}^{FB}(\kappa)$ defined by (\ref{IS_1}). Suppose there exists a sequence of channel  input distributions without feedback $\{{\bf P}^*(da_i| {\cal I}_i^{noFB})\equiv \pi_i^{*,noFB}(da_i|{\cal I}_i^{noFB}): {\cal I}_i^{noFB} \subseteq \{b^{-1}, a_0, \ldots, a_{i-1}\}, i=0, \ldots,n\} \in {\cal P}_{0,n}^{noFB}(\kappa)$ 
 which induces the maximizing channel input distribution with feedback 
$\{\pi^*_i(da_i|b^{i-1}): i=0,1, \ldots, n\}$. That is, ${\bf P}_{A_i|B^{i-1}}^{noFB}(da_i|b^{i-1})$  given by \eqref{CIS_2gg_new_n_nf21} is equal to $\pi_i^*(da_i|b^{i-1})$, $\forall i=0,1,\ldots,n$. Then, it is clear that this sequence also induces the optimal joint distribution and conditional distribution defined by   (\ref{CIS_2gg_new_n}), (\ref{CIS_3a_n}), and consequently $C_{A^n \rar B^n}^{FB}(\kappa)$ and $C_{A^\infty \rar B^\infty}^{FB}(\kappa)$ are achieved without using feedback. 
\par In the following theorem, we prove that this condition  is not only sufficient but also necessary for any  channel input distribution without feedback to achieve  the finite time feedback information capacity, $C_{A^n \rar B^n}^{FB}(\kappa)$.\\
\begin{theorem}(Necessary and sufficient conditions for $C_{A^n \rar B^n}^{FB}(\kappa)=C_{A^n ; B^n}^{noFB}(\kappa)
$)\\
\label{ch4gpnf}
\noi Consider channel (\ref{ch_1}) and let $\{\pi_i^*(da_i|b^{i-1}):i=0, \ldots, n\} \in \overline{\cal P}_{[0,n]}^{FB}(\kappa)$ denote the maximizing distribution  in $C_{A^n \rar B^n}^{FB}(\kappa)$ defined by (\ref{IS_1}),  and let $\Big\{{\bf  P}_{A_i, B^i}^{\pi^*}(da_i, db^i), {\bf  P}_{B_i|B^{i-1}}^{\pi^*}(db_i|b^{i-1}): i=0, \ldots, n\Big\}$ denote the corresponding joint and transition distributions as defined by   (\ref{CIS_2gg_new_n}), (\ref{CIS_3a_n}).  \\
Then 
\bea
C_{A^n \rar B^n}^{FB}(\kappa)=C_{A^n ; B^n}^{noFB}(\kappa) \label{feenotincap}
\eea
 if and only if there exists a sequence of channel  input distributions  
\beae
\Big\{{\bf P}^*(da_i| {\cal I}_i^{noFB})\equiv \pi_i^{noFB,*}(da_i|{\cal I}_i^{noFB}): {\cal I}_i^{noFB} \subseteq \{b^{-1}, a_0, \ldots, a_{i-1}\}, i=0, \ldots,n\Big\} \in {\cal P}_{0,n}^{noFB}(\kappa)\nonumber
\eeae
 which 
induces the maximizing channel input distribution with feedback 
$\{\pi^*_i(da_i|b^{i-1}): i=0,1, \ldots, n\}$.
\end{theorem}

\begin{proof}  In general,  the inequality $C_{A^n \rar B^n}^{FB}(\kappa)\geq C_{A^n ; B^n}^{noFB}(\kappa)$ holds. Moreover, by Section~\ref{ssec:b_feed} the distributions $\{{\bf  P}_{A_i, B^i}^{\pi^*}(da_i, db^i), {\bf  P}_{B_i|B^{i-1}}^{\pi^*}(db_i|b^{i-1}): i=0, \ldots, n\}$ are induced by the channel, which is fixed, and the optimal conditional distribution $\{\pi_i^*(da_i|b^{i-1}): i=0, \ldots, n\} \in \overline{\cal P}_{[0,n]}^{FB}(\kappa)$. Then, equality holds  if and only if there exists a distribution without feedback   $\{\pi_i^{noFB,*}(da_i|{\cal I}_i^{noFB}): i=0, \ldots, n\} \in   {\cal P}_{0,n}^{noFB}(\kappa)$ which induces  $\{\pi_i^*(da_i|b^{i-1}): i=0, \ldots, n\} \in \overline{\cal P}_{[0,n]}^{FB}(\kappa)$. This follows from the fact that the distributions $\Big\{{\bf  P}_{A_i, B^i}^{\pi^*}(da_i, db^i),$ ${\bf  P}_{B_i|B^{i-1}}^{\pi^*}(db_i|b^{i-1}): i=0, \ldots, n\Big\}$ are induced by the feedback distribution, $\{\pi_i^*(da_i|b^{i-1}): i=0, \ldots, n\}$ and the channel distribution. This completes the proof. 
\end{proof}
Theorem~\ref{ch4gpnf} provides a sufficient condition for feedback not to increase  capacity, i.e. $C_{A^{\infty} \rar B^{\infty}}^{FB}(\kappa)=C_{A^{\infty} ; B^{\infty}}^{noFB}(\kappa)$, since if \eqref{feenotincap} holds, then $\lim_{n\rar\infty}\frac{1}{n+1}C_{A^n \rar B^n}^{FB}(\kappa)=\lim_{n\rar\infty}\frac{1}{n+1}C_{A^n ; B^n}^{noFB}$ . In Section~\ref{sec:cap_nf} we demonstrate an application of Theorem~\ref{ch4gpnf} to a specific channel with memory, where we show that an input distribution without feedback induces $\{\pi_i^*(da_i|b^{i-1}): i=0, \ldots, n \}$, hence feedback does not increase capacity. 

\section{Dynamic Programming and Necessary Sufficient Conditions for Non-nested Optimization}
\label{UMCO}
In this section we employ the structural properties of capacity achieving channel input distributions with feedback to derive dynamic programming recursions and  necessary and sufficient conditions for the single letter characterization \eqref{cap_1}, to hold. Specifically, we provide the following results for the UMCO channel.
\begin{description}
\item[(a)] Necessary and sufficient conditions to determine when dynamic programming recursions, which are nested optimization problems, reduce to non-nested optimization problems. 

\item[(b)] Repeat (a) for the per unit time infinite horizon.

\item[(c)] Upper bounds on the probability of maximum likelihood decoding.
\end{description}
\par The time-varying UMCO channel is defined by
\begin{align}
{\bf P}_i(b_i|b^{i-1},a^i)\tri {\bf P}_i(b_i|b_{i-1},a_i), \hso i=0, \ldots, n  \label{UMCO_CHAN_SEC3}
\end{align}
and the transmission cost constraint is defined by 
\bea
\frac{1}{n+1}{\bf E}\left\{\sum_{i=0}^{n}{\gamma}^{UM}_i(A_i,B_{i-1})\right\}\leq \kappa \label{CC_UMCO_CHAN_SEC3} 
\eea
where ${\gamma}^{UM}_i:{\mb A}_{i}\times{\mb B}_{i-1}\longmapsto [0,\infty)$. At $i=0$ the conditional distribution depends on   $\{b_{-1}, a_0\}$, where $b_{-1} \in {\mb B}_{-1}$ is the initial data which are either known to the encoder and the decoder or, $b_{-1}=\{\emptyset\}$.
For simplicity, of presentation and technical assumptions needed, we consider a channel model with transmission cost function, defined on finite alphabet spaces. However, all main results extend to abstract alphabet spaces and channel distributions, which depend on finite memory on past channel output. Moreover, our analysis and the corresponding theorems  can be  extended to channels with finite memory on the previous channel outputs by exploiting the structural form of the capacity achieving  distributions given in \cite{kourtellaris2015information, kourtellarisISIT2016}.

%
%
%
For the above model, it is shown in \cite{{kourtellaris2015information}, kourtellarisISIT2016} that maximizing directed information, $I(A^n\rar B^n)$, over ${\cal P}^{FB}_{[0,n]}$ or ${\cal P}^{FB}_{[0,n]}(\kappa)$ occurs in the subset of conditional distributions that satisfy the following conditional independence.
\bea
{\bf P}_i(a_i|a^{i-1},b^{i-1})={\bf P}_i(a_i|b_{i-1})\equiv \pi_i(a_i|b_{i-1}), \ i=0,1,\ldots,n. \label{ch_in_disd} 
\eea
Consequently, we have the following Markovian properties.  
\bea
 {\bf P}_i(a_i,b_i|a^{i-1}, b^{i-1})&=&{\bf P}_i^\pi(a_i,b_i| a_{i-1}, b_{i-1}),  \hso i=0,1, \ldots, n, \label{FOM1} \\     
      {\bf P}_i({b_i|b^{i-1}})&=&{\bf P}_i^\pi({b_i| b_{i-1}}), \hso  i=0,1, \ldots, n,
       \label{FOM_2}\\
 {\bf P}_i^\pi(b_i|b_{i-1}) &=&\sum_{a_i \in {\mb A}_i} {\bf P}_i(b_i|b_{i-1}, a_i) {\pi}_i(a_i|b_{i-1}),  \hso i=0,1, \ldots, n.
 \label{CID_A.1}     
      \eea
where the superscript indicates the dependence on the channel input distribution \eqref{ch_in_disd}.      
In view of these Markov properties, the characterization of the FTFI capacity (i.e., (\ref{FBLF_1})) is
given by\footnote{When clear from the context, the subscript notation of the distributions is omitted, i.e., ${\bf P}^{\pi}_{B_i|B_{i-1}}(b_i|b_{i-1})\equiv {\bf P}^{\pi}_i(b_i|b_{i-1})$.} 
\begin{align}
C_{A^n \rar B^n}^{FB,UMCO} \tri &\sup_{\sr{\circ}{\cal P}^{FB}_{[0,n]} } {\bf E}^{\pi}_{\mu} \Big\{  \sum_{i=0}^n  \log \Big(   \frac{{\bf P}_i(B_i|B_{i-1}, A_i)}{{\bf P}_i^\pi(B_i|B_{i-1})}\Big) \Big\}  \\
 =&\sup_{\sr{\circ}{\cal P}^{FB}_{[0,n]} } \sum_{i=0}^n I(A_i; B_i|B_{i-1})\label{eq.x1}
 \end{align}
 where
 \begin{align}
 \sr{\circ}{ {\cal P}}_{[0,n]}^{FB} \tri&     \big\{ \pi_i(a_i|b_{i-1}):i=0, 1, \ldots, n\big\} \subset  {\cal P}_{[0,n]}^{FB}. 
\end{align}
Similarly, for conditional distributions with transmission cost the characterization of FTFI capacity  is given by 
\begin{align}
C_{A^n \rar B^n}^{FB,UMCO}(\kappa) \tri &\sup_{ \sr{\circ}{ {\cal P}}_{[0,n]}^{FB}(\kappa)   } {\bf E}^\pi_{\mu} \Big\{    \sum_{i=0}^n  \log \Big(   \frac{ {\bf P}_i(B_i|B_{i-1}, A_i)}{ {\bf P}_i^\pi(B_i|B_{i-1})}\Big) \Big\} \\
 =&\sup_{\sr{\circ}{ {\cal P}}_{[0,n]}^{FB}(\kappa)  } \sum_{i=0}^n I(A_i; B_i|B_{i-1}) \label{CID_A.2.2}
\end{align}
where
\bea
\sr{\circ}{ {\cal P}}_{[0,n]}^{FB}(\kappa) \tri     \Big\{ {\pi}_i(a_i|b_{i-1}), i=0, 1, \ldots, n:  \: \frac{1}{n+1} \sum_{i=0}^n  {\bf E}^\pi_{\mu}\Big( \gamma_i^{UM}(A_i, B_{i-1})\Big) \leq \kappa \Big\}.  \label{CID_A.2}
\eea
Since the joint process $\{B_{-1}, A_0, B_0, \ldots, A_n, B_n \}$ and channel output process $\{B_{-1}, B_0, \ldots, B_n\}$ are Markov, we   explore the connection of the above optimization problems to Markov Decision theory,  to derive  the results listed in (a)-(c). 
We do this in the next sections. 
\subsection{Necessary and Sufficient Conditions via Dynamic Programming: The Finite Horizon case}\label{subsec.DPandALG}
To derive the necessary and sufficient conditions for any channel input distribution to maximize directed information, i.e., item (a),   we first apply dynamic programming on a finite horizon. 

\subsubsection{Without Transmission Cost Constraint}
The dynamic programming recursion for $C_{A^n \rar B^n}^{FB,UMCO}$ is obtained as follows.  Let $V_t(b_{t-1})$ represent the value function, that is, the maximum expected total cost on the future time horizon $\{t,t+1,\dots,n\}$ given output $B_{t-1}=b_{t-1}$ at time $t-1$, defined by
\begin{align}
V_t(b_{t-1})&=\sup_{\pi_i(a_i|b_{i-1}):i=t,t+1, \dots,n}{\bf E}^\pi\Big\{\sum_{i=t}^{n}\log\Big(\frac{{\bf P}_i(B_i|B_{i-1},A_i)}{{\bf P}_i^\pi(B_i|B_{i-1})}\Big)|B_{t-1}=b_{t-1}\Big\}\label{dyn_ck1}
\end{align}
where the  transition probability of the channel output process is 
\begin{equation}
{\bf P}_t^\pi(b_t|b_{t-1})=\sum_{a_t \in {\mb A}_t}{\bf P}_t(b_t|b_{t-1},a_t)\pi_t(a_t|b_{t-1}). \label{T_P_1}
\end{equation}
Then  (\ref{dyn_ck1}) satisfies the following dynamic programming recursions.
\begin{align}
V_n(b_{n-1})=&\sup_{\pi_n(a_n|b_{n-1})}\sum_{(a_n, b_n) \in {\mb A}_n \times {\mb B}_n}   \log\Big(\frac{{\bf P}_n(b_n|b_{n-1},a_n)}{{\bf P}_n^\pi(b_n|b_{n-1})}\Big){\bf P}_n(b_n|b_{n-1},a_n)\pi_n(a_n|b_{n-1}), \label{DP.eq.3a}\\
V_t(b_{t-1})=&\sup_{\pi_t(a_t|b_{t-1})}\Big\{  \sum_{a_t \in {\mb A}_t} \Big[ \sum_{b_t \in {\mb B}_t}      \log\Big(\frac{{\bf P}_t(b_t|b_{t-1},a_t)}{{\bf P}_t^\pi(b_t|b_{t-1})}\Big){\bf P}_t(b_t|b_{t-1},a_t) \nonumber \\
& +\sum_{b_t \in {\mb B}_t}V_{t+1}(b_{t}){\bf P}_t(b_t|b_{t-1},a_t)\Big]\pi_t(a_t|b_{t-1})\Big\},\quad t=0,1,\dots,n-1.\label{DP.eq.3b}
\end{align}
For a fixed initial distribution ${\bf P}_{B_{-1}}(b_{-1})=\mu(b_{-1})$ we have
\bea 
C_{A^n \rar B^n}^{FB,UMCO}=\sum_{b_{-1} \in {\mb B}_{-1} } V_0(b_{-1})\mu(b_{-1}).\label{conn_val_fu}
\eea
Clearly, by using the properties of relative entropy, we can show that the right hand side of the dynamic programming recursion, (\ref{DP.eq.3a}), is a concave function of the input distribution $\pi_n(a_n|b_{n-1})$. Similarly at each step of the recursion, the right hand side of the dynamic programming recursion, (\ref{DP.eq.3b}), is a concave function of the input distribution $\pi_t(a_t|b_{t-1})$, since the future channel input distributions, $\{\pi_{t+1}(a_{t+1}|b_{t}),\ldots,\pi_n(a_n|b_{n-1})\}$ are fixed to their optimal strategies. Utilizing this observation we have the following necessary and sufficient conditions for any channel input distribution to maximize the right hand side of the dynamic programming recursions \eqref{DP.eq.3a} and \eqref{DP.eq.3b}.\\

\begin{theorem}(Necessary and sufficient conditions)\label{nessufco}{\ \\}
 The necessary and sufficient conditions for any input distribution $\{\pi_t(a_t|b_{t-1}):t=0,1,\ldots,n\}$ to achieve the supremum of the dynamic programming recursions (\ref{DP.eq.3a}) and (\ref{DP.eq.3b})  are the following. For each $b_{n-1}\in {\mb B}_{n-1}$, there exist $V_n(b_{n-1})$ such that 
\begin{align}
 V_n(b_{n-1}) =& \sum_{b_{n}\in {\mb B}_n}\log\Big(\frac{{\bf P}_n(b_n|a_n,b_{n-1})}{{\bf P}^{\pi}_{n}(b_n|b_{n-1})}\Big){\bf P}_n(b_n|a_n,b_{n-1}), \hso   \forall{a_n}\in{\mb A}_n \hso \mbox{if} \hso \pi_n(a_n|b_{n-1})\neq{0},  \label{suff_equa_ter_con}\\
V_n(b_{n-1}) \leq & \sum_{b_{n}\in {\mb B}_n}\log\Big(\frac{{\bf P}_n(b_n|a_n,b_{n-1})}{{\bf P}^{\pi}_{n}(b_n|b_{n-1})}\Big){\bf P}_n(b_n|a_n,b_{n-1}), \hso   \forall{a_n}\in{\mb A}_n \hso \mbox{if} \hso \pi_n(a_n|b_{n-1})={0}  \label{suff_equa_ter_con_a}
\end{align}
and for each $t=n-1,n-2\ldots,1,0$ there exist  $V_t(b_{t-1})$ such that    
\begin{align}
 V_t(b_{t-1}) = & \sum_{b_{t}\in {\mb B}_t}\Big\{\log\Big(\frac{{\bf P}_t(b_t|a_t,b_{t-1})}{{\bf P}_t^{\pi}(b_t|b_{t-1})}\Big)+V_{t+1}(b_{t})\Big\}{\bf P}_t(b_t|a_t,b_{t-1}), \hso \forall {a_t}\in{\mb A}_t, \hso \mbox{if} \hso \pi_t(a_t|b_{t-1})\neq{0}, \label{suff_equa}\\
 V_t(b_{t-1}) \leq & \sum_{b_{t}\in {\mb B}_t}\Big\{\log\Big(\frac{{\bf P}_t(b_t|a_t,b_{t-1})}{{\bf P}_t^{\pi}(b_t|b_{t-1})}\Big)+V_{t+1}(b_{t})\Big\}{\bf P}_t(b_t|a_t,b_{t-1}), \hso \forall {a_t}\in{\mb A}_t, \hso \mbox{if} \hso  \pi_t(a_t|b_{t-1})={0}. \label{suff_equa_a}
\end{align}
 Moreover, $\{V_t(b_{t-1}): (t, b_{t-1})\in \{0, \ldots, n\}\times {\mb B}_{t-1}\}$ is the value function  defined by (\ref{dyn_ck1}).
\end{theorem}
\begin{proof} The derivation is given in \cite{stavrou2016sequential}.
\end{proof}

\ \

Before we proceed further, in the next remark, we  relate Theorem~\ref{nessufco} to the necessary and sufficient conditions of DMCs derived in \cite{gallager}.\\ 

\begin{remark}(Relation to necessary and sufficient conditions of DMCs)\\
\label{rem_DMC}
(a) Suppose the channel is a time-varying DMC, i.e., 
\bea
{\bf P}_t(b_t|b_{t-1}, a_t)={\bf P}_t(b_t|a_t),\hst  t=0, \ldots, n. \label{remar_eq}
\eea
 Since the optimal distribution of DMCs, which maximizes  the directed information $I(A^n \rar B^n)$ is memoryless, i.e., ${\bf P}_{A_t|A^{t-1}, B^{t-1}}(a_t|a^{t-1},b^{t-1})={\bf P}_{A_t}(a_t)\equiv \pi_t(a_t), t=0, \ldots, n$, then  (\ref{T_P_1}) reduces to 
${\bf P}_t^\pi(b_t)=\int_{{\mb A}_t}{\bf P}_t(b_t|a_t)\pi_t(a_t), t=0, \ldots, n$. By replacing in (\ref{suff_equa_ter_con}) and (\ref{suff_equa_ter_con_a}), the following quantities 
\bea
{\bf P}_t(b_t|b_{t-1}, a_t) \longmapsto {\bf P}_t(b_t|a_t), \; {\bf P}_t^\pi(b_t|b_{t-1}) \longmapsto {\bf P}_t^\pi(b_t)=\int_{A_t}{\bf P}_t(b_t|\alpha_t)\pi_t(\alpha_t), t=0, \ldots, n
\eea 
we obtain
\begin{align}
 V_n(b_{n-1})\equiv \overline{V}_n =& \sum_{b_{n}}\log\Big(\frac{{\bf P}_n(b_n|a_n)}{{\bf P}^{\pi}_{n}(b_n)}\Big){\bf P}_n(b_n|a_n), \hso   \forall{a_n}\in{\mb A}_n \hso \mbox{if} \hso \pi_n(a_n)\neq{0},  \label{suff_equa_ter_con_c}\\
V_n(b_{n-1}) \equiv \overline{V}_n \leq & \sum_{b_{n}}\log\Big(\frac{{\bf P}_n(b_n|a_n)}{{\bf P}^{\pi}_{n}(b_n)}\Big){\bf P}_n(b_n|a_n), \hso   \forall{a_n}\in{\mb A}_n \hso \mbox{if} \hso \pi_n(a_n)={0}  \label{suff_equa_ter_con_d}
\end{align}
where  $V_n(b_{n-1})=\overline{V}_n$ is a constant number, independent of $b_{n-1}$.  Moreover, for each $t$, from (\ref{suff_equa}) and (\ref{suff_equa_a}) we obtain
\begin{align}
 V_t(b_{t-1})\equiv \overline{V}_t = & \sum_{b_{t}}\log\Big(\frac{{\bf P}_t(b_t|a_t)}{{\bf P}_t^{\pi}(b_t)}\Big){\bf P}_t(b_t|a_t) + \overline{V}_{t+1}, \hso \mbox{if} \hso \forall {a_t}\in{\mb A}_t, \hso \pi_t(a_t)\neq{0}, \label{suff_equa_e}\\
 V_t(b_{t-1})\equiv \overline{V}_t \leq & \sum_{b_{t}}\log\Big(\frac{{\bf P}_t(b_t|a_t)}{{\bf P}_t^{\pi}(b_t)}\Big){\bf P}_t(b_t|a_t)+ \overline{V}_{t+1}, \hso \mbox{if} \hso \forall {a_t}\in{\mb A}_t, \hso \pi_t(a_t)={0} \label{suff_equa_f}
\end{align}
where $ V_t(b_{t-1})\equiv \overline{V}_t $ is a constant number independent of $b_{t-1}$, for $t=n-1,n-2,\ldots,1,0$. Consequently, by evaluating $V_t(b_{t-1})=\overline{V}_{t}$ at $t=0$, we obtain the following identities. 
\begin{align}
\overline{V}_0= \max_{\pi_t(a_t): t=0, \ldots, n} {\bf E}^\pi \Big\{ \sum_{t=0}^n 
\log\Big(\frac{{\bf P}_t(b_t|a_t)}{{\bf P}_t^{\pi}(b_t)}\Big)\Big\}
= \sum_{t=0}^n \max_{\pi_t(a_t)}{\bf E}^\pi \Big\{ \log\Big(\frac{{\bf P}_t(b_t|a_t)}{{\bf P}_t^{\pi}(b_t)}\Big)\Big\}. \label{remar_eq1}
\end{align}
As expected, (\ref{remar_eq1}) shows that under (\ref{remar_eq}), the sequence of nested optimization problems reduces to a sequence of non-nested optimization problems.\\
(b) Suppose the channel is time-invariant (homogeneous) DMC. In this case, $
{\bf P}_t(b_t|b_{t-1}, a_t)={\bf P}(b_t|b_{t-1}, a_t), t=0, \ldots, n$, and the equations in (a) reduce to the  single set of necessary and sufficient conditions obtained in \cite{gallager}, that is,  letting $\overline{V}= C \tri \max_{{\bf P}_A}I(A; B)$, then 
\begin{align}
\overline{V} =& \sum_{b}\log\Big(\frac{{\bf P}(b|a)}{{\bf P}^{\pi}(b)}\Big){\bf P}(b|a), \hso   \forall{a}\in{\mb A} \hso \mbox{if} \hso \pi(a)\neq{0},  \label{suff_equa_ter_con_cc}\\
\overline{V} \leq & \sum_{b}\log\Big(\frac{{\bf P}(b|a)}{{\bf P}^{\pi}(b)}\Big){\bf P}(b|a), \hso   \forall{a}\in{\mb A} \hso \mbox{if} \hso \pi(a)={0}.  \label{suff_equa_ter_con_dd}
\end{align}
\end{remark}

In view of Remark~\ref{rem_DMC},
next we identify necessary and  sufficient conditions for any optimal channel input conditional distribution, which is a solution of the dynamic programming recursions to be time-invariant, i.e., item (b), and to exhibit a non-nested property reminiscent to that of DMCs, i.e., item (c). \\

We derived such conditions based on the following definition. \\ 

\begin{definition}(Non-nested optimization)\\
Given a channel distribution $\{{\bf P}_t(b_t|b_{t-1},a_t): t =0, \ldots, n\}$, the optimization problem  $C_{A^n \rar B^n}^{FB,UMCO}$ defined by  (\ref{conn_val_fu}) is called \\
(a) non-nested if and only if the value function \eqref{dyn_ck1} satisfies the  following non-nested identity.
\begin{align}
V_t(b_{t-1})=\sum_{i=t}^n \sup_{\pi_i(a_i|b_{i-1})}{\bf E}^\pi\Big\{\log\Big(\frac{{\bf P}_i(B_i|B_{i-1},A_i)}{{\bf P}_i^\pi(B_i|B_{i-1})}\Big)\Big|B_{t-1}=b_{t-1}\Big\} \label{dyn_TC_NN_un}
\end{align}
for all $(t, b_{t-1})\in \{0, 1, \ldots, n\} \times {\mb B}_{t-1}$; \\
(b) non-nested and time-invariant if and only if the value function satisfies the following identity.
\begin{align}
V_t(b_{t-1})=(n-t+1) \sup_{\pi_t(a_t|b_{t-1})}{\bf E}^\pi\Big\{\log\Big(\frac{{\bf P}_t(B_t|B_{t-1},A_t)}{{\bf P}_t^\pi(B_t|B_{t-1})}\Big)\Big|B_{t-1}=b_{t-1}\Big\} \label{dyn_TC_NN_un_st}
\end{align}
for all $(t, b_{t-1})\in \{0, 1, \ldots, n\} \times {\mb B}_{t-1}$.\\
\end{definition}

Clearly, if we can identify conditions so that the optimization problem defined by (\ref{conn_val_fu}) is non-nested, then by evaluating the value function (\ref{dyn_TC_NN_un}) at time $t=0$ we obtain the analogue of (\ref{remar_eq1}), for channels with memory. This means that the optimal channel input distribution at each time instant is obtained by maximizing $I(A_i;B_i|B_{i-1}=b_{i-1})$ over $\pi(a_i|b_{i-1})$, for which $b_{i-1}$ is fixed. Moreover, if the optimization problem is non-nested and time invariant, then by evaluating \eqref{dyn_TC_NN_un_st} at $t=0$, we obtain
\begin{align}
C_{A^n \rar B^n}^{FB,UMCO}=(n+1)\sum_{b_{-1}\in {\mb B}_{-1}}I(A_0;B_0|B_{-1}=b_{-1}).  \label{dyn_TC_NN_un_st21}
\end{align}
Next, we state the main theorem, which generalizes the non-nested and time-invariant properties of memoryless channels given in Remark~\ref{rem_DMC}, to channels with memory. \\

\begin{theorem}(Necessary and sufficient conditions for non-nested optimization)\\
\label{non-nest_the}
(a) Consider any channel distribution $\{{\bf P}_i(b_i|b_{i-1},a_i), : i =0, \ldots, n\}$.\\ The optimization problem  $C_{A^n \rar B^n}^{FB,UMCO}$ defined by  (\ref{conn_val_fu}) is non-nested and the value function is characterized by  (\ref{dyn_TC_NN_un}) if and only if 
\begin{align}
&\mbox{there exists  constants $\big\{\overline{V}_t: t=0, \ldots, n\big\}$ such that}   \hso V_t(b_{t-1})=\overline{V}_t,\hso  \forall (t,  b_{t-1}) \in  \{0,1,\ldots, n\} \times  {\mb B}_{t-1} \nonumber \\
& \mbox{which satisfy (\ref{suff_equa_ter_con})-(\ref{suff_equa_a})}. \label{ST_D}
\end{align} 
(b) Consider any time-invariant channel distribution $\{{\bf P}(b_i|b_{i-1},a_i): i=0, \ldots, n\}$. \\
The optimization problem  $C_{A^n \rar B^n}^{FB,UMCO}$ defined by  (\ref{conn_val_fu}) is non-nested and time-invariant and the value function is characterized by  
\begin{align}
V_t(b_{t-1})= \overline{V}_t \tri& (n-t+1) \sup_{\pi^{TI}(a_i|b_{i-1})}{\bf E}^\pi\Big\{\log\Big(\frac{{\bf P}(B_i|B_{i-1},A_i)}{{\bf P}^{\pi^{TI}}(B_i|B_{i-1})}\Big)\Big|B_{i-1}=b_{i-1}\Big\}, \nonumber \\ & \forall (t, b_{t-1})\in \{0, 1, \ldots, n\} \times {\mb B}_{t-1} \label{dyn_TC_NN_un_st_TI}
\end{align}
where $\{\pi_i(a_i|b_{i-1})= \pi^{TI}(a_i|b_{i-1}): i=0, \ldots, n\}$ and $\{{\bf  P}_i^{\pi}(b_i|b_{i-1})= {\bf P}^{\pi^{TI}}(b_i|b_{i-1}): i=0, \ldots, n\}$ are time-invariant,   if and only if 
 \begin{align}
\mbox{there exists a constant $\overline{V}_n$ such that}   \hso V_n(b_{n-1})=\overline{V}_n,\hso  \forall b_{n-1} \in  {\mb B}_{n-1}  \hso \mbox{which satisfies (\ref{suff_equa_ter_con}), (\ref{suff_equa_ter_con_a})}. \label{ST_D_TI}
\end{align} 
 \end{theorem}
\begin{proof} (a) Suppose (\ref{ST_D}) holds. Then by Theorem~\ref{nessufco}, for any $t$, the optimal strategy $\pi_t(a_t|b_{t-1})$ is not affected by the future strategies $\{\pi_i(a_i|b_{i-1}): i=t+1, t+2, \ldots, n\}$ for all $t=0, 1,\ldots, n-1$. Hence, the optimization problem  $C_{A^n \rar B^n}^{FB,UMCO}$ is non-nested. Conversely, if  (\ref{dyn_TC_NN_un}) holds, since its left hand side is the value function defined by  (\ref{dyn_ck1}), then necessarily for each $t$, the value function is a constant, i.e.,  $\{V_i(b_{i-1})=\overline{V}_i: i=t+1, \ldots, n\}$, for $t=0, 1,\ldots, n-1$. In view of Theorem~\ref{nessufco}, then   (\ref{ST_D}) holds. \\
(b) This is degenerate case of part (a). Suppose (\ref{ST_D_TI}) holds and consider the necessary and sufficient conditions given in Theorem~\ref{nessufco} at time $t=n-1$. Since $V_n(b_{n-1})=\overline{V}_n, \forall b_{n-1}$, then by (\ref{suff_equa}) (and similarly for (\ref{suff_equa_a})) we have $V_{n-1}(b_{n-2})= \{\cdot \} + \overline{V}_n, \forall a_{n-1}, \pi_{n-1}(a_{n-1}|b_{n-2})$, where the term $\{\cdot\}$ is the first right hand side term in  (\ref{suff_equa}). Since the channel is time-invariant, subtracting the term $\overline{V}_n$ from both sides of the equation  (\ref{suff_equa})  (i.e., corresponding to $t=n-1$), then the resulting equations are precisely  (\ref{suff_equa_ter_con}) (and similarly for (\ref{suff_equa_ter_con_a})). Thus,  (\ref{dyn_TC_NN_un_st}) holds for $t=n-1$. To complete the derivation, we use induction, that is, we assume validity of (\ref{dyn_TC_NN_un_st}) for $t\in \{n,n-1, \ldots, i+1\}$ and we show it also holds for $t=i$. This is similar to the case $t=n-1$ hence it is omitted. Conversely, if (\ref{dyn_TC_NN_un_st}) holds, then using the time-invariant property of the channel, then necessarily (\ref{ST_D_TI}) holds (as in part (a)).
\end{proof}

Clearly, the above theorem is a generalization of Remark~\ref{rem_DMC} to channels with memory.  In Section~\ref{cabistsych} we present one example. However, additional ones can be identified by invoking the necessary and sufficient conditions of  Theorem~\ref{nessufco} and Theorem~\ref{non-nest_the}.

\subsubsection{With Transmission Cost Constraints.}
All statements of the previous section generalize to $C_{A^n \rar B^n}^{FB,UMCO}(\kappa)$ defined by  (\ref{CID_A.2.2}), where the transmission cost constraint is given by (\ref{CID_A.2}). In view of the convexity of the optimization problem, and  existence of an interior point of the constraint set $\sr{\circ}{ {\cal P}}_{[0,n]}^{FB}(\kappa)$ (i.e., Slater's condition), by Lagrange duality theorem \cite{luenberger1968optimization}, then the constraint and unconstraint problems are equivalent, that is, 
\begin{align}
C_{A^n \rar B^n}^{FB,UMCO}(\kappa) =& \inf_{s \geq 0}  \sup_{ \pi_i(a_i|b_{i-1}): i=0,1,  \ldots, n } {\bf E}^\pi_{\mu} \Big\{ \sum_{i=0}^n \Big[ \log \Big(   \frac{{\bf P}_i^\pi(b_i|b_{i-1}, a_i)}{{\bf P}_i(b_i|b_{i-1})}\Big) -s \Big(\gamma^{UM}(A_i, B_{i-1})-(n+1)\kappa\Big)\Big]\Big\}
  \label{CID_TC_1}
\end{align}
where $s\in [0, \infty)$ is the Lagrange multiplier associated with the constraint. \\
The dynamic programming recursions are obtained as follows. Let $V_t^s(b_{t-1})$ represent value function  on the future time horizon $\{t,t+1,\dots,n\}$ given output $B_{i-1}=b_{i-1}$ at time $t-1$, defined by
\begin{align}
V_t^s(b_{t-1})=\sup_{\pi_i(a_i|b_{i-1}):i=t,t+1,\dots,n}{\bf E}^\pi\Big\{\sum_{i=t}^{n}\Big[\log\Big(\frac{{\bf P}_i(B_i|B_{i-1},A_i)}{{\bf P}_i^\pi(B_i|B_{i-1})}\Big)-s \gamma^{UM}(A_i, B_{i-1})\Big] |B_{t-1}=b_{t-1}\Big\}. \label{dyn_TC}
\end{align}
The corresponding dynamic programming recursions are the following.  
\begin{align}
V_n^s(b_{n-1})=&\sup_{\pi_n(a_n|b_{n-1})}\Big\{\sum_{a_n \in {\mb A}_n}\Big[ \sum_{b_n \in {\mb B}_n}\log\Big(\frac{{\bf P}_n(b_n|b_{n-1},a_n)}{{\bf P}_n^\pi(b_n|b_{n-1})}\Big){\bf P}_n(b_n|b_{n-1},a_n)-s\gamma_n^{UM}(a_n,b_{n-1})\Big]\pi_n(a_n|b_{n-1})\Big\}\label{DP.eq.3a_TC}\\
V_t^s(b_{t-1})=&\sup_{\pi_t(a_t|b_{t-1})}\Big\{\sum_{a_t \in {\mb A}_t}\Big[\sum_{b_t \in {\mb B}_t}\Big(\log\Big(\frac{{\bf P}_t(b_t|b_{t-1},a_t)}{{\bf P}_t^\pi(b_t|b_{t-1})}\Big){\bf P}_t(b_t|b_{t-1},a_t) +V_{t+1}^s(b_{t})\Big){\bf P}_t(b_t|b_{t-1},a_t)\nonumber \\
&-s\gamma_t^{UM}(a_t,b_{t-1})\Big]   \pi_t(\alpha_t|b_{t-1})\Big\}. \label{DP.eq.3b_TC}
\end{align}
Moreover, for a fixed initial distribution ${\bf P}_{B_{-1}}(b_{-1})=\mu(b_{-1})$, then
\bea 
C_{A^n \rar B^n}^{FB,UMCO}(\kappa)=\inf_{s\geq 0}\Big\{ \sum_{b_{-1}} V_0^s(b_{-1})\mu(b_{-1})+ s(n+1)\kappa)\Big\}.
\eea
The analogues of Theorem~\ref{nessufco} and Theorem~\ref{non-nest_the} are  stated as a corollary. \\

\begin{corollary}(Necessary and sufficient conditions)\label{nessufco_TC}{\ \\}
(a) The necessary and sufficient conditions for any input distribution $\{\pi_t(a_t|b_{t-1}):t=0,1,\ldots,n\}$ to achieve the supremum of the dynamic programming recursions (\ref{DP.eq.3a_TC}) and (\ref{DP.eq.3b_TC})  are the following. \\
 For each $b_{n-1}\in {\mb B}_{n-1}$, there exist $V_n^s(b_{n-1})$ such that 
\begin{align}
 V_n^s(b_{n-1}) =& \sum_{b_{n}\in {\mb B}_n}\log\Big(\frac{{\bf P}_n(b_n|a_n,b_{n-1})}{{\bf P}^{\pi}_{n}(b_n|b_{n-1})}\Big){\bf P}_n(b_n|a_n,b_{n-1}) -s\gamma_n^{UM}(a_n,b_{n-1}), \hso   \forall{a_n}\in{\mb A}_n \hso \mbox{if} \hso \pi_n(a_n|b_{n-1})\neq{0},  \label{suff_equa_ter_con_TC}\\
V_n^s(b_{n-1}) \leq & \sum_{b_{n}\in {\mb B}_n}\log\Big(\frac{{\bf P}_n(b_n|a_n,b_{n-1})}{{\bf P}^{\pi}_{n}(b_n|b_{n-1})}\Big){\bf P}_n(b_n|a_n,b_{n-1}) -s\gamma_n^{UM}(a_n,b_{n-1}), \hso   \forall{a_n}\in{\mb A}_n \hso \mbox{if} \hso \pi_n(a_n|b_{n-1})={0}  \label{suff_equa_ter_con_a_TC}
\end{align}
and for each $t=0,1,\ldots,n-1$ there exist $V_t^s(b_{t-1})$ such that    
\begin{align}
 V_t^s(b_{t-1}) = & \sum_{b_{t} \in {\mb B}_t}\Big\{\log\Big(\frac{{\bf P}_t(b_t|a_t,b_{t-1})}{{\bf P}_t^{\pi}(b_t|b_{t-1})}\Big)+V_{t+1}(b_{t})\Big\}{\bf P}_t(b_t|a_t,b_{t-1})\nonumber \\
 & -s\gamma_t^{UM}(a_t,b_{t-1}), \hso \forall {a_t}\in{\mb A}_t, \hso \mbox{if}  \hso \pi_t(a_t|b_{t-1})\neq{0}, \label{suff_equa_TC}\\
 V_t^s(b_{t-1}) \leq & \sum_{b_{t}\in {\mb B}_t}\Big\{\log\Big(\frac{{\bf P}_t(b_t|a_t,b_{t-1})}{{\bf P}_t^{\pi}(b_t|b_{t-1})}\Big)+V_{t+1}(b_{t})\Big\}{\bf P}_t(b_t|a_t,b_{t-1})\nonumber \\
 &-s\gamma_t^{UM}(a_t,b_{t-1}), \hso \forall {a_t}\in{\mb A}_t, \hso \mbox{if}  \hso \pi_t(a_t|b_{t-1})={0}. \label{suff_equa_a_TC}
\end{align}
 Moreover, $\{V_t^s(b_{t-1}): (t, b_{t-1})\in \{0, \ldots, n\}\times {\mb B}_{t-1}\}$ is the value function  defined by (\ref{dyn_TC}).\\
(b) The optimization problem  $C_{A^n \rar B^n}^{FB,UMCO}(\kappa)$ is non-nested and the value function is characterized by  
\begin{align}
V_t^s(b_{t-1})=\sum_{i=t}^n \sup_{\pi_i(a_i|b_{i-1})}{\bf E}^\pi\Big\{\log\Big(\frac{{\bf P}_i(B_i|B_{i-1},A_i)}{{\bf P}_i^\pi(B_i|B_{i-1})}\Big)-s \gamma^{UM}(A_i, B_{i-1})|B_{i-1}=b_{i-1}\Big\} \label{dyn_TC_NN}
\end{align}
for all $(t, b_{t-1})\in \{0, 1, \ldots, n\} \times {\mb B}_{t-1}$ if and only if 
\begin{align}
&\mbox{there exists  constants $\big\{\overline{V}_t^s: t=0, \ldots, n\big\}$ such that}   \hso V_t^s(b_{t-1})=\overline{V}_t^s,\hso  \forall (t,  b_{t-1}) \in  \{0,1,\ldots, n\} \times  {\mb B}_{t-1}  \nonumber \\
& \mbox{which satisfy (\ref{suff_equa_ter_con_TC})-(\ref{suff_equa_a_TC})}. \label{ST_D1}
\end{align} 
(b) If the channel distribution is time-invariant   $\{{\bf P}(b_i|b_{i-1},a_i): i=0, \ldots, n\}$ and $\gamma_i^{UM}(\cdot,\cdot)=\gamma^{UM}(\cdot,\cdot):i=0,1,\ldots,n$, then 
the optimization problem  $C_{A^n \rar B^n}^{FB,UMCO}$ is non-nested and time-invariant, and the value function is  characterized by  
\begin{align}
V_t^s(b_{t-1})\equiv \overline{V}_t^s= (n-t+1)\sup_{\pi^{TI}(a_i|b_{i-1})}{\bf E}^{\pi^{TI}}\Big\{\log\Big(\frac{{\bf P}(B_i|B_{i-1},A_i)}{{\bf P}^\pi(B_i|B_{i-1})}\Big)-s \gamma^{UM}(A_i, B_{i-1})\Big|B_{i-1}=b_{i-1}\Big\} \label{dyn_TC_NN_SS}
\end{align}
where $\{\pi_i(a_i|b_{i-1})= \pi^{TI}(a_i|b_{i-1}): i=0, \ldots, n\}$ and $\{{\bf  P}_i^{\pi}(b_i|b_{i-1})= {\bf P}^{\pi^{TI}}(b_i|b_{i-1}): i=0, \ldots, n\}$ are time-invariant,   if and only if 
 \begin{align}
\mbox{there exists a constant $\overline{V}_n^s$ such that}   \hso V_n(b_{n-1})=\overline{V}_n,\hso  \forall b_{n-1} \in  {\mb B}_{n-1}  \hso \mbox{which satisfies (\ref{suff_equa_ter_con_TC}), (\ref{suff_equa_ter_con_a_TC})}. \label{ST_D_TI1}
\end{align} 
\end{corollary}
\begin{proof} The derivation is precisely as in Theorem~\ref{nessufco}.
\end{proof}

\subsection{Necessary and Sufficient Conditions via Dynamic Programming: The Infinite Horizon case}\label{subsec.DPandALG_IH}
In this section, we first  identify sufficient conditions the convergence of the per unit time limit of the characterization of FTFI capacity, using the ergodic theory of Markov decision with randomized strategies, and infinite  horizon dynamic programming. Then, we apply these  to derive  necessary and sufficient conditions  for any channel input distribution to maximize the  infinite horizon extremum problems 
$C_{A^\infty \rar B^\infty}^{FB,UMCO}$ and $C_{A^\infty \rar B^\infty}^{FB,UMCO}(\kappa)$.  
\par For the material of this section we make the following assumption.\\

\begin{assumptions}(Time-Invariant or homogeneous)\\
\label{TI_AS}
The  channel distribution and transmission cost function are time-invariant, and the optimal strategies are restricted to time-invariant  strategies,    i.e., 
\begin{align}
&{\bf P}_i(b_i|b_{i-1},a_i)={\bf P}(b_i|b_{i-1},a_i),\hso {\gamma}_{i}^{UM}(a_i,b_{i-1})\equiv{\gamma}^{UM}(a_i,b_{i-1}), \hso i=0, \ldots, n, \\
& \pi_i(a_i|b_{i-1})=\pi^{\infty}(a_i|b_{i-1}), \hso i=0, \ldots, n
\end{align}
and ${\mb A}_i={\mb A}, {\mb B}_i={\mb B}, i=0, \ldots, n$. 
Moreover, the initial distribution ${\bf P}_{B_{-1}}=\mu(b_{-1})$ is assumed fixed.
\end{assumptions}
By invoking Assumptions~\ref{TI_AS}, we can introduce the corresponding extremum problem as follows. For fixed initial distribution $\mu(db_{-1}) \in {\cal M}({\mb B})$, we define 
\begin{equation}
J(\pi^{\infty},\mu)\tri \liminf_{n\longrightarrow \infty}  \frac{1}{n}{\bf E}_{\mu}^{\pi^{\infty}}\Big\{\sum_{i=0}^{n-1}\log\Big(\frac{{\bf P}(B_i|B_{i-1},A_i)}{{\bf P}^{\pi^{\infty}}(B_i|B_{i-1})}\Big)\Big\}\equiv \liminf_{n\longrightarrow \infty}  \frac{1}{n}\sum_{i=0}^{n-1}I(A_i;B_i|B_{i-1})   . \label{TI_BM_a}
\end{equation}
By taking the  supremum over all channel input distributions \cite{hernandezlerma-lasserre1996} and by using the fact that the alphabet spaces are of finite cardinality,  we have the following identity. 
\begin{align}
J(\pi^{\infty,*},\mu)\tri & \sup_{ \pi^\infty(a_i|b_{i-1}): i=0,1, \ldots, } J(\pi^{\infty},\mu) \\
=& \liminf_{n\longrightarrow \infty} \sup_{\pi^\infty(a_i|b_{i-1}): i=0, \ldots, n} \frac{1}{n}{\bf E}_{\mu}^{\pi^{\infty}}\Big\{\sum_{i=0}^{n-1}\log\Big(\frac{{\bf P}(B_i|B_{i-1},A_i)}{{\bf P}^{\pi^{\infty}}(B_i|B_{i-1})}\Big)\Big\}\equiv  C_{A^\infty \rar B^\infty}^{FB,UMCO}. \label{TI_BM}
\end{align}
For abstract alphabet spaces the exchange of $\liminf$ and $\sup$ requires strong conditions \cite{hernandezlerma-lasserre1996}. 
Clearly, the above quantity $J(\pi^{\infty,*},\mu)$ depends on the initial distribution $\mu(db_{-1})$.\\
Similarly, for a fixed initial state $B_{-1}=b_{-1}$ we also have the identity
\begin{align}
\overline{J}(\pi^{\infty,*},b_{-1}) \tri & \sup_{ \pi^\infty(a_i|b_{i-1}): i=0,1, \ldots, } \liminf_{n\longrightarrow \infty} \frac{1}{n}{\bf E}_{b_{-1}}^{\pi^{\infty}}\Big\{\sum_{i=0}^{n-1}\log\Big(\frac{{\bf P}(B_i|B_{i-1},A_i)}{{\bf P}^{\pi^{\infty}}(B_i|B_{i-1})}\Big)\Big\} \label{TI_BM_IS_a} \\
=& \liminf_{n\longrightarrow \infty} \sup_{\pi^\infty(a_i|b_{i-1}): i=0, \ldots, n} \frac{1}{n}{\bf E}_{b_{-1}}^{\pi^{\infty}}\Big\{\sum_{i=0}^{n-1}\log\Big(\frac{{\bf P}(B_i|B_{i-1},A_i)}{{\bf P}^{\pi^{\infty}}(B_i|B_{i-1})}\Big)\Big\} \label{TI_BM_IS}
\end{align}
which depends on the initial state $b_{-1}$. Note that Assumptions~\ref{TI_AS} do not imply that the joint distribution of the process $\{A_0, B_0, A_1, B_1, \ldots, A_n, B_n\}$ is stationary or that the  marginal distribution of the output process $\{B_i: i=0, \ldots, n\}$ is stationary, because stationarity depends on the distribution of the initial state $B_{-1}$. However, it implies that the transition probabilities are time-invariant (i.e., homogeneous), hence,  ${\bf P}_i^{\pi^\infty}(a_i, b_i|a_{i-1}, b_{i-1})\equiv{\bf P}^{\pi^\infty}(a_i, b_i|a_{i-1}, b_{i-1}),  {\bf P}_i^{\pi^\infty}(b_i|b_{i-1})\equiv{\bf P}^{\pi^\infty}(b_i| b_{i-1}), i=0, \ldots, n$.  \\
Next, we develop the material without imposing  transmission cost constraints, because extensions to  problems with transmission cost are easily obtained by using the material of the previous section.

\subsubsection{Sufficient Condition for Asymptotic Stationarity and Ergodicity from Finite-Time Dynamic Programming Recursions}
Consider the problem of maximizing the per unit time limiting version of $C_{A^n \rar B^n}^{FB,UMCO}$, when  the strategies  are restricted to  $\{\pi^{\infty}(a_i|b_{i-1}): \hso i=0, \ldots, n\}$.  From the previous section, the  finite horizon  value function satisfies the dynamic programming equation
\begin{align}
V_t(b_{t-1})=\sup_{\pi^{\infty}(\cdot|b_{t-1})}\Big\{\sum_{a_t}\Big\{\sum_{b_t}\log\Big(\frac{{\bf P}(b_t|b_{t-1},a_t)}{{\bf P}^{\pi^{\infty}}(b_t|b_{t-1})}\Big){\bf P}(b_t|b_{t-1},a_t)
+\sum_{b_t}V_{t+1}(b_t){\bf P}(b_t|b_{t-1},a_t)\Big\}\pi^{\infty}(a_t|b_{t-1})\Big\}. \label{DP_Rev}
\end{align}
Since $b_{t-1}$ is always fixed, we let $V_t(b_{t-1})=V_t(b_{-1}), t=0, \ldots, n-1$. Since the transition probabilities are time-invariant, we can define, for simplicity, the variables $\widetilde{V}_t(b_{-1})=V_{n-t}(b_{-1}), t=0, \ldots, n-1$. Then $\{\widetilde{V}_t(\cdot): t=1, \ldots,n\}$ satisfy the following equation.
\begin{align}\label{inf.DP1}
\widetilde{V}_t(b_{-1})=&\sup_{\pi^{\infty}(\cdot|b_{-1})}\Big\{\sum_{a_0 \in {\mb A}}\Big\{\sum_{b_0 \in {\mb B}}\log\Big(\frac{{\bf P}(b_0|b_{-1},a_0)}{{\bf P}^{\pi^{\infty}}(b_0|b_{-1})}\Big){\bf P}(b_0|b_{-1},a_0)\nonumber\\
&+\sum_{b_0 \in {\mb B}}\widetilde{V}_{t-1}(b_0){\bf P}(b_0|b_{-1},a_0)\Big\}\pi^{\infty}(a_0|b_{-1})\Big\}, \hso t \in \{1, \ldots, n\}.
\end{align}
Next, we introduce a sufficient condition to test whether the per unit time limit of the solution to the   dynamic programming recursions exists and it is independent of the initial state $B_{-1}=b_{-1} \in {\mb B}$.  \\

\begin{assumptions}(Sufficient condition for convergence of dynamic programming recursions) 
\label{Ass_ST}\\
Assume that there exists a $V:{\mb B} \longmapsto {\mathbb R}$, and a $J^*\in \mathbb{R}$ such that for all $ b_{-1}\in \mathbb{B}$
\begin{equation}\label{inf.DP.assum.1a}
\lim_{t\longrightarrow \infty}\Big(\widetilde{V}_t(b_{-1})-tJ^*\Big)=V(b_{-1}).
\end{equation}
\end{assumptions}

Clearly, if Assumptions~\ref{Ass_ST} hold, then  $\lim_{t\longrightarrow \infty} \frac{1}{t} \widetilde{V}_t(b_{-1})= J^*, \forall b_{-1}\in \mathbb{B}$ (because for finite alphabet spaces, the dynamic programming operator maps bounded continous functions to bounded continuous functions) \cite{varayia86}. This means the per unit time limit of the dynamic programming recursion is independent of the initial state $b_{-1}$, which then implies $C_{A^n \rar B^n}^{FB,UMCO}$ is independent of  the choice of the initial distribution $\mu(b_{-1})$. \\

\begin{remark}(Test of asymptotic stationarity and ergodicity)\\
\label{rem_test_1}
Given any channel we can verify that the optimal channel input distribution induces asymptotic stationarity and ergodicity of the corresponding joint process $\big\{(A_i, B_i): i=0, 1, \ldots, \big\}$ by  solving  the dynamic programming recursions analytically for finite ``n'' via (\ref{DP_Rev}), and then identifying conditions on the channel parameters so that Assumptions~\ref{Ass_ST} hold. 
\end{remark}

In view of Assumptions~\ref{Ass_ST}, we have the following lemma.\\

\begin{lemma}\label{lemma_infhor}
If Assumptions~\ref{Ass_ST} hold and there exists a $\big\{\pi^{\infty,*}(a_0|b_{-1}) \in {\cal M}({\mb A}): b_{-1}  \in {\mb B}\big\}$ and a corresponding pair $\Big\{\Big(V(b_{-1}), J^*\Big): b_{-1} \in {\mb B}, \ J^*\in \mathbb{R}\Big\}$, which solves 
\begin{align}
J^*+V(b_{-1})= \sup_{\pi^\infty(a_0|b_{-1})}\Big\{\sum_{a_0}\Big\{\sum_{b_0}\log\Big(\frac{{\bf P}(b_0|b_{-1},a_0)}{{\bf P}^{\pi^\infty}(b_0|b_{-1})}\Big){\bf P}(b_0|b_{-1},a_0) 
+\sum_{b_0}V(b_0){\bf P}(b_0|b_{-1},a_0)\Big\}\pi^\infty(a_0|b_{-1}). \label{inf.DP3}
\end{align}
 then feedback capacity is given by
\begin{align}
J^* =  C_{A^\infty\rar B^\infty}^{FB, UMCO}\tri \liminf_{n\longrightarrow \infty}\frac{1}{n}\sup_{\pi^\infty(\cdot|b_{i-1})} {\bf E}^{\pi^{\infty}}\Big\{\sum_{i=0}^{n-1}\log\Big(\frac{{\bf P}(B_i|B_{i-1},A_i)}{{\bf P}^{\pi^{\infty}}(B_i|B_{i-1})}\Big)\Big\}, \hso \forall \mu(db_{-1}) \in {\cal M}({\mb B}) \label{min_a_c}
\end{align}
and moreover the value $C_{A^\infty\rar B^\infty}^{FB, UMCO}$ does not depend on the choice of the initial distribution  $\mu(db_{-1}) \in {\cal M}({\mb B})$.
\end{lemma}
\begin{proof}
See Appendix~\ref{lemma_infhor_proof}.
\end{proof}

\ \

Thus, we have two different ways to determine sufficient conditions for $J^*$ to correspond to feedback capacity; one based on Remark~\ref{rem_test_1}, and one based on Lemma~\ref{lemma_infhor}, i.e., by solving the infinite horizon dynamic programming equation (\ref{inf.DP3}).

Next, we state the necessary and sufficient conditions for any $\big\{\pi^{\infty}(a_0|b_{-1}) \in {\cal M}({\mb A}): b_{-1}  \in {\mb B}\big\}$ to be a solution of the dynamic programming equation (\ref{inf.DP3}).\\

\begin{theorem}(Infinite horizon Necessary and Sufficient conditions)\\
\label{nessufco_ΙΗ}
Suppose Assumptions~\ref{Ass_ST} hold
and there exists a $\big\{\pi^{\infty,*}(a_0|b_{-1}) \in {\cal M}({\mb A}): b_{-1}  \in {\mb B}, \ J^*\in \mathbb{R}\big\}$ and a corresponding pair $\Big\{\Big(V(b_{-1}), J^*\Big): b_{-1} \in {\mb B}\Big\}$, which solves (\ref{inf.DP3}).\\
 The necessary and sufficient conditions for any input distribution $\{\pi^{\infty}(a_0|b_{-1}) \in {\cal M}({\mb A}): b_{-1} \in {\mb B} \}$ to achieve the supremum of the dynamic programming equation \eqref{inf.DP.assum.1a}  are the following.\\  There exist $\{V(b_{-1}): b_{-1} \in {\mb B}\}$ such that 

\begin{align}
J^* + V(b_{-1}) =& \sum_{b_{0}}\Big( \log\Big(\frac{{\bf P}(b_0|a_0,b_{-1})}{{\bf P}^{\pi^{\infty}}(b_0|b_{-1})}\Big)+V(b_{0}) \Big) {\bf P}(b_0|a_0,b_{-1}) , \hso   \forall{a_0}\in{\mb A}\hso \mbox{if} \hso \pi^{\infty}(a_0|b_{-1})\neq{0},  \label{suff_equa_ter_con_TI}\\
J^*+ V(b_{-1}) \leq & \sum_{b_{0}} \Big( \log\Big(\frac{{\bf P}(b_0|a_0,b_{-1})}{{\bf P}^{\pi^{\infty}}(b_0|b_{-1})}\Big) + V(b_0)\Big){\bf P}(b_0|a_0,b_{-1}), \hso   \forall{a_0}\in{\mb A} \hso \mbox{if} \hso \pi^{\infty}(a_0|b_{-1})={0} . \label{suff_equa_ter_con_a_TI}
\end{align}

 Moreover, $\{V(b_{-1}): b_{-1}\in {\mb B}\}$ is the value function  defined by (\ref{inf.DP3}).
\end{theorem}
\begin{proof} Consider the dynamic programming equation  (\ref{inf.DP3}) and  repeat the necessary steps of the derivation of Theorem~\ref{nessufco}. A more direct approach is to use the necessary and sufficient conditions of Theorem~\ref{nessufco}, as follows. Re-writing the necessary and sufficient conditions   (\ref{suff_equa}), (\ref{suff_equa_a}) as done in  (\ref{der_IH}),  using  Assumptions~\ref{Ass_ST}, to verify  that (\ref{suff_equa_ter_con_TI}), (\ref{suff_equa_ter_con_a_TI}) are the  resulting equations.    
\end{proof}

\subsection{Sufficient Conditions for Asymptotic Stationarity and Ergodicity based on Irreducibility}
In this section  we give another set of assumptions 
based on  irreducibility of the channel output transition probability for each channel input conditional distribution. \\  
 Define 
\begin{align}
\overline{\ell}(b_{i-1},a_i)\tri&  {\bf E}^{\pi^{\infty}} \Big\{ \log\Big(\frac{{\bf P}(B_i|B_{i-1},A_i)}{{\bf P}^{\pi^{\infty}}(B_i|B_{i-1})}\Big) | B_{i-1}=b_{i-1}, A_i=a_i\Big\}\\
\equiv&  \sum_{b_i \in {\mb B}}\log\Big(\frac{{\bf P}(b_i|b_{i-1},a_i)}{{\bf P}^{\pi^{\infty}}(b_i|b_{i-1})}\Big){\bf P}(b_i|b_{i-1},a_i),  \\
\ell(b_{i-1},\pi^{\infty}(b_{i-1}))\tri & {\bf E}^{\pi^{\infty}} \Big\{ \log\Big(\frac{{\bf P}(B_i|B_{i-1},A_i)}{{\bf P}^{\pi^{\infty}}(B_i|B_{i-1})}\Big) | B_{i-1}=b_{i-1}\Big\} \equiv \sum_{a_i \in {\mb A}}\overline{\ell}(b_{i-1},a_i)\pi^{\infty}(a_i|b_{i-1}).
\end{align}
To apply standard results of the Markov Decision (MD) theory from \cite{hernandezlerma-lasserre1996, varayia86}, we introduce the following notation. 
Each element of the alphabet space ${\mb B}$ is identified  by the vector $\mathbb{B}=\{b(1),\dots,b(|\mathbb{B}|)\}$, where $|\mathbb{B}|$ is the cardinality of the set ${\mb B}$. Then we can identify    any $V:\mathbb{B}\mapsto \mathbb{R}$ with a vector in $\mathbb{R}^{|\mathbb{B}|}$. Similarly, any channel input distribution is identified with
\begin{align}
\pi^{\infty} \tri \Big\{ \pi^{\infty}(b(1)), \pi^{\infty}(b(2)), \ldots, \pi^{\infty}(b(|{\mb B}|)\Big\}\tri   \Big\{\pi^{\infty}(\cdot|b(i)) \in {\cal M}({\mb A}): b(i) \in {\mb B}\Big\}.
\end{align}
Next, we define the vector pay-off and channel output transition probability  matrix as follows.   
\begin{align}
\ell(\pi^{\infty})\tri & \Big(\ell(b(1),\pi^{\infty}(b(1)))\quad \dots \quad \ell(b(|\mathbb{B}|),\pi^{\infty}(b(|\mathbb{B}|)))\Big)^T\in \mathbb{R}^{|\mathbb{B}|}, \\
 {\bf P}(\pi^{\infty})=&\Big\{{\bf P}^{\pi^{\infty}}(b_i|b_{i-1}): (b_i, b_{i-1}) \in {\mb B} \times {\mb B}\Big\} \in \mathbb{R}^{|\mathbb{B}|\times |\mathbb{B}|}.
\end{align}
 Let $\{\mu(b(i)):i=1,2, \ldots, | {\mb B}|\}   \in \mathbb{R}^{|\mathbb{B}|}$ be defined by $\mu(b(i))={\bf P}(B_{-1}=b(i)), i=1, \ldots, |{\mb B}|$. 
 
Using the above notation we have the following main theorem.\\

\begin{theorem}(Dynamic programming equation under irreducibility)\\
\label{IRR}
Suppose Assumptions~\ref{TI_AS} holds and  for each channel input distribution $\pi^{\infty}$,  the transition probability matrix of the output process  ${\bf P}(\pi^{\infty}) \equiv \big\{ {\bf P}^{\pi^\infty}(b_0|b_{-1}): (b_0, b_{-1})\in {\mathbb B} \times {\mathbb B}\big\}$ is  irreducible.  \\
Then for any channel input distribution $\pi^{\infty}$ the expression (\ref{TI_BM_a}) is given by  
\begin{equation}
J(\pi^{\infty}, \mu)= \nu(\pi^{\infty})^T\ell(\pi^{\infty})\equiv J(\pi^\infty)
\end{equation}
i.e., it is independent of $\mu(\cdot)$, where $\nu(\pi^{\infty})$ is the unique invariant probability distribution of the channel output process $\{B_0, B_1, \ldots, \}$, which satisfies  
\bea
{\bf P}(\pi^{\infty})\nu(\pi^{\infty})=\nu(\pi^{\infty}).
\eea
If there exists a time-invariant Markov channel distribution $\pi^\infty(\cdot|\cdot)$ such that 
\begin{equation*}
J(\pi^{\infty, *})=\max_{\pi^\infty}J(\pi^\infty)
\end{equation*}
then there exists a pair $(V(\pi^{\infty,*},\cdot),J(\pi^{\infty,*}))$, $V(\pi^{\infty, *},\cdot):\mathbb{B}\mapsto \mathbb{R}^{|\mathbb{B}|}$ and $J(\pi^\infty)\in \mathbb{R}$ that is a solution of the dynamic programming equation
\begin{equation}
J(\pi^{\infty,*})+V(\pi^{\infty,*},b_{-1})=\sup_{\pi^\infty(\cdot|b_{-1})}\Big\{\ell(b_{-1},\pi^\infty(b_{-1}))+\sum_{z \in {\mb B}}V(\pi^{\infty, *},z){\bf P}^{\pi^\infty}(z|b_{-1})\Big\}.\label{IR_DP_1}
\end{equation}
Moreover, $J^* \equiv J(\pi^{\infty,*})=C_{A^\infty \rar B^\infty}^{FB,UMCO}$ satisfies (\ref{min_a_c}) and corresponds to feedback capacity.
\end{theorem}
\begin{proof} This is  shown in Appendix~\ref{IRR_proof}.
\end{proof}

\ \

We make the following comments.\\

\begin{remark}(Comments on Theorem~\ref{IRR})\\
(a) Theorem~\ref{IRR} gives sufficient conditions in terms of irreducibility of channel output transition probability matrix ${\bf P}(\pi^\infty)$  to test whether the per unit time limit of the FTFI capacity corresponds to feedback capacity. Unfortunately, it is not possible to know prior to solving the dynamic programming equation (\ref{IR_DP_1}) whether the irreducibility condition holds, because the transition probability ${\bf P}(\pi^\infty)$ is a functional of the optimal channel input distribution. A similar issue occurs in the analysis provided by Chen and Berger \cite[Lemma~2, Theorem~3]{chen-berger2005},  In view of this technicality it is more appropriate to apply the necessary and sufficient conditions of Theorem~\ref{nessufco} to determine the optimal channel input distribution and corresponding characterization of FTFI capacity, and then follow the suggestion given under Remark~\ref{rem_test_1}. 

(b) The solution of $V$ obtained from (\ref{IR_DP_1}) is unique up to an additive constant, and if $\pi^*(\cdot|b_{-1})$ attains the maximum in \eqref{IR_DP_1} for every $b_{-1}$, then $\pi^*(\cdot|\cdot)$ is an optimal channel input  distribution, and the maximum cost is $J^*$. 

(c) In specific application examples it may happen that the optimal channel input probability distribution   $\pi^{\infty,*}(\cdot|\cdot)$ induces a transition probability matrix ${\bf P}(\pi^{\infty,*})$ which is reducible, i.e., not irreducible.  For completeness, this specific case is addressed in Remark~\ref{rem_inf4}.
\end{remark}

Next, we provide an iterative algorithm to compute the optimal channel input distribution and the feedback capacity. In Section \ref{subsec.ex.cost.constrnt}, we illustrate how Algorithm \ref{alg.general.pol.iter.aver.cost} is implemented through an example.\\
\begin{algorithm}
\caption{}
\noi
\ben
\item[1)] Let $m=0$ and select an arbitrary stationary Markov channel input symbol distribution $\pi_0$.
\item[2)]  Solve the equation\een
\begin{equation}
J(\pi_m)e+V(\pi_m)=\ell(\pi_m)+V^T(\pi_m){\bf P}(\pi_m),\quad e\triangleq (1,\dots,1)\in \mathbb{R}^{|\mathbb{B}|}
\end{equation}\ben
\item[]for $J(\pi_m)\in \mathbb{R}$ and $V(\pi_m)\in \mathbb{R}^{|\mathbb{B}|}$. 
\item[3)]   Let 
\begin{equation}
\pi_{m+1}=\argmax_{\pi} \Big\{\ell(\pi)+V^T(\pi){\bf P}(\pi)\Big\}.
\end{equation} 
\item[4)] If $\pi_{m+1}=\pi_m$, let $\pi^*=\pi_m$; else let $m=m+1$ and return to step 2.  
\een
\label{alg.general.pol.iter.aver.cost}
\end{algorithm}
\noi \\

\begin{remark}\label{rem_inf4}
Theorem~\ref{IRR} and  Algorithm \ref{alg.general.pol.iter.aver.cost} pre-suppose that we  know in advance that the transition probability matrix ${\bf P}(\pi^\infty)$  of the channel output process, when evaluated at the optimal strategy $\pi^{\infty,*}(\cdot|\cdot)$ is irreducible. If irreducibility does not hold, then the dynamic programming equation \eqref{IR_DP_1} may not be sufficient to give the optimal channel input distribution and  the feedback capacity. In particular, if ${\bf P}(\pi^\infty)$ is reducible then \eqref{IR_DP_1} need not have a solution. To overcome this limitation an additional equation is added to \eqref{IR_DP_1} giving the following pair of equations.
\begin{align}
J^*(b_{-1})&=\sup_{\pi^\infty(\cdot|b_{-1})}\Big\{\int_{{\mb A}\times {\mb B}}J^*(b_0){\bf P}^{\pi^\infty}(b_0|b_{-1})\Big\}\label{inf.gen.DP.1}\\
J^*(b_{-1})+V(b_{-1})&=\sup_{\pi^\infty(\cdot|b_{-1})}\Big\{\int_{{\mb A}\times {\mb  B}}\Big\{\log\Big(\frac{{\bf P}(b_0|a_0,b_{-1})}{{\bf P}^{\pi^\infty}(b_0|b_{-1})}\Big)+V(b_0)\Big\}{\bf P}^{\pi^\infty}(b_0|b_{-1})\Big\}.\label{inf.gen.DP.2}
\end{align}
We refer to the pair \eqref{inf.gen.DP.1} and  \eqref{inf.gen.DP.2} as the generalized  dynamic programming equations. The proposed pair of dynamic programming equations completely characterize feedback capacity. 
\end{remark}

\subsection{Error exponents for the UMCO Channel with feedback}  \label{sec_error_exp_umco}
In this section, we provide bounds on the probability of error of maximum likelihood decoding,  by utilizing  the results in \cite{gallager} and \cite{permuter2006}. However, we go one step further and  show how to compute this bound, taking advantage of the structure of the capacity achieving  distribution.
\par Consider the channel $\big\{{\bf P}_i(b_i|b_{i-1}, a_i): i=0, \ldots, n\big\}$, where ${\mb B}_i={\mb B}, {\mb A}_i={\mb A}, i=0, \ldots, n$. Let ${\bf P}_{e,m}^{(n)}(b_{-1})$ denote the probability of error for an arbitrary message $m \in {\cal M}_n\tri  \big\{1,2,\ldots, M_n=\lfloor 2^{n R}\rfloor\big\}$, given the initial state $b_{-1}\in {\mb B}$. From  \cite{permuter2006}  there exists a feedback code for which the probability of error is bounded above as follows. \footnote{If the initial state is known both to the encoder and the decoder then the cardinality of the state alphabet, $|{\mb B}|$, in \eqref{prob_error1in} and \eqref{prob_error2in}  are removed [Problem~5.37, \cite{gallager}]. }
\beae
{\bf P}_{e,m}^{(n)}(b_{-1})&\leq& 4|{\mb B}| 2^{\{-n[-\rho R + {\overline F}_n(\rho) ] \}}, \hst \forall m \in  {\cal M}_n, \hso b_{-1} \in {\mb B}_{-1}, \hso  0\leq \rho \leq 1, \nms \label{prob_error1in}\\
{\overline F}_n(\rho) &\tri &\frac{-\rho\log{|{\mb B}|}}{n}+ \max_{{\bf P}_i(a_i|a^{i-1},b^{i-1}):i=0,1,\ldots,n}\left[\min_{b_{-1} \in {\mb B}}E^{\bf P}_{0,n}\left(\rho,b_{-1}\right)\right] \nms \label{prob_error2in}\\
E^{\bf P}_{0,n}\left(\rho,b_{-1}\right)&=&
-\frac{1}{n}\log{\sum_{(b_0, \ldots, b_{n-1})}\left[\sum_{(a_0, \ldots, a_{n-1})}\prod_{i=0}^{n-1}{\bf P}_i(a_i|a^{i-1},b^{i-1}){{\bf P}_i(b_i|a_i,b_{i-1})}^{\frac{1}{1+\rho}}\right]^{1+\rho}}. \nms \label{prob_error3in}
\eeae
However, by restricting the channel input distribution in \eqref{prob_error2in}, \eqref{prob_error3in}, to the set $\big\{\pi_i(da_i|b_{i-1}): i=0, \ldots, n\big\} \in \sr{\circ}{ {\cal P}}_{[0,n]}^{FB}$, the following upper bound is obtained.
\beae
{\bf P}_{e,m}^{(n)}(b_{-1})&\leq& 4|{\mb B}| 2^{\{-n[-\rho R + F_n(\rho)] \}}, \hst \forall m \in  {\cal M}_n, \hso b_{-1} \in {\mb B}_{-1}, \hso  0\leq \rho \leq 1, \nms \label{prob_error1}\\
F_n(\rho) &\tri&\frac{-\rho\log{|{\mb B}|}}{n}+ \min_{b_{-1} \in {\mb B}}E^\pi_{0,n}\left(\rho,b_{-1}\right),\hso \forall \big\{\pi_i(da_i|b_{i-1}): i=0, \ldots, n\big\} \in \sr{\circ}{ {\cal P}}_{[0,n]}^{FB},\nms \label{prob_error2}\\
E^\pi_{0,n}\left(\rho,b_{-1}\right)&=&
-\frac{1}{n}\log{\sum_{(b_0, \ldots, b_{n-1})}\left[\sum_{(a_0, \ldots, a_{n-1})}\prod_{i=0}^{n-1}\pi_i(a_i|b_{i-1}){{\bf P}_i(b_i|a_i,b_{i-1})}^{\frac{1}{1+\rho}}\right]^{1+\rho}}.\nms\label{prob_error3}
\eeae
Next, we derive simplified equations for \eqref{prob_error1} -\eqref{prob_error3}, in order to compute the bound on the probability of error.  For the rest of the analysis we view the memory of the channel on the previous output symbol as the state of the channel, defined by  $s_{i-1}\tri b_{i-1}, i=0,1,\ldots,n-1$. Then we  transform the channel to an equivalent channel of the form ${\bf P}_i(b_i|a_i,b_{i-1})={\bf P}_i(b_i|a_i,s_{i-1}), i=0, \ldots, n$. Since the state of the channel is known at the decoder, we apply the methodology used to derive Theorem 5.9.3, \cite{gallager} and  \eqref{prob_error3}, to obtain an upper bound on the probability of error, which is computationally less intensive than \eqref{prob_error1}, as follows.  At each time $i$, the channel distribution  is further transformed to 
\beae
{\bf P}_i(b_i,s_i|a_i,s_{i-1})=\left\{
  \begin{array}{l l}
   {\bf  P}_i(b_i|a_i,s_{i-1}) & \quad \text{if $s_i=b_i$}\\ \\ 
    0 & \quad \text{otherwise}.  
    \end{array} \right. \hso i=0, \ldots, n.
\label{new_channel}  
\eeae
Substituting (\ref{new_channel}) into (\ref{prob_error3}) gives the following equivalent expression.
\beae
E^\pi_{0,n}\left(\rho,b_{-1}\right)&=&
-\frac{1}{n}\log{\sum_{(b_0, \ldots, b_{n-1})}\left[\sum_{(a_0, \ldots, a_{n-1})}\prod_{i=0}^{n-1}\pi_i(a_i|b_{i-1}){{\bf P}_i(b_i|a_i,s_{i-1})}^{\frac{1}{1+\rho}}\right]^{1+\rho}}\nms \label{prob_error31}\\
&=& -\frac{1}{n}\log{\sum_{(s_0, \ldots, s_{n-1})}\sum_{(b_0, \ldots, b_{n-1})}\left[\sum_{(a_0, \ldots, a_{n-1})}\prod_{i=0}^{n-1}\pi_i(a_i|s_{i-1}){{\bf P}_i(b_i,s_i|a_i,s_{i-1})}^{\frac{1}{1+\rho}}\right]^{1+\rho}}\nms \label{prob_error32}\\
&=& -\frac{1}{n}\log{\sum_{(s_0, \ldots, s_{n-1})}\prod_{i=0}^{n-1}\sum_{b_i}\left[\sum_{a_i}\pi_i(a_i|s_{i-1}){{\bf P}_i(b_i,s_i|a_i,s_{i-1})}^{\frac{1}{1+\rho}}\right]^{1+\rho}}. \nms\label{prob_error33}
\eeae
Define the inner summations in (\ref{prob_error33}) by
\bea
\Lambda_i^\pi(s_i,s_{i-1})\tri\sum_{b_i}\left[\sum_{a_i}\pi_i(a_i|s_{i-1}){{\bf P}_i(b_i,s_i|a_i,s_{i-1})}^{\frac{1}{1+\rho}}\right]^{1+\rho}, \hso i=0, \ldots, n. \label{prob_error4}
\eea
Then, by substituting \eqref{prob_error4} in \eqref{prob_error33}, we obtain
\bea
E^\pi_{0,n}\left(\rho,b_{-1}\right)
=-\frac{1}{n}\log{\sum_{(s_0, \ldots, s_n)}\prod_{i=0}^{n-1}\Lambda_i^\pi(s_i,s_{i-1})}. \label{prob_error5}
\eea
Let $\Big\{ \Lambda_i^\pi(s_i,s_{i-1}): (s_i, s_{i-1}) \in {\mb B} \times {\mb B}\Big\}$ denote the matrix with elements identified by $\Lambda_i^\pi(s_i,s_{i-1}),  s_i=1,\ldots, |{\mb B}|, s_{i-1}=1,\ldots, |{\mb B}|$, that is, the matrix is denoted by
\bea
\left[ \Lambda_i^\pi(s_i,s_{i-1})\right] \tri \begin{bmatrix}
\Lambda_i^\pi(1,1)		& \dots	 & \Lambda_i^\pi(1,|{\mb B}|)      \\
\vdots		& \ddots & \vdots \\
\Lambda_i^\pi(|{\mb B}|,1) 	& \dots 	 & \Lambda_i^\pi(|{\mb B}|,|{\mb B}|)
\end{bmatrix}, \hst i=0, \ldots, n.
\eea
The computation of the error probability is difficult, in view of the time varying properties of the channel distribution and the channel input distribution, which implies the matrix  $\left[ \Lambda_i^\pi(s_i,s_{i-1})\right]$ is also time-varying. However, by following the derivation of equation (5.9.45) in \cite{gallager}, we derive the following bound on the probability of error for the UMCO channel with feedback.\\

\begin{theorem}(Error probability bound for maximum likelihood decoding)\\
\label{the-exp}
Suppose the channel distribution is time-invariant given by  $\big\{{\bf P}(b_i|b_{i-1}, a_i): i=0, \ldots, n\big\}$ and the probability of error defined by (\ref{prob_error1})-(\ref{prob_error3}) is evaluated at any time-invariant channel input distribution $\big\{\pi^{TI}(a_i|b_{i-1}): i=0, \ldots, n\big\}$. Then 
\itemize
\item[(i)] The matrix $\left[ \Lambda_i^\pi(s_i,s_{i-1})\right]=\left[ \Lambda^{\pi^{TI}}(s_i,s_{i-1})\right], i=0, \ldots, n$ is time-invariant.
\item[(ii)]If the time-invariant matrix $\left[ \Lambda^{\pi^{TI}}(s_i,s_{i-1})\right]$ is irreducible, then  there exists a feedback code for which the probability of error is bounded above as follows.
\bea
{\bf P}_{e,m}^{(n)}&\leq& 4|{\mb B}|\frac{v_{max}}{v_{min}} 2^{\left\{-n\left[-\rho R -\log{\lambda_{max}^{\pi^{TI}}\left(\rho\right)}  \right] \right\}}, \hst  \forall m \in  {\cal M}_n, \hso  0\leq \rho \leq 1, \label{prob_error6}
\eea
where $\lambda_{max}^{\pi^{TI}}\left(\rho\right)$ is the largest eigenvalue of the matrix $\left[ \Lambda^{\pi^{TI}}(s_i,s_{i-1})\right]$, and $v_{max}$ and $v_{min}$ are
the maximum and  minimum components, respectively, of the positive eigenvector that corresponds to the largest eigenvalue. 
\end{theorem}
\begin{proof} (i) The first statement is due to the assumptions and follows directly from  the fact that $
E^{\pi}_{0,n}\left(\rho,b_{-1}\right)=E^{\pi^{TI}}_{0,n}\left(\rho, b_{-1}\right)\tri -\frac{1}{n}\log{\sum_{(s_0, \ldots, s_n)}\prod_{i=0}^{n-1}\Lambda^{\pi^{TI}}(s_i,s_{i-1})}$.\\  (ii) For an irreducible matrix $\left[ \Lambda(s_i,s_{i-1})\right]$ with non negative components  we can apply the Frobenius theorem, to show that the following inequality holds \cite{gallager}.
\bea
\bigg{|}E^{\pi^{TI}}_{0,n}\left(\rho,b_{-1}\right)+\log{\lambda_{max}^{\pi^{TI}}\left(\rho\right)} \bigg{|}\leq\frac{1}{n}\log{\frac{v_{max}}{v_{min}}}. \label{lem_gal}
\eea
The upper bound (\ref{prob_error6}) follows from the last expression. 
\end{proof} 
%
%
%
%
Note that the probability of error in Theorem~\ref{the-exp} is independent of the initial state $b_{-1} \in {\mb B}$.  In Section~\ref{ee_bssc} we evaluate \eqref{prob_error6} of Theorem~\ref{the-exp} for a specific channel with memory.

\section{The BSSC with \& without Feedback and with \& without Transmission Cost}
\label{cabistsych}

\par In this section, we apply the main results of the previous section to the unit memory channel Binary State Symmetric Channel (BSSC) defined by
\begin{IEEEeqnarray}{l}
 {\bf P}(b_i|a_i,b_{i{-}1}) {=} \bbordermatrix{~ & 0,0 & 0,1 & 1,0 & 1,1   \cr
                  0 & \alpha & \beta & 1{-}\beta & 1{-}\alpha  \vspace*{0.5cm} \cr                   
                  1 & 1{-}\alpha & 1{-}\beta  & \beta &  \alpha \cr}, \hso  i=0, 1, 2, \ldots, n, \hso (\alpha, \beta)\in [0,1] \times [0,1].
                  \label{BSSC_1} \IEEEeqnarraynumspace
\end{IEEEeqnarray}
We show using Theorem~\ref{non-nest_the}, that the feedback capacity optimization problem is non-nested and the optimal channel input distribution is time invariant. Further, we derive explicit expressions for feedback capacity and capacity without feedback, and we show that the capacity achieving distribution and the corresponding transition probability of the channel output processes are characterized by doubly stochastic matrices. Moreover, we show that feedback does not increase capacity, and that capacity without feedback is achieved by  a first order Markov channel input distribution, which is also doubly stochastic.   
\par First we show that the  BSSC, is equivalent to a channel with state information $s_i \tri a_i\oplus b_{i-1}, i=0,1, \ldots, n$, where $\oplus$ denotes
the modulo2 addition, as depicted in Fig.~\ref{fig.equiv.model}. 
Clearly, this transformation is one to one and onto, i.e.,  for a fixed  channel input symbol value $a_i$ (respectively channel output symbol value $b_{i-1})$ then $s_i$ is uniquely determined  by the value of $b_{i-1}$ (respectively $a_i$) and vice-versa. Hence, we obtain the following equivalent representation of the BSSC.
\begin{align}
 {\bf P}(b_i|a_i,s_i=0)  =& \bbordermatrix{~ \cr
                  & \alpha & 1-\alpha \cr
                  & 1-\alpha & \alpha \cr}, \hso i=0,1, \ldots, n,   \label{state_zero}  \\
                   {\bf P}(b_i|a_i,s_i=1)=& \bbordermatrix{~ \cr
                  & \beta & 1-\beta \cr
                  & 1-\beta & \beta \cr}, \hso i=0,1, \ldots, n. \label{state_one}
                  \end{align}
The above transformation  highlights the symmetric form of the BSSC, since, for a fixed  state $s_i\in\{0,1\}$, the channel  decomposes (\ref{BSSC_1})  into  two Binary Symmetric Channels (BSC), with transition probabilities given by (\ref{state_zero}) and (\ref{state_one}), respectively. Therefore, for a fixed value of previous output symbol, $b_{i-1}$, the encoder by choosing the current input symbol, $a_i$, knows which of the two BSC's is applied at each transmission time. This decomposition motivates the name state-symmetric channel.

\par The following notation will be used in the rest of the paper. 
\begin{itemize}
\item BSSC$(\alpha,\beta)$ denotes the BSSC with transition probabilities defined by (\ref{BSSC_1});
\item BSC$(1-\alpha)$ denotes the ``state zero" channel defined by (\ref{state_zero});
\item BSC$(1-\beta)$ denotes the ``state one" channel defined by (\ref{state_one}).
\end{itemize}
\begin{figure}
    \centering
    \begin{subfigure}[b]{0.49\textwidth}
        \centering
        \includegraphics[width=\textwidth]{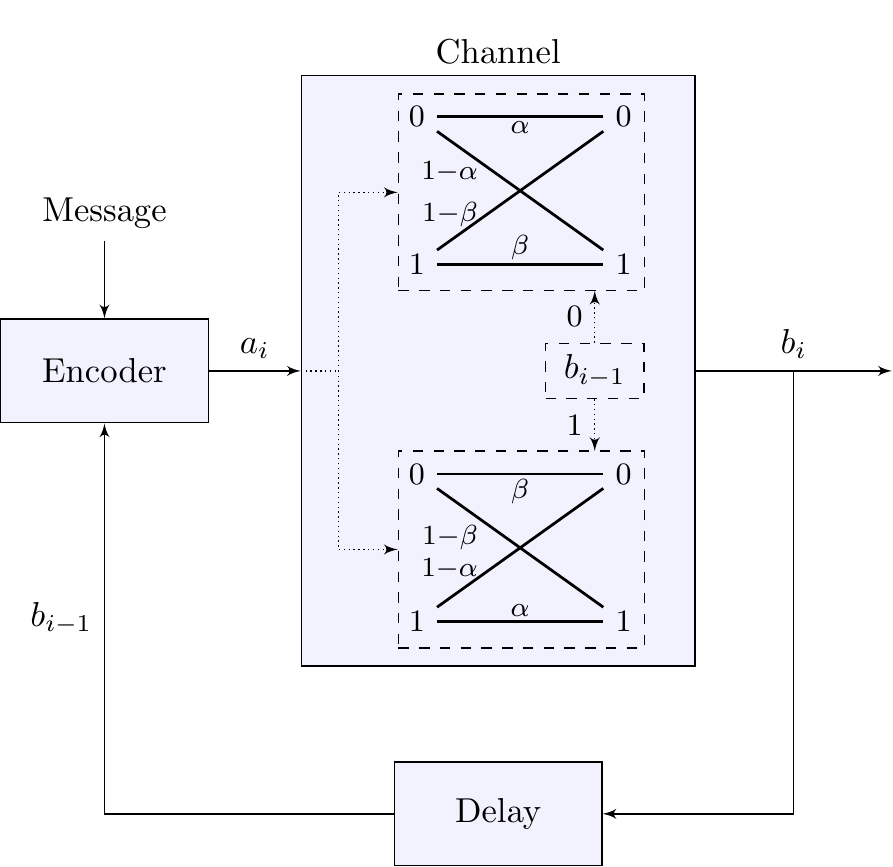}
        \caption{Previous Output STate (POST) channel.}\label{fig.DP.bssc.con.cap_a}
    \end{subfigure}
    \hfill
    \begin{subfigure}[b]{0.49\textwidth}
        \centering
        \includegraphics[width=\textwidth]{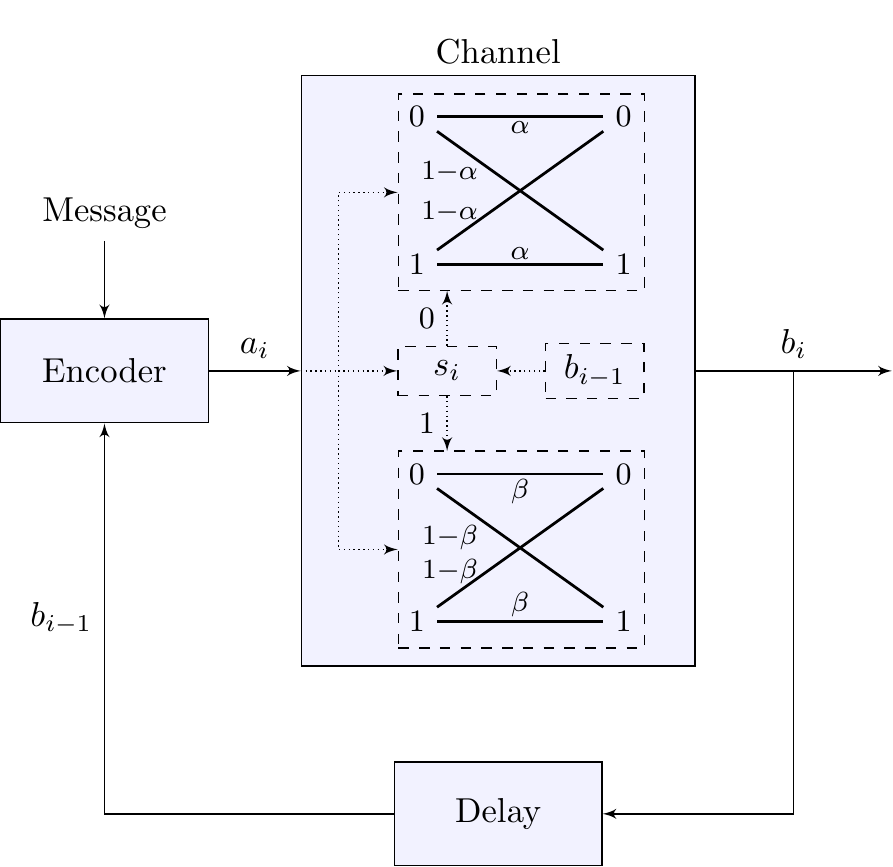}      
        \caption{Binary State Symmetric Channel (BSSC).}\label{fig.DP.bssc.con.cap_b}
    \end{subfigure}
    \caption{An equivalent model.}\label{fig.equiv.model}
\end{figure}

\par  The necessity  of imposing transmission cost constraint on the channel, is  discussed by Shannon in \cite[pp.~162--163]{Shannon59} and it is encapsulated in the following statement. `` There is a curious and provocative duality between the properties of a source with a distortion measure and those of a channel. This duality is enhanced if we consider channels in which there is a ``cost" associated with the different input letters, and it is desired to find the capacity subject to the constraint that the expected cost not exceed a certain quantity...". In \cite{kcbisit2015}, it is shown that the BSSC is in perfect duality with the Binary Symmetric Markov Source (BSMS) with respect to a transmission cost function for the channel and a fidelity constraint for the source. This is a generalization of JSCM of the discrete memoryless Bernoulli source with single letter Hamming distortion transmitted over a memoryless BSC.
\par Next, we illustrate that the cost constraint is natural when imposed on  the BSSC. The memory on the previous output symbols and its decomposable nature, allow us to impose a cost function related to the state of the channel.  
\par The physical interpretation of the transmission cost is the following. The two states of the BSSC are
\begin{itemize}
\item $s_i=0$ which is defined as the ``state zero" channel and corresponds to a BSC with crossover probability ($1-\alpha$);
\item $s_i=1$ which is defined as the ``state one" channel and corresponds to a BSC with crossover probability ($1-\beta$);
\end{itemize}
Suppose  $\alpha>\beta\geq 0.5$. 
Then the capacity of the state zero channel is greater than the capacity of the state one channel. With ``abuse" of terminology,  the state zero channel is interpreted as the ``good channel" and the state one channel is interpreted as the ``bad channel". With such interpretation   it is reasonable to impose a higher cost, 
when employing the ``good channel", and a lower cost, when employing the ``bad channel".   This policy is quantified by assigning a binary pay-off  equal to $``1"$, that is, when the the good channel is used, and  a pay-off equal to $``0"$, that is, when  the bad channel is used. \\
\begin{definition}(Binary cost function for the BSCC)\label{bin_cf}\ \\
The cost function of the BSSC satisfies
 \begin{equation}
{\gamma}(a_i,b_{i-1})=\overline{a_i\oplus b_{i-1}}  \tri\left\{
  \begin{array}{l l}
    1 & \quad \text{if $a_i=b_{i-1}$ $(s_i=0)$}\\
    0 & \quad \text{if $a_i\neq b_{i-1}$ $(s_i=1)$ }
 \label{sds} \end{array}  \right. \hso i=0,1 \ldots, n.
 \end{equation}
The average transmission cost constraint is defined by
\bea
\frac{1}{n+1}{\bf E}\left\{\sum_{i=0}^{n}{\gamma}(A_i,B_{i-1})\right\}\leq \kappa, \hso \kappa\in[0,\kappa_{max}] \label{qvcostc1_av} 
\eea
where the  letter-by-letter average transmission cost is given by
\bea
{\bf E}\big\{{\gamma}(A_i,B_{i-1})\big\}
={\bf P}_i(a_i=0,b_{i-1}=0)+{\bf P}_i(a_i=1,b_{i-1}=1)={\bf P}_i(s_i=0).
\hspace{-0.15cm}\nms \label{qvcostc1} 
\eea
\end{definition}
\par This cost function may differ, according to someone's preferences. For example, if we want to penalize the use of the ``bad" channel, we may employ the complement of the cost function (\ref{sds}). A more general cost function is 
 \begin{equation}
{\gamma}(a_i,b_{i-1}) \tri  \left\{
  \begin{array}{l l}
    \overline{\gamma} & \quad \text{if $a_i=b_{i-1}$ $(s_i=0)$}\\
    1-\overline{\gamma} & \quad \text{if $a_i\neq b_{i-1}$ $(s_i=1)$ }
 \label{sds1} \end{array}  \right. \hso i=0,1 \ldots, n
 \end{equation}
where $\overline{\gamma}\in[0,1]$. However, the binary form of the transmission cost does not downgrade the problem, since, the average cost is a linear functional, and  it can be easily upgraded to more complex forms, without affecting the proposed methodology.
\par Additional observations regarding the above formulation are given in the following remark. \\
\begin{remark}(Cost function)
\label{rem-BSSC}
\begin{enumerate}
\item If only the good channel is used, that is, ${\bf P}_i(s_i=0)=1$, then the capacity of the  BSSC is equal to zero, because  ${\bf P}(s_i=0)=1$ corresponds to channel input $a_i=b_{i-1}$, a  deterministic function for $ i=0,1,\ldots, n$ (this also follows from $I(A_i; B_i|B_{i-1})\Big|_{A_i=B_{i-1}}=0, i=0,1, \ldots,n$) . The capacity of the  BSSC is also equal to zero if only the bad channel is used ${\bf P}_{i}(s_i=0)=0$, for $ i=0, 1, \ldots, n$.
\item It is shown shortly  that the optimal channel input distribution that achieves the unconstrained capacity of the BSSC, corresponds to  a fixed occupation  of the two states. Upon introducing the transmission cost constraint, one is not allowed to use the state corresponding to the good channel beyond a certain threshold, because the overall cost of  transmission needs to be satisfied.
\item If $\beta>\alpha\geq 0.5$, then we reverse the transmission cost, while if $\alpha$ and $\beta$ are less
than $0.5$, then  flip the corresponding channel input probabilities.
\end{enumerate}
\end{remark}

\subsection{Capacity of the BSSC with feedback}
In this section, we apply Theorem~\ref{nessufco} and Theorem~\ref{non-nest_the} to calculate the closed form expressions of the capacity achieving channel input distribution,  the corresponding channel output distributions, and to show that these are time-invariant. Further, we employ these theorems to calculate the feedback capacity with and without cost constraints. \\
\subsubsection{Feedback capacity of the BSSC  without transmission cost} \ \\
\par In the next theorem we show that feedback capacity of the BSSC, without cost constraint, is given by a single letter expression and that the optimal input distribution is time invariant. \\
\begin{theorem}{(Feedback capacity and time-invariant property of the optimal distributions)}\label{op_in_out_dis_the} \ \\
Consider the BSSC defined by (\ref{BSSC_1}) with feedback, without transmission cost. Then the following hold.
\begin{itemize}
\item[(a)] The capacity achieving channel input distribution  and the corresponding channel output distribution which maximize the FTFI capacity, $C_{A^n \rar B^n}^{FB,BSSC}$, are time-invariant and given by the following expressions. 
\bea
\pi^*_i(a_{i}|b_{i{-}1})&=&\pi^{TI}(a_{i}|b_{i{-}1})= \bbordermatrix{~ \cr
                  & \nu & 1-\nu \cr
                  & 1-\nu &\nu\cr}, \hso \forall i\in\{0,1,\ldots,n\} \label{bssc_theo_oid}\\
{{\bf P}_i^{{\pi}^*}(b_{i}|b_{i-1})}&=&{{\bf P}^{TI}(b_{i}|b_{i-1})}= \bbordermatrix{~ \cr
                  & \lambda & 1-\lambda \cr
                  & 1-\lambda &\lambda\cr}, \hso  \hso \forall i\in\{0,1,\ldots,n\}\label{bssc_theo_oud}
\eea 
where
\bea
\lambda=\frac{1}{1+2^\mu}, \hso \mu=\frac{H(\beta)-H(\alpha)}{1-\alpha-\beta}, \hso \nu=\frac{1-(1-\beta)(1+2^\mu)}{(\alpha+\beta-1)(1+2^\mu)}. \label{lam_mu_v}
\eea
Moreover,
\bea
C_{A^n \rar B^n}^{FB,BSSC} &=&(n+1)\max_{\pi^{TI}(a_0|b_{-1})}I(A_0;B_0|B_{-1}=b_{i-1}), \hso \forall b_{i-1}\in\{0,1\}\label{eq_CAP_1ntc000}\\&=&(n+1)\left[H(\lambda){-}\nu H({\alpha}){-}(1{-}\nu)H({\beta})\right]\label{eq_CAP_1ntc0}.
\eea
\item[(b)] The feedback capacity is given by 
\bea
C^{FB, BSSC}_{A^\infty \rar B^\infty}&=&\max_{\pi^{TI}(a_0|b_{-1})}I(A_0;B_0|B_{-1}=b_{i-1}), \hso \forall b_{i-1}\in\{0,1\} \label{eq_CAP_1ntc}
\\&=&H(\lambda){-}\nu H({\alpha}){-}(1{-}\nu)H({\beta})\label{eq_CAP_1ntc1}
\eea 
and it is independent of the initial state.
\end{itemize}
\end{theorem}
\begin{proof}
The proof of Theorem.~\ref{op_in_out_dis_the} is given in Appendix \ref{op_in_out_dis_the_proof}.\end{proof}
\par Theorem~\ref{op_in_out_dis_the}, specifically \eqref{eq_CAP_1ntc000}, illustrates the non-nested and time-invariant property, which gives a direct connection of the BSSC and  memoryless channels. Note, that these properties hold due to the ``symmetric" form of the BSSC. As will show at the end of the current section via simulations, the time-invariant property does not hold for  general Binary Unit Memory Channel (BUMC).
\par The BSSC without cost constraint is equivalent to the POST channel investigated in \cite{asnani13j}. The authors in \cite{asnani13j} derived an expression for feedback capacity, which is equivalent to (\ref{eq_CAP_1ntc}), by using the convex hull theorem. Theorem~\ref{op_in_out_dis_the} compliments the results in \cite{asnani13j} in the sense that it provides closed form expressions of the capacity achieving distribution and the corresponding optimal channel output conditional distribution. More importantly, it shows that these distributions are time-invariant and correspond to the non-nested optimization problem \eqref{eq_CAP_1ntc000}, which is directly analogous to Shannon's two-letter capacity formulae of memoryless channels.
\par The structure of our expression  (\ref{eq_CAP_1ntc1}) provides  insight on how the occupancy of the two states affects the capacity. Recall that the state of the channel defines which of the two binary symmetric channels is in use at each time instant. Since  ${\bf P}_{S_i}(0)={\bf P}_{A_i,B_{i{-}1}}(0,0)+{\bf P}_{A_i,B_{i{-}1}}(1,1)$, then by substituting the capacity achieving input distribution we have ${\bf P}_{S_i}(0)=\nu$.
Thus, the optimal occupancy, or equivalently the optimal time sharing, among the two binary symmetric channels with crossover probabilities $\alpha,\beta$, is given by $\nu$ which is a function of the channel parameters $\alpha$ and $\beta$. This interpretation is obvious in the feedback capacity expression (\ref{eq_CAP_1ntc1}) and this expression is similar to the capacity of the memoryless binary symmetric channel. However, for the BSSC the maximization of the output process corresponds to a time invariant, first order  doubly stochastic Markov process.

%

\subsubsection{Feedback capacity of the BSSC with transmission cost} \ \\
\par Next, we consider the BSSC with transmission cost constraint defined by (\ref{qvcostc1_av}). Since $C^{FB, BSSC}_{A^n \rar B^n}(\kappa)$ is a convex optimization problem the optimal channel input conditional distribution occurs on the boundary of the constraint, i.e., for $\kappa\geq \kappa_{max}$ $C^{FB, BSSC}_{A^n \rar B^n}(\kappa)$ is constant and equal to the unconstrained capacity given in Theorem~\ref{op_in_out_dis_the}.\\

\begin{theorem}\label{cor_cos_fee}
Consider the BSSC defined by (\ref{BSSC_1}) with feedback and transmission cost constraint defined by (\ref{qvcostc1_av}). Then the following hold.
\begin{itemize}
\item[(a)] The optimal channel input distribution which corresponds to $C_{A^n\rightarrow B^n}^{FB,BSSC}(\kappa)$ and the optimal output distribution, are time-invariant and given by
\bea
\pi^*_i(a_{i}|b_{i{-}1})&=&\pi^{TI}(a_{i}|b_{i{-}1})= \bbordermatrix{~ \cr
                  & \kappa & 1-\kappa \cr
                  & 1-\kappa &\kappa\cr}, \hso \forall i=0,1,\ldots,n\label{con_inp_the}\\
{{\bf P}_i^{{\pi}^*}(b_{i}|b_{i-1})}&=&{{\bf P}^{TI}(b_{i}|b_{i-1})}= \bbordermatrix{~ \cr
                  & \bar\lambda & 1-\bar\lambda \cr
                  & 1-\bar\lambda &\bar\lambda\cr}, \hso  \hso \forall i=0,1,\ldots,n\label{con_out_the}\
\eea 
where
\beae
\bar\lambda=\alpha\kappa+(1-\kappa)(1-\beta). \label{cor_lam_bar}
\eeae
Moreover,
\beae
C_{A^n \rar B^n}^{FB,BSSC}(\kappa) &=&(n+1)\max_{\pi^{TI}(a_0|b_{-1}):\mathbb{E}\{a_i,b_{i-1}\}\leq\kappa}I(A_0;B_0|B_{-1}=b_{i-1}), \hso \forall b_{i-1}\in\{0,1\}\nms
\eeae
\item[(b)]\par The feedback capacity is given by 
\beae
C_{A^n \rar B^n}^{FB,BSSC}(\kappa)=\left\{
  \begin{array}{l l}
    H(\bar\lambda){-}\kappa H({\alpha}){-}(1{-}\kappa)H({\beta}) & \quad \text{if $\kappa \leq \kappa_{max}$}\\ \\ 
    H(\lambda){-}\kappa_{max} H({\alpha}){-}(1{-}\kappa_{max})H({\beta}) & \quad \text{if $\kappa > \kappa_{max}$}  
    \end{array} \right.
\label{eq_CAP_1}  
\eeae
where $\kappa_{max}$ is equal to $\nu$ defined by \eqref{lam_mu_v}.
\end{itemize}
\end{theorem}
This proof is similar to the proof of Theorem~\ref{op_in_out_dis_the} and is given in Appendix \ref{cor_cos_fee_proof}.\\

\begin{figure}
    \centering
    \begin{subfigure}[b]{0.47\textwidth}
        \centering
\includegraphics[width=1\linewidth]{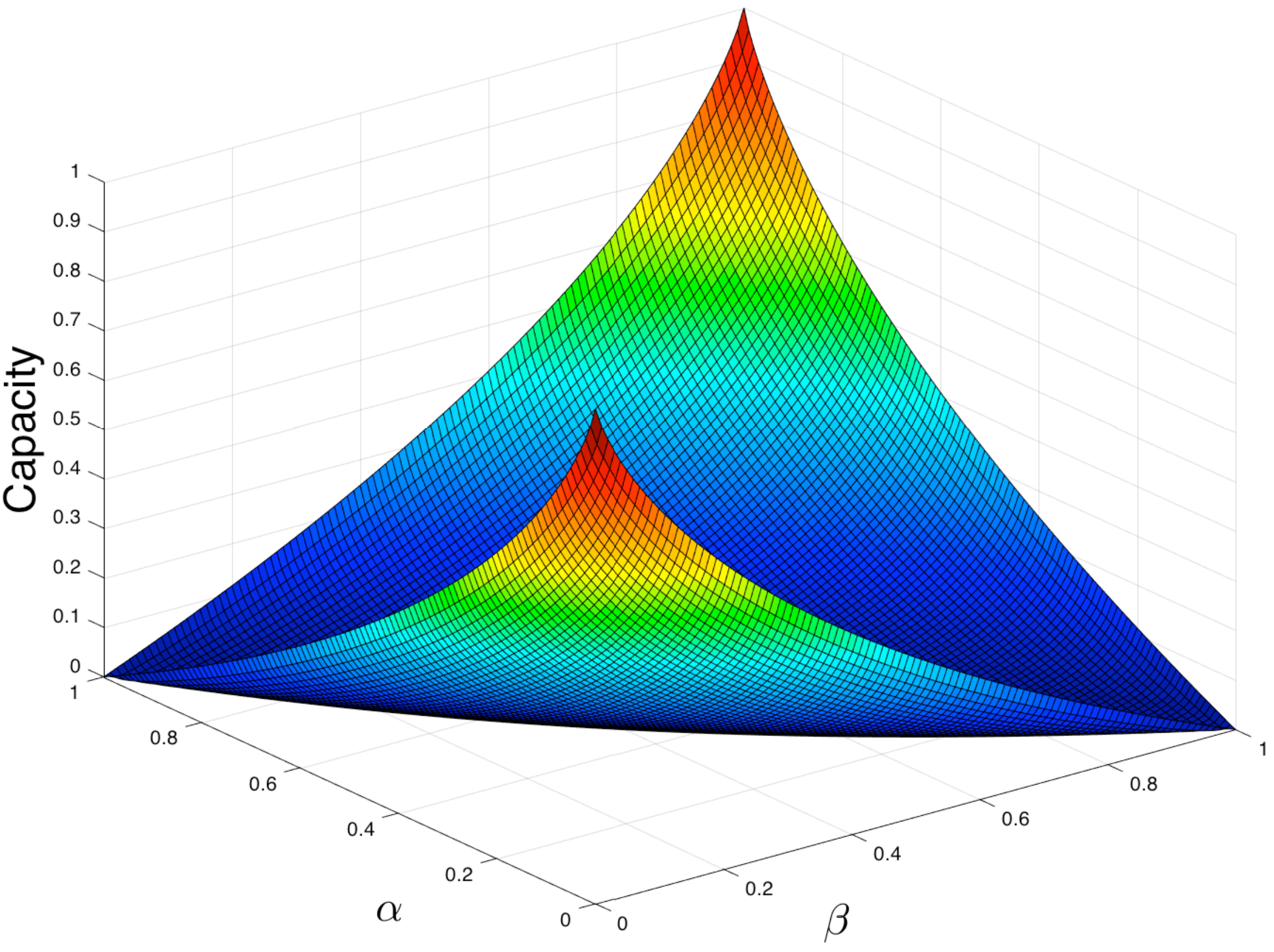}
        \caption{Unconstrained Capacity}\label{fig.DP.bssc.con.cap_a}
    \end{subfigure}
    \hfill
    \begin{subfigure}[b]{0.5\textwidth}
        \centering
        \includegraphics[width=\textwidth]{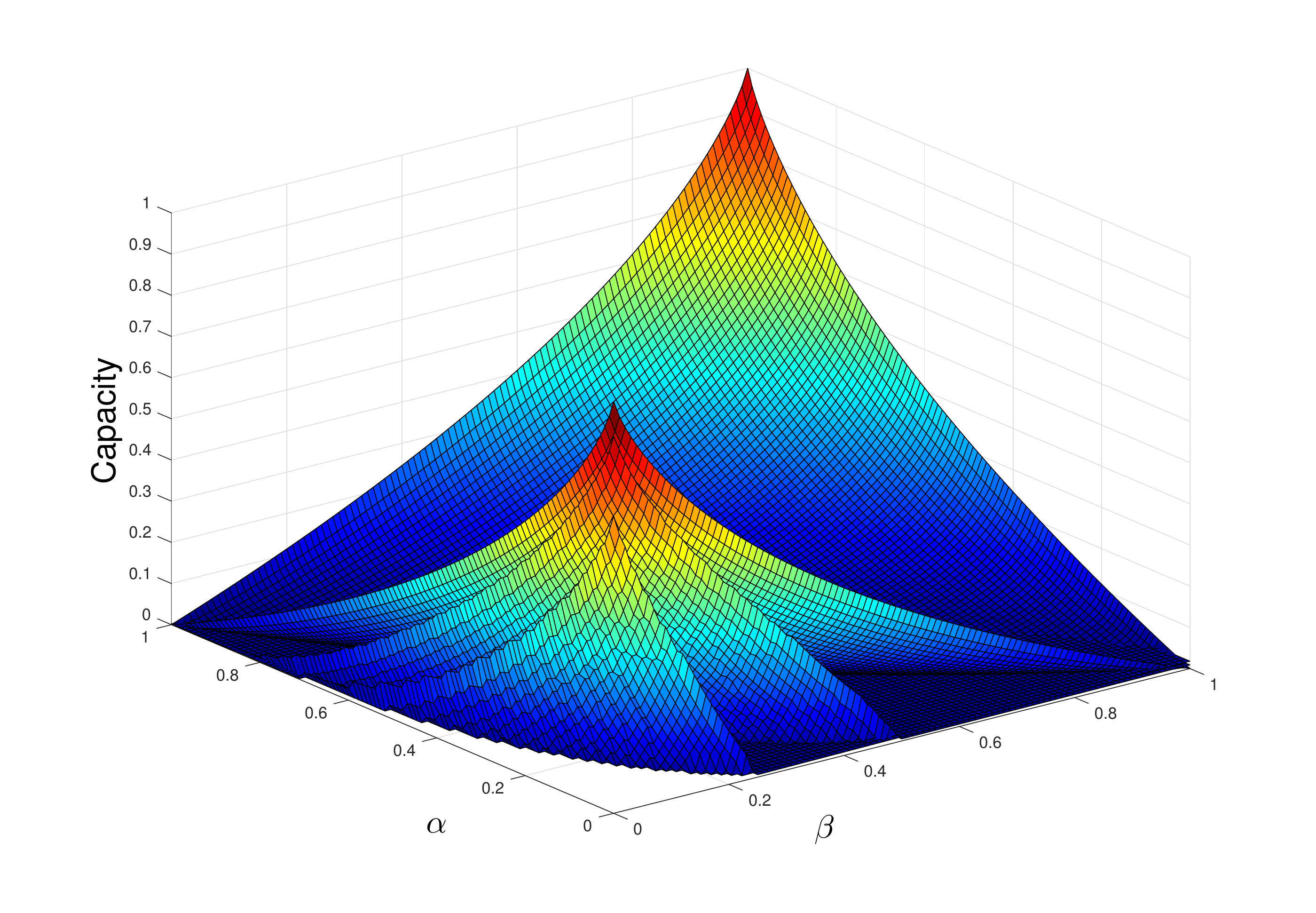}      
        \caption{Constrained Capacity.}\label{fig.DP.bssc.con.cap_b}
    \end{subfigure}
    \caption{Capacity of BSSC with feedback.}\label{fig.DP.bssc.con.cap}
\end{figure}

The unconstrained and constrained feedback  capacity of the $BSSC$ are depicted in Figure \ref{fig.DP.bssc.con.cap}. In particular, Figure \ref{fig.DP.bssc.con.cap_a} depicts the unconstrained capacity of the $BSSC$  for all possible values of the parameters $\alpha,\beta\in[0,1]$. Figure \ref{fig.DP.bssc.con.cap_b}, depicts how the transmission cost affects the capacity of the $BSSC$ for all possible values of the parameters $\alpha,\beta\in[0,1]$, and for three different choices $\kappa$. The inner plot corresponds to the unconstrained case ($\kappa=\kappa_{max}$).


\subsubsection{Error exponents for the BSSC with feedback} \label{ee_bssc} \ \\

\noi In this section we apply the results of Section~\ref{sec_error_exp_umco} to the BSSC, and we evaluate the error exponent and the probability of error, for the capacity achieving  input distribution with feedback denoted by $\pi^{TI}$ and defined by \eqref{bssc_theo_oid}.

It is straightforward to verify that evaluating \eqref{prob_error2} at the capacity achieving input distribution defined \eqref{bssc_theo_oid}, this term is independent of the initial state of the channel, and   is given by 
 \bea 
E^{\pi^{TI}}_{0,n}\left(\rho, b_{-1}\right) \equiv E^{\pi^{TI}}_{0,n}\left(\rho \right).
\eea
 Consequently, the upper bound bound on the probability of error is also independent of the initial state, and is given by
\begin{align}
{\bf P}_{e,m}^{(n)}\leq& 4|{\mb B}| 2^{\{-n[-\rho R + F_n(\rho)] \}}, \hst \forall m \in  {\cal M}_n,  \hso  0\leq \rho \leq 1.  \label{prob_error1_bssc}
\end{align} 
 Moreover, since the capacity achieving  distribution is time invariant, then $\Lambda_i^\pi(s_i,s_{i-1})=\Lambda^{\pi^{TI}}(s_i,s_{i-1}):i=0,1,\ldots,n$. Then, by substituting the time invariant capacity achieving distribution and the channel distribution in \eqref{prob_error4}, we obtain
\bea
\Lambda^{\pi^{TI}}(0,0)&=&\Lambda^{\pi^{TI}}(1,1)=\left[\nu \alpha^{\frac{1}{1+\rho}}+ (1-\nu)(1-\beta)^{\frac{1}{1+\rho}}\right]^{1+\rho} \\
\Lambda^{\pi^{TI}}(0,1)&=&\Lambda^{\pi^{TI}}(1,0)= \left[\nu(1- \alpha)^{\frac{1}{1+\rho}}+ (1-\nu)\beta^{\frac{1}{1+\rho}}\right]^{1+\rho}
\eea
The largest eigenvalue for the resulted $2\times 2$ Toeplitz matrix matrix and the ratio of the maximum and  minimum components of the positive eigenvector that correspond to the largest eigenvalue are given by
\beae
{\lambda_{max}^{\pi^{TI}}\left(\rho\right)}&=&\Lambda(0,0)+\Lambda(0,1)\nonumber\\&=&\left[\nu \alpha^{\frac{1}{1+\rho}}+ (1-\nu)(1-\beta)^{\frac{1}{1+\rho}}\right]^{1+\rho}+\left[\nu(1- \alpha)^{\frac{1}{1+\rho}}+ (1-\nu)\beta^{\frac{1}{1+\rho}}\right]^{1+\rho}.\nms \label{ee_lam_bssc} \\
\frac{v_{max}}{v_{min}}&=&1. \label{vmaxvmin}
\eeae
Substituting \eqref{ee_lam_bssc} and \eqref{vmaxvmin} in \eqref{lem_gal} we obtain
\bea
E^{\pi^{TI}}_0\left( \rho \right)&=&-log{\lambda_{max}^{\pi^{TI}}\left(\rho\right)}  \label{lem_gal1}\\
F_{\infty}\left(\rho\right)&\tri&\lim_{n\rar \infty}F_{n}\left(\rho \right)=-log{\lambda_{max}^{\pi^{TI} }\left(\rho \right)}.
\eea
Then, by definition
\bea
E^{\pi^{TI}}_r\left(R\right)\tri\max_{0\leq \rho \leq 1}\left\{F_{\infty}(\rho)-\rho{R}\right\}.\label{prob_error6ab}
\eea
Hence, the probability of error is given by
\bea
{\bf P}_{e,m}^{(n)}&\leq& 4\times|2|\times 2^{\left\{-n\left[-\rho R -log{\lambda_{max}^{\pi^{TI}}\left(\rho\right)}  \right] \right\}}. \label{prob_error6ac}
\eea
Better bounds can be obtained if both the encoder and the decoder know the initial state of the channel. In this case the cardinality of the state, $|2|$, is omitted from \eqref{prob_error6ac} [Problem~5.37, \cite{gallager}]. The error exponent and the probability of error, optimized with respect to $\rho$, are given in Fig.~\ref{fig_error}.
Obviously, even better bounds can be obtained by optimizing with respect to the channel input distribution. However, even for DMC's, the error exponent which is analogue to \eqref{prob_error6ab}, is often evaluated at the capacity achieving distribution of the ergodic capacity.

\begin{figure}
    \hspace{-0.4cm}
    \begin{subfigure}[b]{0.495\textwidth}
        \includegraphics[scale=0.52]{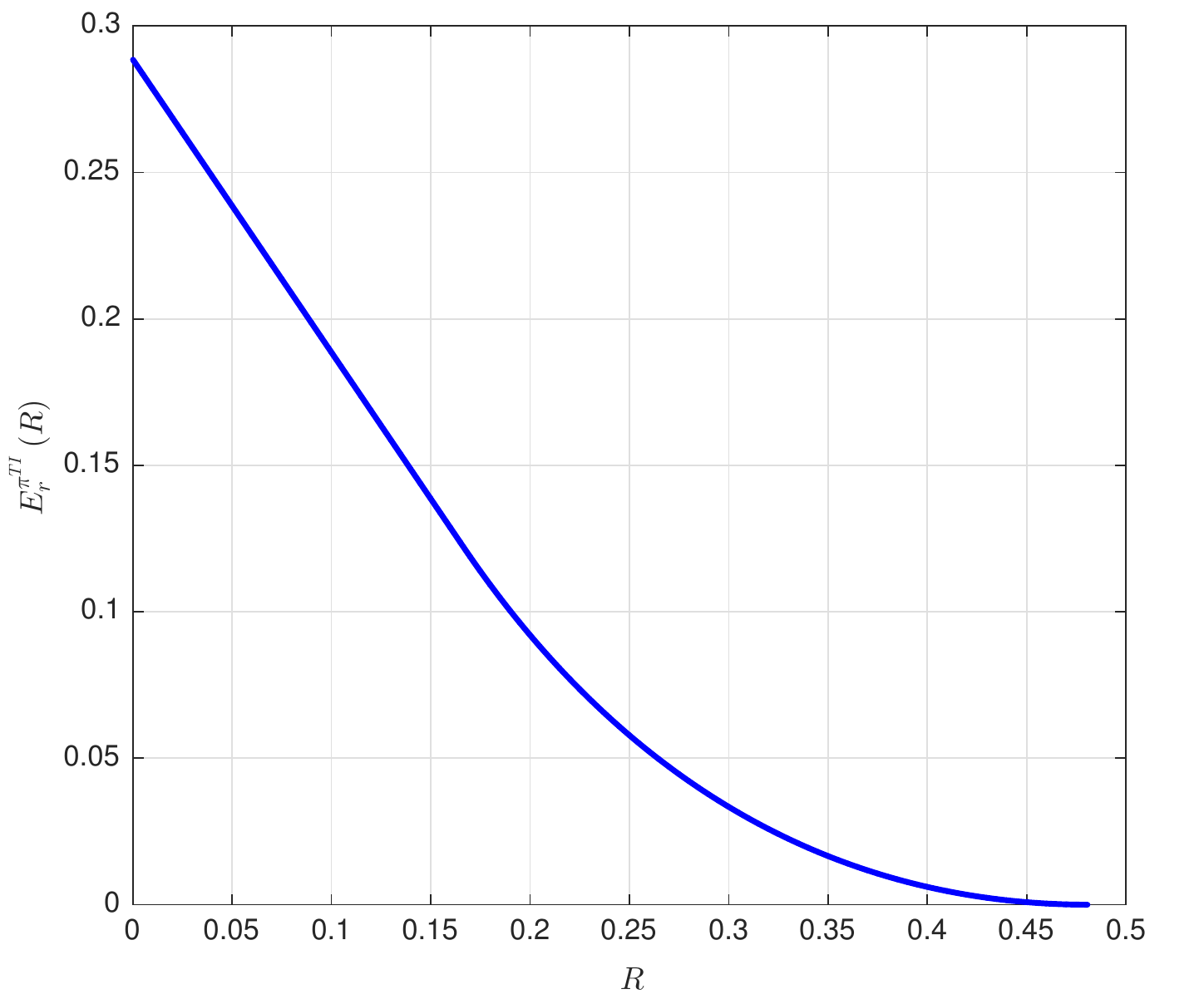}
        \caption{Error exponent.}\label{fig_error_exponent}
    \end{subfigure}
    \hfill
    \begin{subfigure}[b]{0.495\textwidth}
        \centering
        \includegraphics[scale=0.55]{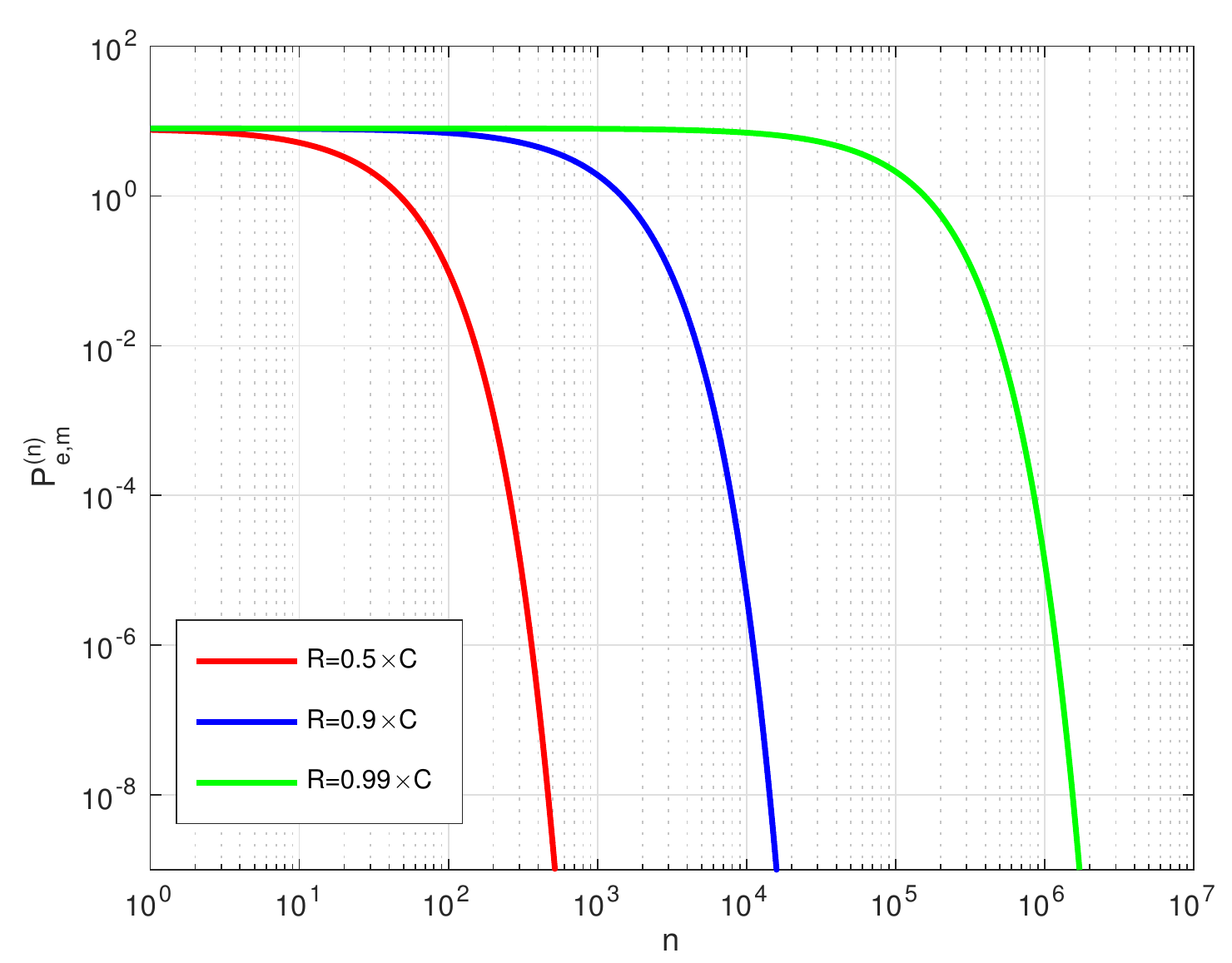}      
        \caption{Probability of error.}\label{fig_prob_error}
    \end{subfigure}
    \vspace{-0.5cm}
    \caption{Error exponent and probability of error for the BSSC with parameters $\alpha=0.95$, $\beta=0.8$.}\label{fig_error}
\end{figure}

\subsection{Capacity without feedback of the BSSC}\label{sec:cap_nf}
In this section we apply Theorem~\ref{ch4gpnf}, to show that the feedback capacity of the BSSC is achieved by a time invariant first order channel input distribution without feedback. 

\ \
\begin{theorem}(Capacity of BSSC  without  Feedback with \& without Transmission Cost)\label{ch4spnf} \\
Consider the BSSC defined by (\ref{BSSC_1}) without feedback. Then the following hold.
\begin{enumerate}
\item[(a)]  For a channel with transmission cost constraint defined by (\ref{qvcostc1_av}), the optimal channel input distribution which corresponds to $C_{A^n; B^n}^{noFB,BSSC}(\kappa)$ is time-invariant first-order Markov, and it is given by
\beae
{\bf P}^{noFB,*}_i(a_i|a^{i-1}) &=&\pi^{noFB, TI}(a_{i}|a_{i{-}1}) = \bbordermatrix{~ &  &  \cr
                   & \dfrac{1-\kappa-\sigma}{1-2\sigma} & \dfrac{\kappa-\sigma}{1-2\sigma}   \cr
                   & \dfrac{\kappa-\sigma}{1-2\sigma}   & \dfrac{1-\kappa-\sigma}{1-2\sigma} \cr}, \hso i=1,2,\ldots,n, \nms
                  \label{opima} \\ \nonumber
\eeae
where $\sigma={\alpha}{\kappa}+{\beta}({1-\kappa})$. Moreover \eqref{opima} 
induces the optimal channel input and channel output distributions ${\pi}^{TI}(a_i|b_{i-1})$ and ${{\bf P}^{TI}(b_{i}|b_{i-1})}$  of the BSSC  with feedback and transmission cost.\\
\item[(b)]  For a channel  without transmission cost (a)  holds with  $\kappa=\kappa^*$ and $\sigma=\sigma^*={\alpha}{\kappa^*}+{\beta}({1-\kappa^*})$.\\
\item[(c)]  The capacity the BSSC without feedback and transmission cost is given by
\beae
C_{A^n; B^n}^{noFB,BSSC}&=& (n+1)\max_{\pi^{noFB, TI}(a_{1}|a_{0})}I(A_1;B_1|B_{0})\label{bssc_cap_nofb}= (n+1)C^{FB, BSSC}_{A^\infty \rar B^\infty} \label{bssc_cap_nofb_new}
\eeae
and similarly, if there is a transmission cost.
\end{enumerate}
\end{theorem}

\begin{proof} (a) By applying Theorem~\ref{ch4gpnf}, it suffices to show that there exists an input distribution without feedback which induces  the capacity achieving channel input distribution with feedback, $\pi^*_i(a_{i}|b_{i{-}1})$. For the BSSC, it is  clear that,  if any input distribution without feedback induces $\pi^*_i(a_{i}|b_{i{-}1})=\pi^{TI}(a_{i}|b_{i{-}1})$ given by \eqref{con_inp_the},
then it also induces the optimal output process ${\bf P}_i^{{\pi}^{noFB,TI},*}$ $(b_{i}|b_{i-1})={{\bf P}^{TI}(b_{i}|b_{i-1})}$ given by \eqref{con_out_the}, since
\bea
{\bf P}^{TI}(b_i|b_{i-1})=\sum_{a_{i} \in \{0,1\}}{\bf P}(b_i|a_i, b_{i-1}){\pi}^{TI}(a_i|b_{i-1}).
\eea
Suppose the distribution of the initial state $b_{-1}$ is given by the stationary distribution of the output process, that is, ${\bf P}_{b_{-1}}(0)={\bf P}_{b_{-1}}(1)=0.5$.  Then, we show by induction that there exist a time invariant, first order Markov channel input distribution without feedback that induces the time invariant channel input distribution with feedback. For $i=0$, the optimal channel input distribution without feedback is equal to optimal channel input distribution with feedback, that is, ${\pi}_0^{noFB}(a_0|b_{-1})={\pi}^{TI}(a_0|b_{-1})$, and is given by \eqref{con_inp_the}, since $b_{-1}$ is the initial state known at the encoder. Therefore, the  corresponding channel output distribution with feedback,  ${\bf P}_0^{{\pi}^*}(b_0|b_{-1})$, is induced and since ${\bf P}_0^{{\pi}^*}(b_0|b_{-1})={\bf P}^{TI}(b_0|b_{-1})$ is doubly stochastic, then ${\bf P}_0^*(b_0=0)={\bf P}_0^*(b_0=1)=0.5$. 

For $i=1$, the following identities hold, in general. 
\begin{align}
{\bf P}_1(a_1|b_{0})
&=\sum_{ a_{0} \in \{0,1\} } {\bf P}_1(a_1|a_{0}, b_{0}) {\bf P}_0(a_{0}|b_{0})\nonumber\\
&=\sum_{a_{0} \in \{0,1\} }{\bf P}_1(a_1|a_{0}, b_{0})\frac{{\bf P}_0(b_{0},a_{0})}{{\bf P}_0(b_{0})}\nonumber\\
&=\sum_{a_{0} \in \{0,1\} }\frac{{\bf P}_1(a_1|a_{0}, b_{0})}{{\bf P}_0(b_{0})}\sum_{b_{-1} \mathrlap{\in \{0,1\} }}{{\bf P}(b_{0}|a_{0}, b_{-1})
 {{\bf P}_0(a_{0}|b_{-1})}{\bf P}(b_{-1}) }\label{cap_nf_pr1}
\end{align}
Next using \eqref{cap_nf_pr1}, we investigate whether there exists a first order Markov channel input distribution without feedback, ${\bf P}_1(a_1|a_{0},b_{0})={\pi}^{noFB}_1(a_1|a_{0})$, which induces the  time-invariant capacity achieving input distribution with feedback, ${\pi}^{TI}(a_1|b_{0})$, given by \eqref{con_inp_the}. Therefore, we need to determine whether the following identity holds for some ${\pi}^{noFB}_1(a_1|a_{0})$. From \eqref{cap_nf_pr1},
\begin{align}
{\pi}^{TI}(a_1|b_{0})\sr{?}{=}\sum_{a_{0} \in \{0,1\} }\frac{{\pi}_1^{noFB}(a_1|a_{0})}{{\bf P}^*_0(b_{0})}\sum_{b_{-1} \mathrlap{\in \{0,1\} }}{{\bf P}(b_{0}|a_{0}, b_{-1})
 {{\pi^{TI}}(a_{0}|b_{-1})}{\bf P}(b_{-1}) }\label{cap_nf_pr2}
\end{align}
 Note that ${\bf P}_0(a_{0}|b_{-1})={\pi^{TI}}(a_{0}|b_{-1})$ and ${\bf P}_0(b_{0})={\bf P}^*_0(b_{0})$ hold due to step $i=0$. Solving the system of resulting equations, yields that there exists a channel input distribution without feedback, defined by \eqref{opima}, that induces ${\pi}^{TI}(a_1|b_{0})$. Therefore, it also induces the time invariant optimal output distribution, ${\bf P}_i^{{\pi}^*}(b_1|b_{0})={\bf P}^{TI}(b_1|b_{0})$ given by \eqref{con_out_the}, and its corresponding optimal marginal distribution ${\bf P}^*(b_0)$.
 
Next, suppose that for time up to time  $i=j-1$, the first order Markov input distribution defined by  \eqref{opima} induces the time invariant capacity achieving distribution with feedback, $\{{\pi}^{TI}(a_{i}|b_{i-1}):i=2,3,\ldots,j-1\}$, given by \eqref{con_inp_the}, and therefore it induces, $\{{\bf P}^{TI}(b_i|b_{i-1}):i=2,3,\ldots,j-1\}$ given by \eqref{con_out_the}, and its corresponding optimal marginal distribution $\{{\bf P}^*(b_i):i=2,3,\ldots,j-1\}$. Then, at time $i=j$, the following identity holds.
\begin{align}
{\bf P}_j(a_j|b_{j-1})
&=\sum_{ a^{j-1},b^{j-2}}{\bf P}_j(a_j|a^{j-1}, b^{j-1}){\bf P}_{j-1}(a^{j-1},b^{j-2}|b_{j-1})\nonumber\\
&=\sum_{ a^{j-1},b^{j-2}}\frac{{\bf P}_j(a_j|a^{j-1}, b^{j-1})}{{\bf P}(b_{j-1})}{\bf P}(b_{j-1}|a^{j-1},b^{j-2}){\bf P}(a_{j-1}|a^{j-2}b^{j-2}){\bf P}(a^{j-2},b^{j-2})\nonumber\\
&=\sum_{ a^{j-1},b^{j-2}}\frac{{\bf P}_j(a_j|a^{j-1}, b^{j-1})}{{\bf P}^*(b_{j-1})}{\bf P}(b_{j-1}|a_{j-1},b_{j-2}){\pi}^{TI}(a_{j-1}|b_{j-2}){\bf P}^*(a^{j-2},b^{j-2}). \label{cap_nf_pr1nn}
\end{align}
The last equality holds since the distributions ${\bf P}^*(b_{j-1})$, ${\pi}^{TI}(a_{j-1}|b_{j-2})$, ${\bf P}^*(a^{j-2},b^{j-2})$ were induced from the previous steps $i=0,1,\ldots,j-1$. Subsequently, we investigate whether there exists a first order Markov channel input distribution, ${\bf P}_j(a_j|a^{j-1}, b^{j-1})={\pi}^{noFB}_j(a_j|a_{j-1})$, that satisfies \eqref{cap_nf_pr1nn}. That is, 
\begin{align}
{\pi}^*(a_{j}|b_{j-1})&\sr{?}{=}\sum_{ a^{j-1},b^{j-2}}\frac{{\pi}^{noFB}_j(a_j|a_{j-1})}{{\bf P}^*(b_{j-1})}{\bf P}(b_{j-1}|a_{j-1},b_{j-2}){\pi}^{TI}(a_{j-1}|b_{j-2}){\bf P}^*(a^{j-2},b^{j-2})\nonumber\\
&=\sum_{ a_{j-1}}\frac{{\pi}^{noFB}_j(a_j|a_{j-1})}{{\bf P}^*(b_{j-1})}\sum_{b_{j-2}}{\bf P}(b_{j-1}|a_{j-1},b_{j-2}){\pi}^{TI}(a_{j-1}|b_{j-2})\sum_{ a^{j-2},b^{j-3}}{\bf P}^*(a^{j-2},b^{j-2})\nonumber\\
&=\sum_{ a_{j-1}}\frac{{\pi}^{noFB}_j(a_j|a_{j-1})}{{\bf P}^*(b_{j-1})}\sum_{b_{j-2}}{\bf P}(b_{j-1}|a_{j-1},b_{j-2}){\pi}^*(a_{j-1}|b_{j-2}){\bf P}^*(b_{j-2})\label{cap_nf_pr1nn1}
\end{align}
Solving, the system of equation resulting from equation \eqref{cap_nf_pr1nn1}, yields the time-invariant first order Markov input distribution defined by \eqref{opima}. Since, the time invariant first order Markov channel input distribution without feedback defined by \eqref{opima}, induces the optimal channel input distribution with feedback $\forall i=1,2,\ldots,j$, then it is the time invariant capacity achieving input distribution without feedback. \\
(b) Holds since for  the BSSC  without transmission cost $\kappa=\kappa^*$, and therefore $\sigma=\sigma^*={\alpha}{\kappa^*}+{\beta}({1-\kappa^*})$.\\
{(c)} Since, $\{{\pi}^{noFB,TI}(a_i|a_{i-1})\equiv {\pi}^{noFB,TI}(a_1|a_{0}):i=1,2,\ldots,n \}$ induces  $\{{\pi}^{TI}(a_i|b_{i-1}):i=1,2,\ldots,n \}$ given by \eqref{con_inp_the}, and  $\{{\bf P}^{TI}(b_i|b_{i-1})i=1,2,\ldots,n \}$ given by \eqref{con_out_the}, then the channel capacity without feedback and transmission cost is given by \eqref{bssc_cap_nofb}. Similarly, for the constrained capacity we have $C_{A^n; B^n}^{noFB,BSSC}(\kappa)= C^{FB, BSSC}_{A^\infty \rar B^\infty}(\kappa)$.
\end{proof}

\subsection{Special cases of the BSSC}
\subsubsection{ Memoryless BBSC (${\alpha}={\beta}=1-\epsilon, \ \epsilon\neq 0.5$)} \ 
\par Consider the trivial case where ${\alpha}={\beta} \tri 1-\epsilon$. Then, given the state $s_i=a_i\oplus b_{i-1}$, the BSSC degenerates to the Discrete Memoryless - Binary Symmetric Channel (DM-BSC) with cross
over probability $\epsilon$. By employing \eqref{bssc_theo_oid}-\eqref{eq_CAP_1ntc} and \eqref{opima}, then  $\mu=0$ and $ \nu=\lambda=0.5$, the capacity achieving input distribution and the corresponding output distribution are memoryless and uniformly distributed, and the  capacity expression reduces to 
\bea
C^{DM-BSC}=H((1-\epsilon)(1-{0.5})+\epsilon{0.5})-{0.5}H(\epsilon)-{0.5}H(\epsilon)
=1-H(\epsilon) . \nonumber
\eea
\par This are the known results of the memoryless BSC.
\subsubsection{ Best and Worst BBSC (${\alpha}=1, \ {\beta}=0.5$)} \
\par Consider the case  ${\alpha}=1$ and ${\beta}=0.5$. This channel decomposes to a  noiseless BSC channel with crossover probability $0$ if $s_i=a_i\oplus b_{i-1}=0$, and to  a  noisy BSC channel with crossover probability $0.5$ if $s_i=a_i\oplus b_{i-1}=1$. By invoking \eqref{bssc_theo_oid}-\eqref{eq_CAP_1ntc}, then  $\nu=0.6$,
$\lambda=0.8$, the channel capacity is equal to
\bea
C^{FB, BSSC}\Big|_{\alpha=1, \beta=0.5}=C^{noFB, BSSC}\Big|_{\alpha=1, \beta=0.5}=H(0.2)-{0.6}H(1)-{0.4}H(0.5)=0.3219 \nonumber
\eea
the optimal channel  input distributions with and without feedback are given by
\vspace{-0.6cm}
\bea
\pi^{TI}(a_i|b_{i-1}) = \bbordermatrix{~ \cr
                  & 0.6 & 0.4 \cr
                  & 0.4 & 0.6 \cr}, \hso 
{\pi}^{noFB, TI}(a_i|a_{i-1}) = \bbordermatrix{~  \cr
                  & 2/3 & 1/3 \cr
                  & 1/3 & 2/3 \cr}\nonumber
\eea
and the optimal channel output distribution for both is given by
\bea
{\bf P}^{TI}(b_i|b_{i-1}) = \bbordermatrix{~ \cr
                  & 0.8 & 0.2 \cr
                  & 0.2 & 0.8 \cr} . \nonumber
\eea
This completes the analysis of degenerate BSSC.

\subsection{Capacity of the Binary Input Binary Output - Unit Memory Channel Output (BIBO-UMCO) channel with feedback}\label{subsec.DP.example}

\begin{figure}
    \centering
    \begin{subfigure}[b]{0.49\textwidth}
        \centering
        \includegraphics[width=\textwidth]{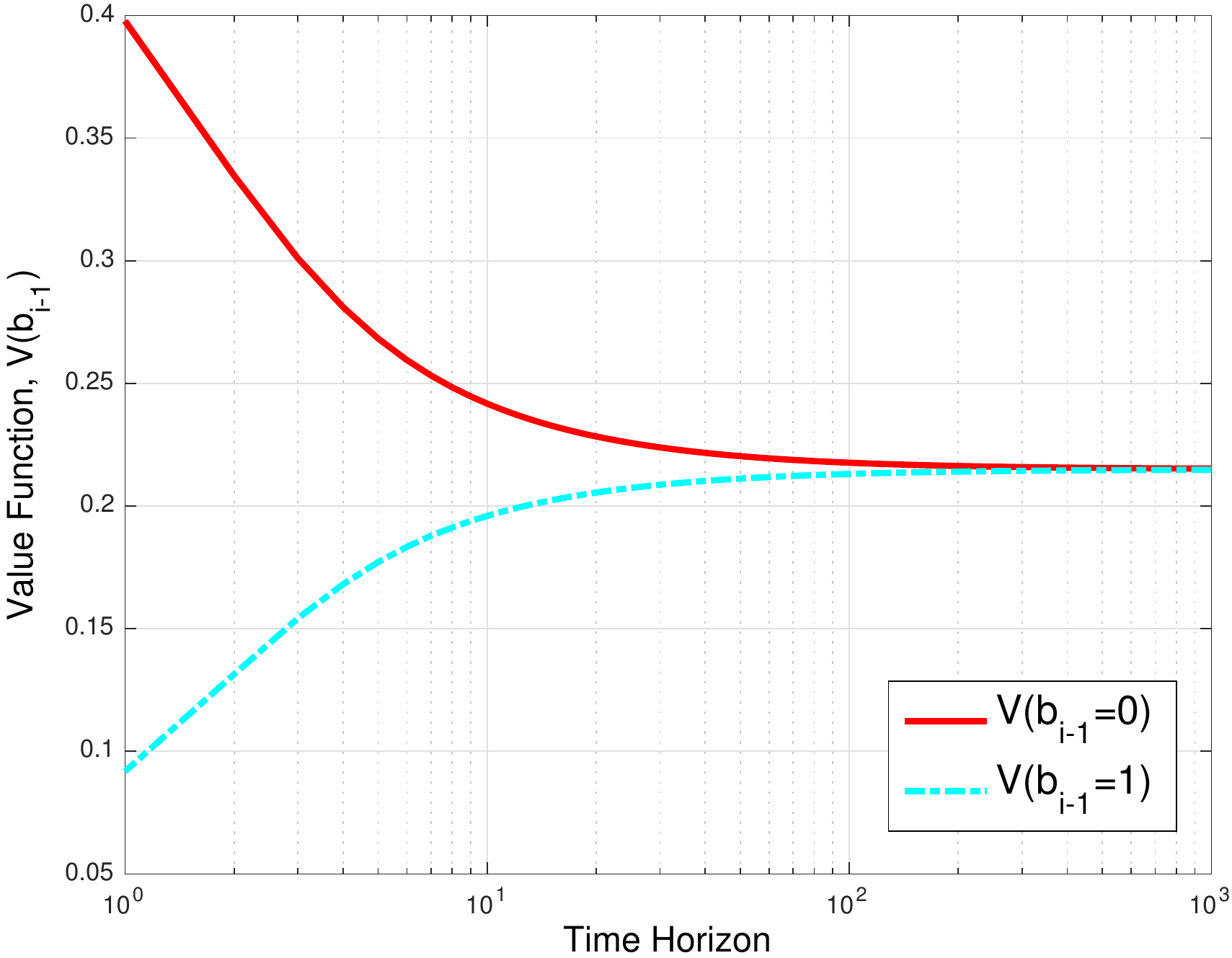}
        \caption{Value Function.}\label{fig.DP.bsnsc.a}
    \end{subfigure}
    \hfill
    \begin{subfigure}[b]{0.49\textwidth}
        \centering
        \includegraphics[width=\textwidth]{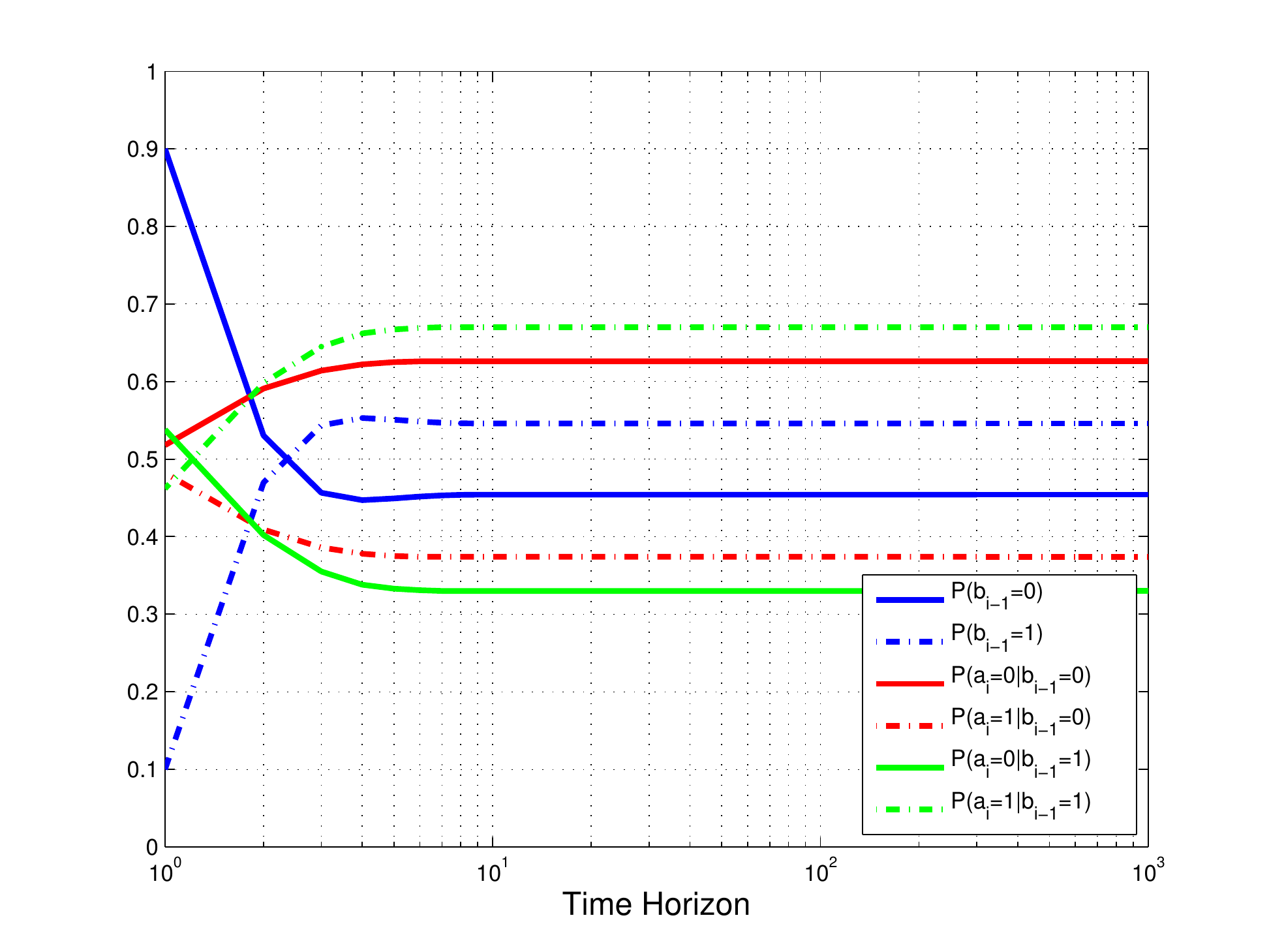}      
        \caption{Optimal Input and Output Distributions.}\label{fig.DP.bsnsc.b}
    \end{subfigure}
    \caption{Unconstrained $BIBO-UMCO$ channel with parameters $\alpha_1=0.9$, $\alpha_2=0.2$, $\alpha_3=0.1$ and $\alpha_4=0.4$.}\label{fig.DP.bsnsc}
\end{figure}
\begin{figure}[!h]
    \centering
    \begin{subfigure}[b]{0.49\textwidth}
        \centering
        \includegraphics[width=\textwidth]{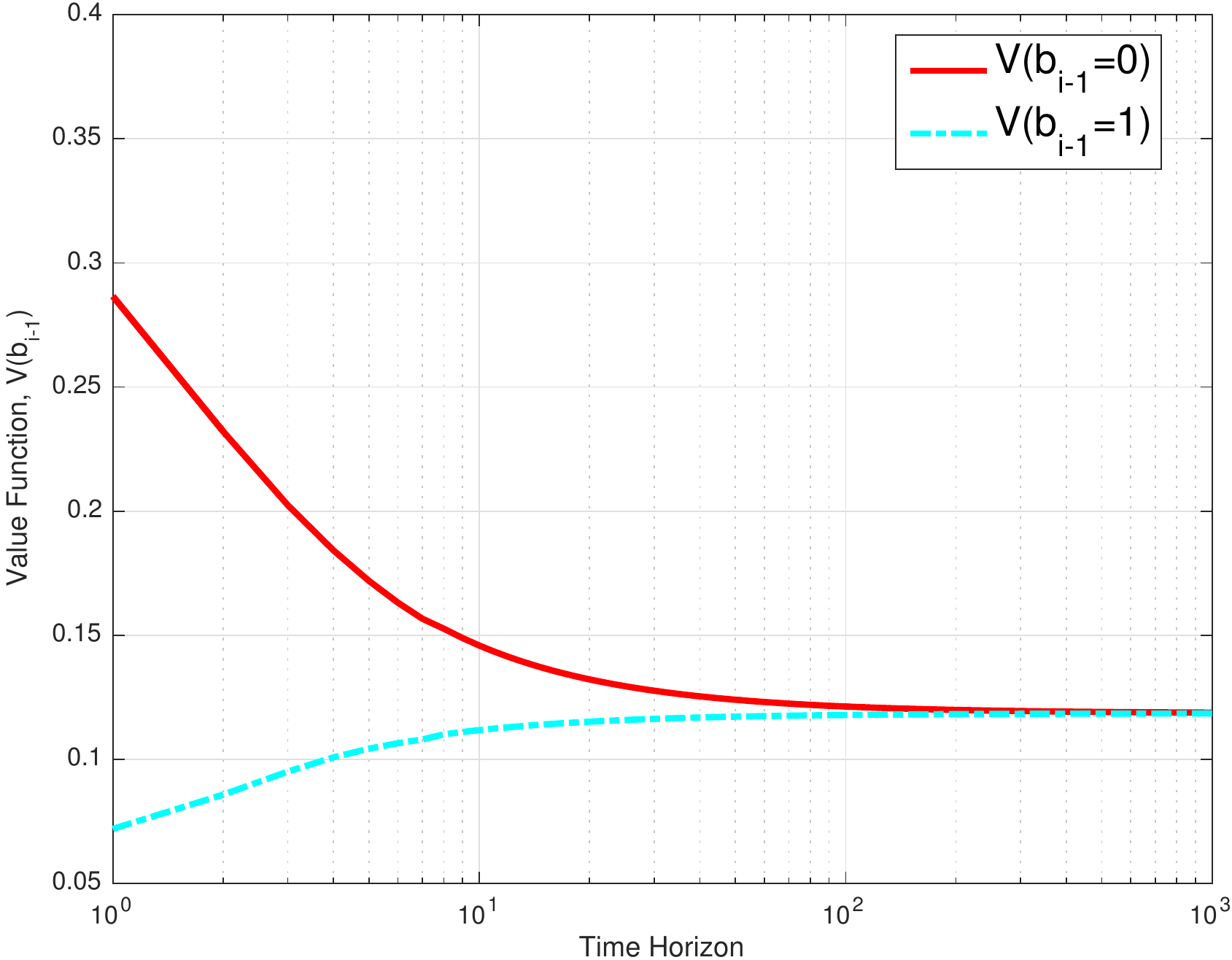}
        \caption{Value Function.}\label{fig.DP.bsnsc.a1}
    \end{subfigure}
    \hfill
    \begin{subfigure}[b]{0.49\textwidth}
        \centering
        \includegraphics[width=.95\textwidth]{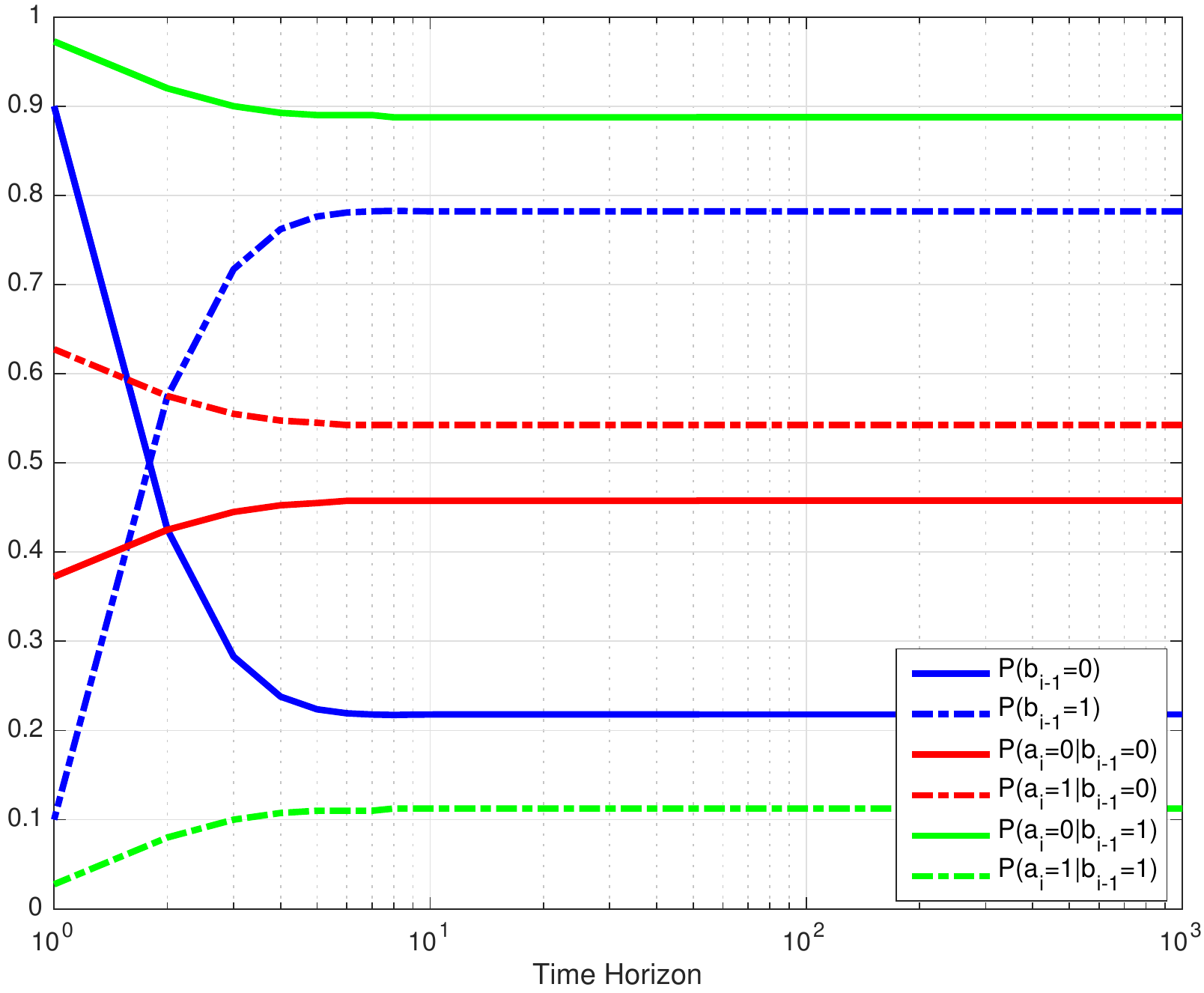}      
        \caption{Optimal Input and Output Distributions.}\label{fig.DP.bsnsc.b1}
    \end{subfigure}
    \caption{Constrained  $BIBO-UMCO$ channel with parameters $\alpha_1=0.9$, $\alpha_2=0.2$, $\alpha_3=0.1$, $\alpha_4=0.4$ and $k=0.1877$.}\label{fig.DP.bsnsccon1}
\end{figure}

In this section, we employ the dynamic programming results obtained in Section \ref{subsec.DPandALG}, to calculate the feedback capacity of $BIBO-UMCO$ channel, denoted by 
\begin{equation}\label{trans.prob.matr.bsnsc}
 {\bf P}(b_i|b_{i-1},a_i)=\bordermatrix{&00&01&10&11\cr
            0&\alpha_1&\alpha_2&\alpha_3&\alpha_4\cr
            1&1-\alpha_1&1-\alpha_2&1-\alpha_3&1-\alpha_4\cr},\quad i=0,\dots,n
\end{equation}
with and without transmission cost. In addition, we calculate  the  capacity achieving input distributions with feedback and the respective optimal output distributions.\\
\subsubsection{Without cost constraint} \label{subsec.ex.cost.constrnt} Consider the $BIBO-UMCO$ channel \eqref{trans.prob.matr.bsnsc} with parameters $\alpha_1=0.9$, $\alpha_2=0.2$, $\alpha_3=0.1$ and $\alpha_4=0.4$. By employing dynamic programming equations \eqref{DP.eq.3a}-\eqref{DP.eq.3b} the convergence of the value functions without transmission cost, and the convergence of the optimal input distributions with feedback and the corresponding output distributions are depicted in Figures \ref{fig.DP.bsnsc.a} and \ref{fig.DP.bsnsc.b}, respectively. To characterize the feedback capacity and the capacity achieving input distribution of  the BIBO-UMCO channel we employ Algorithm \ref{alg.general.pol.iter.aver.cost}, which yields the following results.
\begin{equation*}
{\pi}^{\infty}(a_i|b_{i-1}) = \bbordermatrix{~ \cr
                  & 0.626 & 0.33 \cr
                  & 0.374 & 0.67 \cr}, \qquad C^{FB,BIBO-UMCO}=0.215 \ \text{bits/per channel use.}
\end{equation*}

%
%
%

\subsubsection{With cost constraint.} Consider the $BIBO-UMCO$ channel \eqref{trans.prob.matr.bsnsc} with parameters $\alpha_1=0.9$, $\alpha_2=0.2$, $\alpha_3=0.1$, $\alpha_4=0.4$ and  $k=0.1877$. By employing dynamic programming equations \eqref{DP.eq.3a_TC}-\eqref{DP.eq.3b_TC}, Figures \ref{fig.DP.bsnsc.a1} and \ref{fig.DP.bsnsc.b1} depict the convergence of the value functions with transmission cost, and the convergence of the optimal channel input distributions with feedback and the corresponding output distributions.

\section{Conclusions}
\label{sec_con}
\par We apply the dynamic programming recursions and necessary and sufficient conditions for any channel input conditional distribution to achieve capacity, to identify necessary and sufficient conditions such that the nested optimization problem $C_{A^n \rar B^n}^{FB}$ reduces to a non-nested optimization problem. This gives rise to the single letter characterization of feedback capacity. The  methodology  can be easily generalized {\bf to} channels that have finite memory on the previous outputs.
\par These results are applied to the BSSC with feedback, with and without cost constraint, to calculate the feedback capacity, the capacity achieving input distribution, and the corresponding output distribution. One of the fascinating results is that feedback capacity is characterized by a single letter expression that is precisely analogous to the single letter characterization of capacity of DMCs.  Additionally, we show that a first order Markov channel input distribution without feedback  achieves  feedback capacity. We also derive an upper bound on the error probability of maximum likelihood decoding. 

\appendices

\section{Proof of Lemma.~\ref{lemma_infhor}}\label{lemma_infhor_proof}
We can re-write \eqref{inf.DP1} as follows. 
\begin{align}
\tilde{V}_t(b_{-1})+\frac{1}{t}\tilde{V}_t(b_{-1})=&\sup_{\pi^{\infty}(\cdot|b_{-1})}\Big\{\sum_{a_0}\Big\{\sum_{b_0}\log\Big(\frac{{\bf P}(b_0|b_{-1},a_0)}{{\bf P}^{\pi^{\infty}}(b_0|b_{-1})}\Big){\bf P}(b_0|b_{-1},a_0)\label{inf.DP2}\\
&+\sum_{b_0}\Big(\tilde{V}_{t-1}(b_0)+\frac{1}{t}\tilde{V}_t(b_{-1})\Big){\bf P}(b_0|b_{-1},a_0)\Big\}\pi^{\infty}(a_0|b_{-1})\Big\}.\label{inf.DP2_neq1}
\end{align}
Assumptions~\ref{Ass_ST},  imply that 
\begin{equation}
\lim_{t\longrightarrow\infty}\frac{1}{t}\tilde{V}_t(b_{-1})=J^*,\quad \forall b_{-1}\in \mathbb{B} \label{inf.DP.assum.1b}
\end{equation}
and that the limit does not depend on $b_{-1}\in {\mb B}$. Moreover, under Assumption~\ref{Ass_ST}, \eqref{inf.DP.assum.1b}, taking the limit of both sides of \eqref{inf.DP2_neq1}, the following dynamic programming equation is obtained. 
\begin{align}
J^*+v(b_{-1})=&\lim_{t\longrightarrow\infty}\Big\{\frac{1}{t}\tilde{V}_t(b_{-1})+\Big(\tilde{V}_t(b_{-1})-tJ^*\Big)\Big\}\\
\overset{(a)}=&\lim_{t\longrightarrow\infty}\sup_{\pi^{\infty}(\cdot|b_{-1})}\Big\{\sum_{a_0}\Big\{\sum_{b_0}\log\Big(\frac{{\bf P}(b_0|b_{-1},a_0)}{{\bf P}^{\pi^{\infty}}(b_0|b_{-1})}\Big){\bf P}(b_0|b_{-1},a_0)\\
&+\sum_{b_0}\Big(\tilde{V}_{t-1}(b_0)-(t-1)J^*+\frac{1}{t}\tilde{V}_t(b_{-1})-J^*)\Big){\bf P}(b_0|b_{-1},a_0)\Big\}\pi^{\infty}(a_0|b_{-1})\Big\} \label{der_IH}
\end{align}
where (a) is due to  \eqref{inf.DP2}.
Since the  channel input and output alphabet spaces are at most countable,  then we can interchange of the limit and the maximization operations, to obtain  dynamic programming equation (\ref{inf.DP3}).

\section{Proof of Theorem.~\ref{IRR}}\label{IRR_proof}

 For any $\{\pi^{\infty}(a_i|b_{i-1}): \hso i=0, \ldots, n\}$, (\ref{TI_BM_a})  is expressed  as follows. 
\begin{align}
J(\pi^{\infty},\mu)&=\liminf_{n\longrightarrow \infty}\frac{1}{n}{\bf E}_{\mu}^{\pi^{\infty}}\Big\{\sum_{i=0}^{n-1}\overline{\ell}(b_{i-1},a_i)\Big\}=  \liminf_{n\longrightarrow \infty}\frac{1}{n}  {\bf E}_{\mu}^{\pi^{\infty}}\Big\{\sum_{i=0}^{n-1}\ell(b_{i-1},\pi(b_{i-1}))\Big\}, \hso \forall \mu(b_{-1}) \in {\cal M}({\mb B})\label{av.inf.c}\\
&=\liminf_{n\longrightarrow \infty}\frac{1}{n}\mu^T\Big(\sum_{i=0}^{n-1}{\bf P}(\pi^{\infty})^i\Big)\ell(\pi^{\infty}).
\end{align}
Following  \cite{bertsekas05},  it can be shown that the above limit exists but it may depend on the distribution $\mu(\cdot)$ of $B_{-1}$. However, if ${\bf P}(\pi^{\infty})$ is irreducible  then 
\begin{equation}\label{av.inf.irr1}
J(\pi^{\infty}, \mu^{\infty})=\mu^T {\bf P}_1(\pi^{\infty})\ell(\pi^{\infty})=\nu(\pi^{\infty})^T\ell(\pi^{\infty})
\end{equation}
where ${\bf P}_1(\pi^{\infty})$ is the limiting matrix (this follows by the  Cesaro limit), and $\nu(\pi^{\infty})$ is the unique invariant probability distribution, which satisfies  ${\bf P}(\pi^{\infty})\nu(\pi^{\infty})=\nu(\pi^{\infty})$. From \eqref{av.inf.irr1}, it follows that  $J(\pi^{\infty}, \mu) \equiv J(\pi^{\infty})$, that is, it does not depend on the initial distribution $\mu$ of $B_{-1}$. It can be shown that  if for all stationary Markov channel input distributions $\pi^\infty$ the transition matrix ${\bf P}(\pi^\infty)$ is irreducible, there exists a solution $V:\mathbb{B}\mapsto \mathbb{R}^{|\mathbb{B}|}$ and $J\in \mathbb{R}$, which satisfies \eqref{IR_DP_1}.

\section{Proof of Theorem.~\ref{op_in_out_dis_the}}\label{op_in_out_dis_the_proof}
 (a) First, we  employ the necessary and sufficient conditions of Theorem~\ref{nessufco}, to calculate the optimal input and output distributions and the value function at the terminal time. To show the time-invariant property it is sufficient to prove the the value function of the terminal condition, $V_n(b_{n-1})$, is independent of $b_{n-1}$ (part (b) of Theorem \ref{non-nest_the}). By Theorem~\ref{nessufco}, we have
\begin{align}
 V_n(b_{n-1}) =& \sum_{b_{n}}\log\Big(\frac{{\bf P}_n(b_n|a_n,b_{n-1})}{{\bf P}^{\pi}_{n}(b_n|b_{n-1})}\Big){\bf P}(b_n|a_n,b_{n-1}), \hso   \forall{a_n}\in{\mb A}_n \hso \mbox{if} \hso \pi_n(a_n|b_{n-1})\neq{0}\label{pr_gen_suf}
\end{align}
For $b_{n-1}=0 \  \&  \ a_n=0$, we obtain
\beae
 V_n(b_{n-1}=0) &=& \sum_{b_{n}}\log\Big(\frac{{\bf P}_n(b_n|a_n=0,b_{n-1}=0)}{{\bf P}^{\pi}_{n}(b_n|b_{n-1}=0)}\Big){\bf P}(b_n|a_n=0,b_{n-1}=0)\nonumber\\
   &=& \alpha\log\frac{1-{{\bf P}^{\pi}_{n}(b_{n}=0|b_{n-1}=0)}}{{{\bf P}^{\pi}_{n}(b_{n}=0|b_{n-1}=0)}}+\log\frac{1}{1-{{\bf P}^{\pi}_{n}(b_{n}=0|b_{n-1}=0)}}-H(\alpha).
\nonumber\\ \label{pr_ti_bssc}\eeae
For $b_{n-1}=0 \  \&  \ a_n=1$, we obtain
\beae
 V_n(b_{n-1}=0) &=& \sum_{b_{n}}\log\Big(\frac{{\bf P}_n(b_n|a_n=1,b_{n-1}=0)}{{\bf P}^{\pi}_{n}(b_n|b_{n-1}=0)}\Big){\bf P}(b_n|a_n=1,b_{n-1}=0)\nonumber\\
   &=& (1-\beta)\log\frac{1-{{\bf P}^{\pi}_{n}(b_{n}=0|b_{n-1}=0)}}{{{\bf P}^{\pi}_{n}(b_{n}=0|b_{n-1}=0)}}+\log\frac{1}{1-{{\bf P}^{\pi}_{n}(b_{n}=0|b_{n-1}=0)}}-H(\beta).
\nonumber\\ \label{pr_ti_bssc1}\eeae
By (\ref{pr_gen_suf}), we equate (\ref{pr_ti_bssc}) and (\ref{pr_ti_bssc1}), to deduce
\bea
{{\bf P}^{\pi}_{n}(b_{n}=0|b_{n-1}=0)}=\lambda=\frac{1}{1+2^\mu} \label{pr_out_dis_2}
\eea
where $\lambda$ and $\mu$ are given in (\ref{lam_mu_v}). We repeat the above procedure for the pair $b_{n-1}=1,  a_n=0$ and $b_{n-1}=1,  a_n=1$, to deduce
\bea
{{\bf P}^{\pi}_{n}(b_{n}=1|b_{n-1}=1)}=\frac{1}{1+2^\mu}\equiv\lambda. 
\eea
Therefore the optimal transition probability of the output process at time $n$, is given by the doubly stochastic matrix (\ref{bssc_theo_oud}). Next, we show that the value function, $V_n(b_{n-1})$, is independent of $b_{n-1}$. The value function for $b_{n-1}=1$ and $a_n=1$ is obtained as follows.
\beae
 V_n(b_{n-1}=1) &=& \sum_{b_{n}}\log\Big(\frac{{\bf P}_n(b_n|a_n=1,b_{n-1}=1)}{{\bf P}^{\pi}_{n}(b_n|b_{n-1}=1)}\Big){\bf P}(b_n|a_n=1,b_{n-1}=1)\nonumber\\
   &=& \alpha\log\frac{1-{{\bf P}^{\pi}_{n}(b_{n}=1|b_{n-1}=1)}}{{{\bf P}^{\pi}_{n}(b_{n}=1|b_{n-1}=1)}}+\log\frac{1}{1-{{\bf P}^{\pi}_{n}(b_{n}=1|b_{n-1}=1)}}-H(\alpha)\nonumber\\
      &=& \alpha\log\frac{1-{{\bf P}^{\pi}_{n}(b_{n}=0|b_{n-1}=0)}}{{{\bf P}^{\pi}_{n}(b_{n}=0|b_{n-1}=0)}}+\log\frac{1}{1-{{\bf P}^{\pi}_{n}(b_{n}=0|b_{n-1}=0)}}-H(\alpha).
\nonumber\\
&=&  V_n(b_{n-1}=0)\eeae
Since the value function, $V_n(b_{n-1})$, is independent of $b_{n-1}$, we apply Theorem~\ref{non-nest_the}.(b), to deduce that the optimal channel input and channel output conditional distributions are time invariant. The optimal channel input conditional distribution is calculated via the expression ${{\bf P}^{\pi}_{n}(b_{n}|b_{n-1})}=\sum_{A_i}{\bf P}_n(b_n|a_n,b_{n-1})$ $\pi(a_{n}|b_{n{-}1})$. For $b_{n}=0 \ \mbox{and} \ b_{n-1}=0$, we have
\beae
{{\bf P}^{\pi}_{n}(b_{n}=0|b_{n-1}=0)}&=&\sum_{A_n}{\bf P}_n(b_n=0|a_n,b_{n-1}=0)\pi(a_{n}|b_{n{-}1}=0)\nonumber\\
&=&\alpha\pi(a_{n}=0|b_{n{-}1}=0)+(1-\beta)(1-\pi(a_{n}=0|b_{n{-}1}=0)).\nms\label{pr_in_dis_1}
\eeae
Solving (\ref{pr_in_dis_1}) with respect to the input distribution yields
\bea
\pi(a_{n}=0|b_{n{-}1}=0)=\frac{1-(1-\beta)(1+2^\mu)}{(\alpha+\beta-1)(1+2^\mu)}\equiv\nu.\label{pr_in_dis_2}
\eea
Similarly,
\beae
{{\bf P}^{\pi}_{n}(b_{n}=1|b_{n-1}=1)}&=&\sum_{A_n}{\bf P}_n(b_n=1|a_n,b_{n-1}=1)\pi(a_{n}|b_{n{-}1}=1)\nonumber\\
&=&\alpha\pi(a_{n}=1|b_{n{-}1}=0)+(1-\beta)(1-\pi(a_{n}=0|b_{n{-}1}=0)).\nms\label{pr_in_dis_3}
\eeae
The above, shows (\ref{bssc_theo_oid}). By Theorem~\ref{non-nest_the}.(b), specifically \eqref{dyn_TC_NN_un_st_TI} evaluated at $t=0$, we obtain the following expression for the FTFI capacity.
\beae
C_{A^n \rar B^n}^{FB,BSSC} 
&\sr{(\alpha)}{=}&\sum_{b_{-1}}V_0(b_{-1}){\mu}(b_{-1})\nonumber\\
&\sr{(\beta)}{=}&(n+1)\max_{\pi(a_0|b_{-1})}\sum_{b_0,a_0,b_{-1}}\left(\frac{{\bf P}(b_0|a_0,b_{-1})}{{\bf P}^\pi(b_{0}|b_{-1})}\right){\bf P}(b_0|a_0,b_{-1})\pi(a_{n}|b_{n{-}1}){{\mu}}(b_{-1}), \hso b_{-1}\in\{0,1\}\nonumber\\
&\sr{(\gamma)}{=}&(n+1)\left[ H(\lambda){-}\nu H({\alpha}){-}(1{-}\nu)H({\beta}) \right]
 \eeae
where $(\alpha)$ holds by  definition (equation \eqref{conn_val_fu}), $(\beta)$  holds due to \eqref{dyn_TC_NN_un_st_TI} evaluated at $t=0$, $(\gamma)$ by substituting the time invariant capacity achieving input distribution (\ref{bssc_theo_oid}), the corresponding optimal output distribution (\ref{bssc_theo_oud}) and any value of $b_{-1}\in\{0,1\}$.\\
(b) holds by definition (equation \eqref{feed_capa}).


\section{Proof of Theorem.~\ref{cor_cos_fee}}\label{cor_cos_fee_proof}

(a)  By employing the dynamic programming recursion for the constrained problem (\ref{dyn_TC})
we can show that the value function at the terminal time is independent of $b_{n-1}$. Therefore, by Theorem~\ref{non-nest_the}, the optimization problem is non-nested and the dynamic programming for the constrained capacity is given by 
\begin{align}
V_i(b_{i-1})=&\sup_{\pi(a_i|b_{i-1}),s\leq 0}\Big\{\sum_{A_n}\sum_{B_n}\log\Big(\frac{{\bf P}(b_i|b_{i-1},\alpha_i)}{{\bf P}^\pi(b_i|b_{i-1})}\Big){\bf P}(b_i|b_{i-1},\alpha_i)\pi(\alpha_i|b_{i-1})\nonumber\\
&+s\Big\{\sum_{A_i}\gamma(a_i,b_{i-1})\pi_n(a_i|b_{i-1})-\kappa\Big\}\Big\}, \forall i=0,1,\ldots,n\label{DP.eq_con.3a_pr_co}
\end{align}
Differentiating (\ref{DP.eq_con.3a_pr_co}) with respect to the Lagrangian s, we obtain the optimal input distribution of (\ref{con_inp_the}). The optimal output distribution is then calculated by ${{\bf P}^{\pi}_{n}(b_{n}|b_{n-1})}=\sum_{A_n}{\bf P}_n(b_n|a_n,b_{n-1})\pi(a_{n}|b_{n{-}1})$
to obtain (\ref{con_out_the}).\\
(b) Since (i) the optimal input channel conditional distribution and the channel output conditional distribution are time-invariant and (ii) the value function $V_i(b_{i-1})$ is independent of $b_{i-1}, \ \forall \ i=0,1,\ldots,n$, the proof  is identical to the proof of Theorem~\ref{op_in_out_dis_the}.(b). The value of $\kappa_{max}$ is given when the Lagrangian $s=0$, i.e. the constrained optimization problem is equivalent to the constrained optimization problem. In this case, $s=0$, and the optimal channel input conditional distribution for the constrained case is equal to the optimal channel input conditional distribution for the unconstrained case, thus $\kappa|_{s=0}=\kappa_{max}=\nu$.

\bibliographystyle{IEEEtran}
\bibliography{Bibliography}

\begin{thebibliography}{10}
\providecommand{\url}[1]{#1}
\csname url@samestyle\endcsname
\providecommand{\newblock}{\relax}
\providecommand{\bibinfo}[2]{#2}
\providecommand{\BIBentrySTDinterwordspacing}{\spaceskip=0pt\relax}
\providecommand{\BIBentryALTinterwordstretchfactor}{4}
\providecommand{\BIBentryALTinterwordspacing}{\spaceskip=\fontdimen2\font plus
\BIBentryALTinterwordstretchfactor\fontdimen3\font minus
  \fontdimen4\font\relax}
\providecommand{\BIBforeignlanguage}[2]{{%
\expandafter\ifx\csname l@#1\endcsname\relax
\typeout{** WARNING: IEEEtran.bst: No hyphenation pattern has been}%
\typeout{** loaded for the language `#1'. Using the pattern for}%
\typeout{** the default language instead.}%
\else
\language=\csname l@#1\endcsname
\fi
#2}}
\providecommand{\BIBdecl}{\relax}
\BIBdecl

\bibitem{shannon48}
C.~E. Shannon, ``A mathematical theory on communication,'' \emph{Bell System
  Technical Journal}, no.~27, pp. 379--423, October 1948.

\bibitem{cover-thomas2006}
T.~M. Cover and J.~A. Thomas, \emph{Elements of Information Theory (Wiley
  Series in Telecommunications and Signal Processing)}.\hskip 1em plus 0.5em
  minus 0.4em\relax Wiley-Interscience, 2006.

\bibitem{shannon1956}
C.~E. Shannon, ``The zero error capacity of a noisy channel,'' \emph{IRE
  Transactions on Information Theory}, vol.~2, no.~3, pp. 112--124, 1956.

\bibitem{dubrushin1958}
R.~L. Dobrushin, ``Information transmission in channel with feedback,''
  \emph{Theory of Probability and its Applications}, vol.~3, no.~4, pp.
  367--383, 1958.

\bibitem{cover-pombra1989}
T.~M. Cover and S.~Pombra, ``Gaussian feedback capacity,'' \emph{IEEE
  Transactions on Information Theory}, vol.~35, no.~1, pp. 37--43, 1989.

\bibitem{ihara1993}
S.~Ihara, \emph{Information Theory for Continuous Systems}, ser. Series on
  probability and statistics.\hskip 1em plus 0.5em minus 0.4em\relax World
  Scientific, 1993.

\bibitem{marko1973}
H.~Marko, ``The bidirectional communication theory--a generalization of
  information theory,'' \emph{Communications, IEEE Transactions on}, vol.~21,
  no.~12, pp. 1345 -- 1351, dec 1973.

\bibitem{massey1990}
J.~Massey, ``Causality, feedback and directed information,'' \emph{IEEE
  International Symposium on Information Theory and its Applicationss},
  vol.~72, pp. 303--305, November 2001.

\bibitem{kourtellaris2015information}
C.~K. Kourtellaris and C.~D. Charalambous, ``Information structures of capacity
  achieving distributions for feedback channels with memory and transmission
  cost: Stochastic optimal control \& variational equalities-part i,''
  \emph{arXiv preprint arXiv:1512.04514}, 2015.

\bibitem{kourtellarisISIT2016}
------, ``Information structures of capacity achieving distribution for
  channels with memory and feedback,'' in \emph{Information Theory (ISIT), 2016
  IEEE International Symposium on, accepted for publication}, Juny 2016.

\bibitem{alajaji}
N.~Sen, F.~Alajaji, and S.~Yuksel, ``Feedback capacity of a class of symmetric
  finite-state markov channels,'' \emph{IEEE Transactions on Information
  Theory}, vol.~57, no.~7, pp. 4110--4122, July 2011.

\bibitem{permuter08}
H.~Permuter, P.~Cuff, B.~V. Roy, and T.~Weissman, ``Capacity of the trapdoor
  channel with feedback,'' \emph{IEEE Transactions on Information Theory},
  vol.~54, no.~7, pp. 3150--3165, July 2008.

\bibitem{elishco}
O.~Elishco and H.~Permuter, ``Capacity and coding for the ising channel with
  feedback,'' \emph{IEEE Transactions on Information Theory}, vol.~60, no.~9,
  pp. 5138--5149, Sept 2014.

\bibitem{berger_shannon_lecture}
T.~Berger, ``{Living Information Theory},'' \emph{IEEE Information Theory
  Society Newsletter}, vol.~53, no.~1, March 2003.

\bibitem{chen-berger2005}
J.~Chen and T.~Berger, ``The capacity of finite-state markov channels with
  feedback,'' \emph{IEEE Transactions on Information Theory}, vol.~55, no.~6,
  pp. 780--798, 2005.

\bibitem{asnani13}
H.~Asnani, H.~Permuter, and T.~Weissman, ``Capacity of a post channel with and
  without feedback,'' in \emph{Information Theory Proceedings (ISIT), 2013 IEEE
  International Symposium on}, 2013, pp. 2538--2542.

\bibitem{asnani13j}
H.~Permuter, H.~Asnani, and T.~Weissman, ``Capacity of a post channel with and
  without feedback,'' \emph{IEEE Transactions on Information Theory}, vol.~60,
  no.~10, pp. 6041--6057, Oct 2014.

\bibitem{kourtellaris_itw2015}
C.~K. Kourtellaris and C.~D. Charalambous, ``Capacity of binary state symmetric
  channel with and without feedback and transmission cost,'' in
  \emph{Information Theory Workshop (ITW), 2015 IEEE}, April 2015, pp. 1--5.

\bibitem{kcbisit2015}
C.~K. Kourtellaris, C.~D. Charalambous, and J.~J. Boutros, ``Nonanticipative
  transmission for sources and channels with memory,'' in \emph{Information
  Theory (ISIT), 2015 IEEE International Symposium on}, June 2015, pp.
  521--525.

\bibitem{jelinek}
F.~Jel{\'\i}nek, \emph{Probabilistic information theory: discrete and
  memoryless models}, ser. McGraw-Hill series in systems science.\hskip 1em
  plus 0.5em minus 0.4em\relax McGraw-Hill, 1968.

\bibitem{gastpar}
M.~Gastpar, ``To code or not to code,'' Ph.D. dissertation, Ecole Polytechnique
  F{\'e}d{\'e}rale (EPFL), Lausanne, 2002.

\bibitem{tatikonda2000}
S.~Tatikonda, ``Control over communication constraints,'' Ph.D. thesis, M.I.T,
  Cambridge, MA, 2000.

\bibitem{charalambous-stavrou2012}
C.~Charalambous and P.~Stavrou, ``Directed information on abstract spaces:
  Properties and variational equalities,'' \emph{IEEE Transactions on
  Information Theory}, vol.~PP, no.~99, pp. 1--1, 2016.

\bibitem{stavrou2016sequential}
P.~A. Stavrou, C.~D. Charalambous, and C.~K. Kourtellaris, ``Sequential
  necessary and sufficient conditions for optimal channel input distributions
  of channels with memory and feedback,'' in \emph{2016 IEEE International
  Symposium on Information Theory (ISIT)}, 2016, pp. 300--304.

\bibitem{gallager}
R.~G. Gallager, \emph{Information Theory and Reliable Communication}.\hskip 1em
  plus 0.5em minus 0.4em\relax New York, NY, USA: John Wiley \& Sons, Inc.,
  1968.

\bibitem{luenberger1968optimization}
D.~Luenberger, \emph{Optimization by Vector Space Methods}, ser. Professional
  Series.\hskip 1em plus 0.5em minus 0.4em\relax Wiley, 1968.

\bibitem{hernandezlerma-lasserre1996}
O.~Hernandez-Lerma and J.~Lasserre, \emph{Discrete-Time Markov Control
  Processes: Basic Optimality Criteria}, ser. Applications of Mathematics
  Stochastic Modelling and Applied Probability.\hskip 1em plus 0.5em minus
  0.4em\relax Springer Verlag, 1996, no. v. 1.

\bibitem{varayia86}
P.~R. {Kumar} and P.~{Varaiya}, \emph{Stochastic systems: Estimation,
  identification, and adaptive control}.\hskip 1em plus 0.5em minus 0.4em\relax
  Prentice Hall, 1986.

\bibitem{permuter2006}
H.~Permuter, T.~Weissman, and A.~Goldsmith, ``Capacity of finite-state channels
  with time-invariant deterministic feedback,'' in \emph{2006 IEEE
  International Symposium on Information Theory}, {J}uly 2006, pp. 64--68.

\bibitem{Shannon59}
C.~E. Shannon, ``{Coding theorems for a discrete source with a fidelity
  criterion},'' in \emph{IRE Nat. Conv. Rec., Pt. 4}, 1959, pp. 142--163.

\bibitem{bertsekas05}
D.~Bertsekas, \emph{Dynamic programming and optimal control}.\hskip 1em plus
  0.5em minus 0.4em\relax Athena Scientific, 2005.

\end{thebibliography}
\end{document}